\documentclass[sigconf, nonacm]{acmart}

\usepackage{pvldb}% VLDB/PVLDB formatting overrides, see pvldb.sty for details.
\usepackage{mycommands}
\newcommand{\hlc}[2]{%
	{\sethlcolor{#1}\hl{#2}}%
}

\newcommand{\AuxiliarySampleMotivation}{
	
\begin{figure*}[t]
	\centering
	\begin{tikzpicture}[
		vertex/.style={circle,draw,fill=white,inner sep=1.2pt,minimum size=16pt,font=\small},
		edge/.style={line width=1.0pt},
		elab/.style={font=\scriptsize,fill=white,inner sep=1pt},
		]
		
		\newlength{\cellw}
		\setlength{\cellw}{4.4cm} % horizontal spacing between cells
		
		\begin{scope}[xshift=0\cellw]
			\node[vertex,font=\Large] (u1) at (0,0) {$u$};
			\node[vertex,font=\Large] (v1) at (2,0) {$v$};
			\node[vertex,fill=yellow,font=\Large] (w1) at (1,1.8) {$w$};
			
			\draw[edge, black, very thick]         (u1) -- node[elab, below, yshift=-3pt,font=\Large] {3} (v1);
			\draw[edge, green!60!black]        (v1) -- node[elab, right, xshift=3pt,font=\Large] {1} (w1);
			\draw[edge, green!60!black] (w1) -- node[elab, left, xshift=-3pt,font=\Large]  {2} (u1);
		\end{scope}
		
		\begin{scope}[xshift=1\cellw]
			\node[vertex,font=\Large] (u2) at (0,0) {$u$};
			\node[vertex,font=\Large] (v2) at (2,0) {$v$};
			\node[vertex, fill=yellow,font=\Large] (w2) at (1,1.8) {$w$};
			
			\draw[edge, black, very thick]         (w2) -- node[elab, left, xshift=-3pt,font=\Large] {3} (u2);
			\draw[edge, red, dashed]        (u2) -- node[elab, below, yshift=-3pt,font=\Large] {1} (v2);
			\draw[edge, green!60!black] (v2) -- node[elab, right, xshift=3pt,font=\Large]  {2} (w2);
		\end{scope}
		
		\begin{scope}[xshift=2\cellw]
		\node[vertex,font=\Large] (u3) at (0,0) {$u$};
		\node[vertex,font=\Large] (v3) at (2,0) {$v$};
		\node[vertex, fill=yellow,font=\Large] (w3) at (1,1.8) {$w$};
		
		\draw[edge, black, very thick]         (w3) -- node[elab, left, xshift=-3pt,font=\Large] {3} (u3);
		\draw[edge, red, dashed]        (u3) -- node[elab, below, yshift=-3pt,font=\Large] {2} (v3);
		\draw[edge, green!60!black] (v3) -- node[elab, right, xshift=3pt,font=\Large]  {1} (w3);
	\end{scope}
		
	\end{tikzpicture}
	\caption{Figure explaining the motivation behind the usage of the \emph{auxiliary} samples. 
		Edge labels indicate, without loss of generality, the arrival time $t$ in the stream of edges forming the triangle, for $t \in \{1, 2, 3\}$.
		\hl{Yellow} nodes indicate nodes for which we know the degree from the previous pass (i.e., nodes for which incident edges are sampled in the \emph{main} sample). 
		When the third edge arrives ($t = 3$), we check for the triangles closed by such edge in the \hlc{green!60!black}{main} sample. 
		(Left) Estimating triangles only with edges in the \hlc{green!60!black}{main} sample suffices; (Center, Right) Estimating triangles only with edges in the \hlc{green!60!black}{main} samples \emph{does not} suffice: we need an \hlc{red}{auxiliary} sample of edges. 
		}
	\label{fig:aux_sample}
\end{figure*}
}

\newcommand\PrintFigureIntro[2]{
	\begin{figure*}[t]
		\centering		
		\begin{subfigure}{\textwidth}
			\centering
			\includegraphics[width=\textwidth]{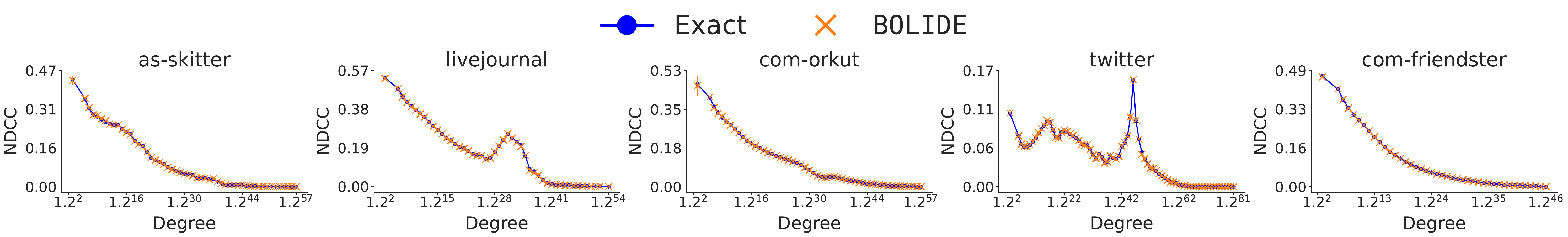}
			\caption{Exact and \algname's estimated node-averaged binned degree-wise clustering coefficient (NDCC) distribution, for degree intervals of powers of 1.2 (i.e., $D_i = [1.2^i, 1.2^{i+1})$).}
			\label{fig:intro_figure_top}
		\end{subfigure}
		\begin{subfigure}{\textwidth}
			\centering
			\includegraphics[width=\textwidth]{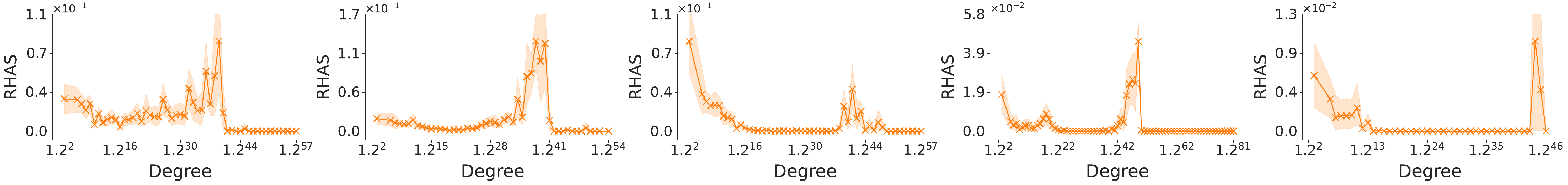}
			\caption{Pointwise RHAS distance between the exact and estimated distribution.}
			\label{fig:intro_figure_bottom}
		\end{subfigure}
		\caption{Results for node-averaged binned degree-wise clustering coefficient (NDCC) distribution, for degree intervals of powers of 1.2 (i.e., $D_i = [1.2^i, 1.2^{i+1})$). For each dataset, the average and standard deviation over 10 independent repetitions are shown.}
		\label{fig:compact_#1_base#2}
	\end{figure*}
}

\newcommand\PrintFigure[2]{
	\begin{figure*}[t]
		\centering
		\includegraphics[width=\textwidth]{res/revision_figures/distros/bolide_#1_bin_size#2.pdf}
		\caption{Exact and estimated distributions for #1 (top) and pointwise RHAS distance between them (bottom), for degree intervals of powers of #2 (i.e., $D_i = [ #2^i, #2^{i+1})$). For each dataset, the average and standard deviation over 10 independent repetitions are shown. 
		Dashed red line indicates the (average) degree threshold $\tau$, separating head and tail estimators of \algname. }
		\label{fig:compact_#1_base#2}
	\end{figure*}
}

\newcommand{\PrintFigureBaslines}{
	\begin{figure*}[t]
		\centering
		\includegraphics[width=\textwidth]{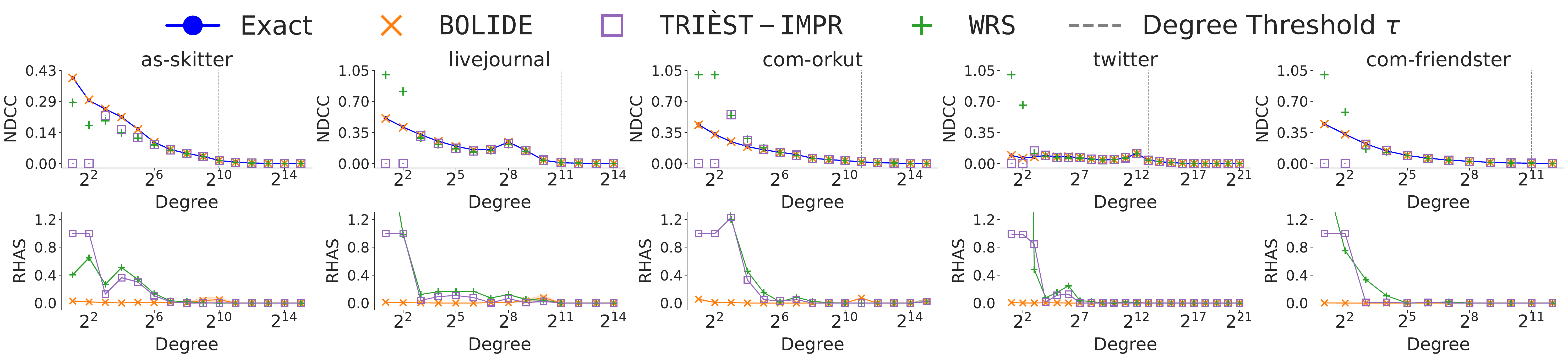}
		\caption{Exact and estimated distributions for NDCC (top) and pointwise RHAS distance between them (bottom), for degree intervals of powers of 2 (i.e., $D_i = [ 2^i, 2^{i+1})$). For each dataset, the average and standard deviation over 10 independent repetitions are shown. 
			Dashed red line indicates the (average) degree threshold $\tau$, separating head and tail estimators of \algname. }
		\label{fig:compact_baselines}
	\end{figure*}
}

\newcommand{\FigureLCC}{
	\begin{figure*}[htbp]
		\centering
		\includegraphics[width=\textwidth]{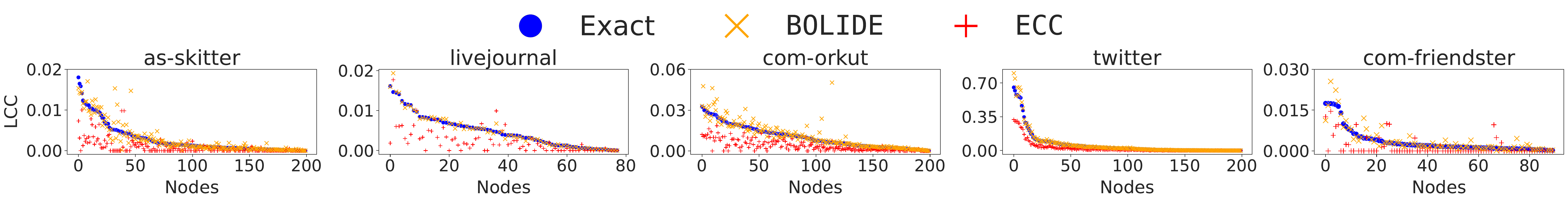}
		\caption{Exact and estimated Local Clustering Coefficients (LCC). For the sake of visualization, we report $\min(200, |H|)$ LCC for nodes sampled uniformly at random from $H$, for each dataset. 
			We display nodes in the $x$-axis sorted in decreasing order of their true LCC value. 
			For \algname\ and \kpalg, the average over 10 independent repetitions is shown for each estimated LCC.}
		\label{fig:lcc_comparison}
	\end{figure*}
}

\newcommand{\GridSearchParams}{
	\begin{figure*}[t]
		\centering
		\includegraphics[width=\textwidth]{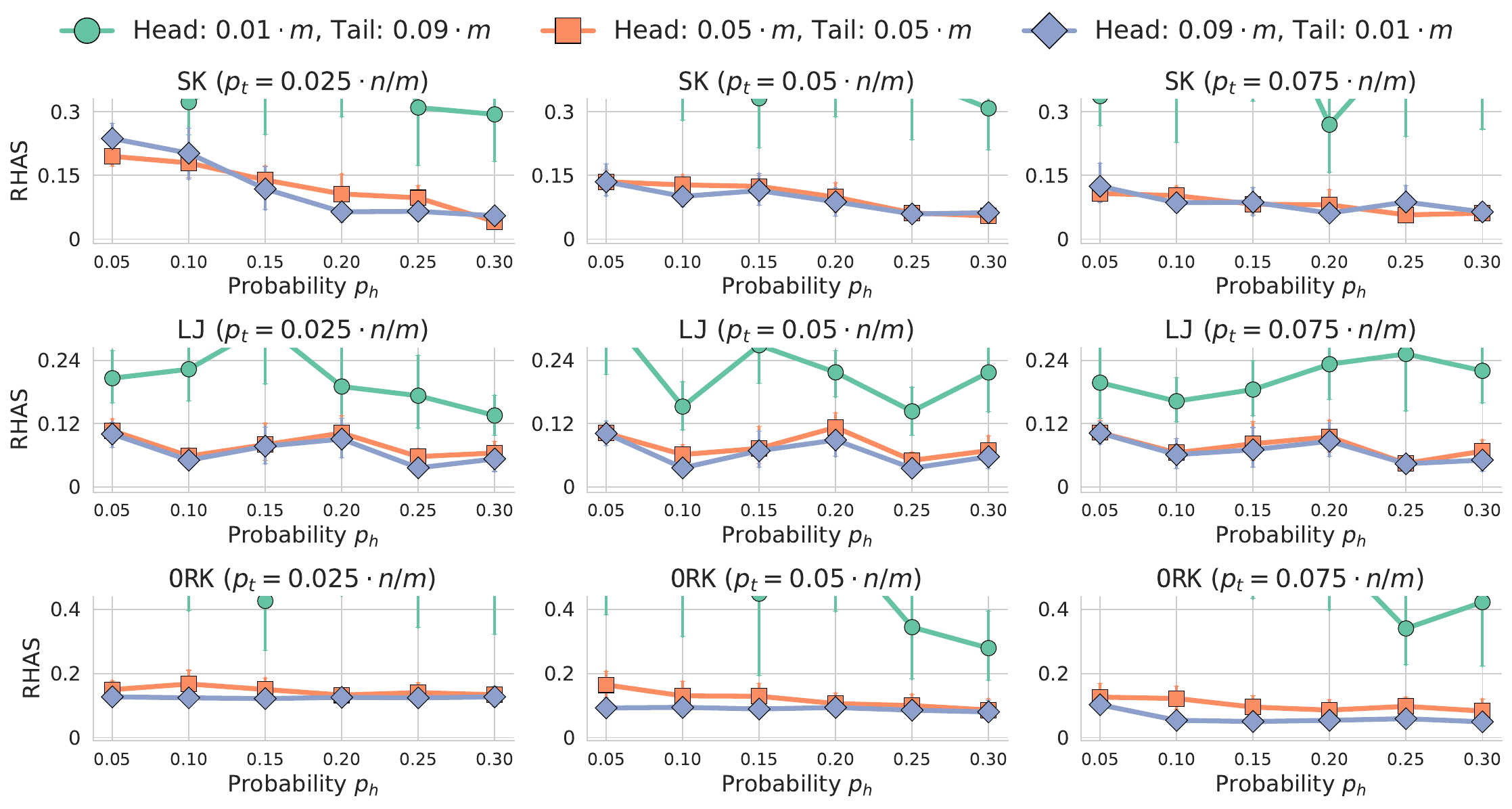}
		\caption{RHAS distance between exact and estimated NDCC distribution, for degree intervals of powers of two (i.e., $D = [2^i, 2^{i+1}]$) versus probability $p_h$ of sampling head nodes. For each row (dataset), each subplot shows a different value of probability $p_t$ of sampling tail nodes. For each dataset and choice of parameters, the average and 95\% confidence intervals over 10 independent repetitions are shown. $y$-axis is shared across rows (datasets).}
		\label{fig:grid_search}
	\end{figure*}
	
}

\newcommand{\EmpiricalBoundsTail}{
	\begin{figure}[t]
		\centering
		\includegraphics[width=0.9\columnwidth]{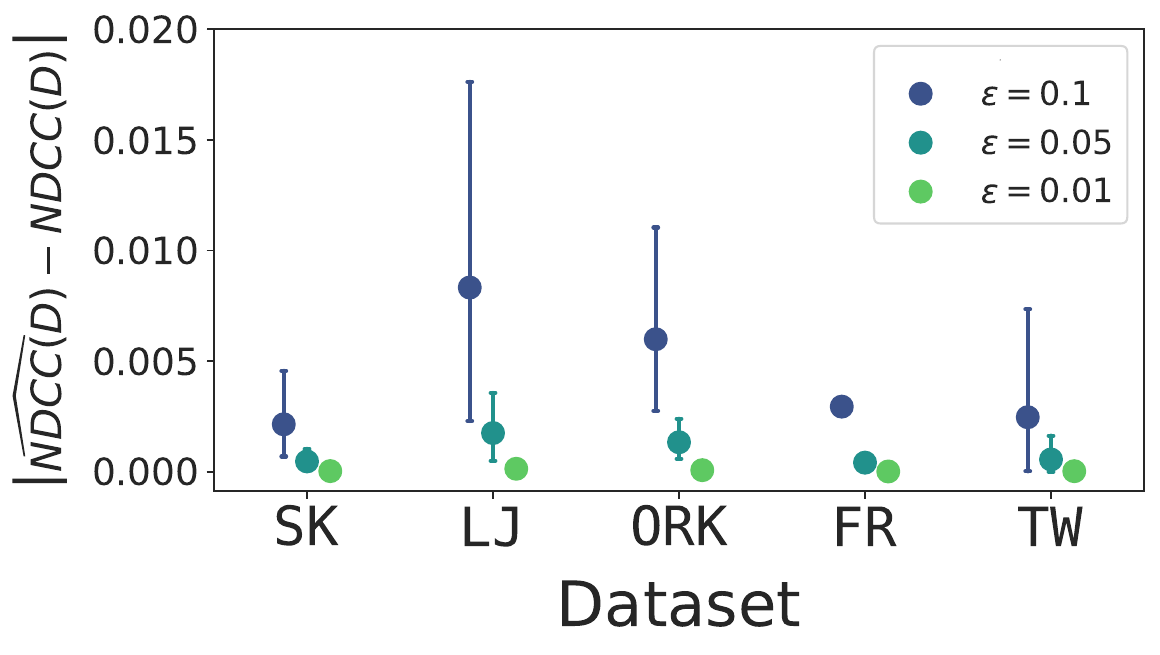}
		\caption{Absolute error of estimates given by bound in Eq.~\ref{eq:ndcc_tail_interval}, for various values of $\varepsilon$. 
		For each dataset and $\varepsilon$, the average and standard deviation over 20 different intervals $D = [2^i, 2^{i+1})$, such that $\tau \le 2^i$, are shown.
		}
		\label{fig:empirical_bounds_tail}
	\end{figure}
}

\newcommand{\FigureECC}{
	\begin{figure*}[t]
		\centering
		\includegraphics[width=\textwidth]{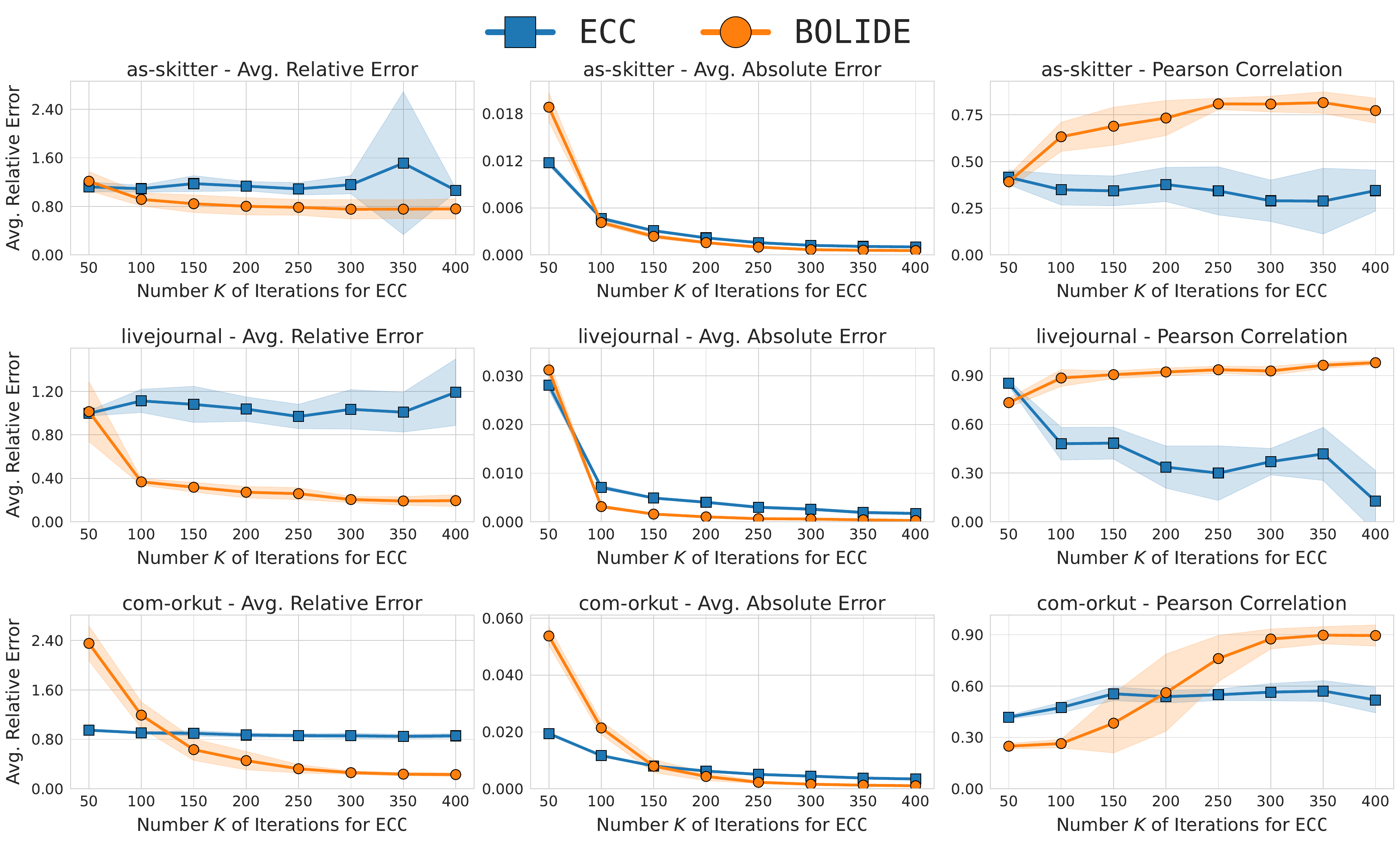}
		\caption{Comparison between \kpalg\ and \algname. For each dataset (row), we report Average Relative Error, Average Absolute Error and Pearson Correlation coefficient over columns, by varying the number $K$ of iterations for \kpalg. For each dataset and metric, the average and standard deviation over 10 independent repetitions are shown. }
		\label{fig:ecc_comparison}
	\end{figure*}
}

\newcommand{\FigurePerf}{
	\begin{figure*}[t]
		\centering
		\includegraphics[width=\textwidth]{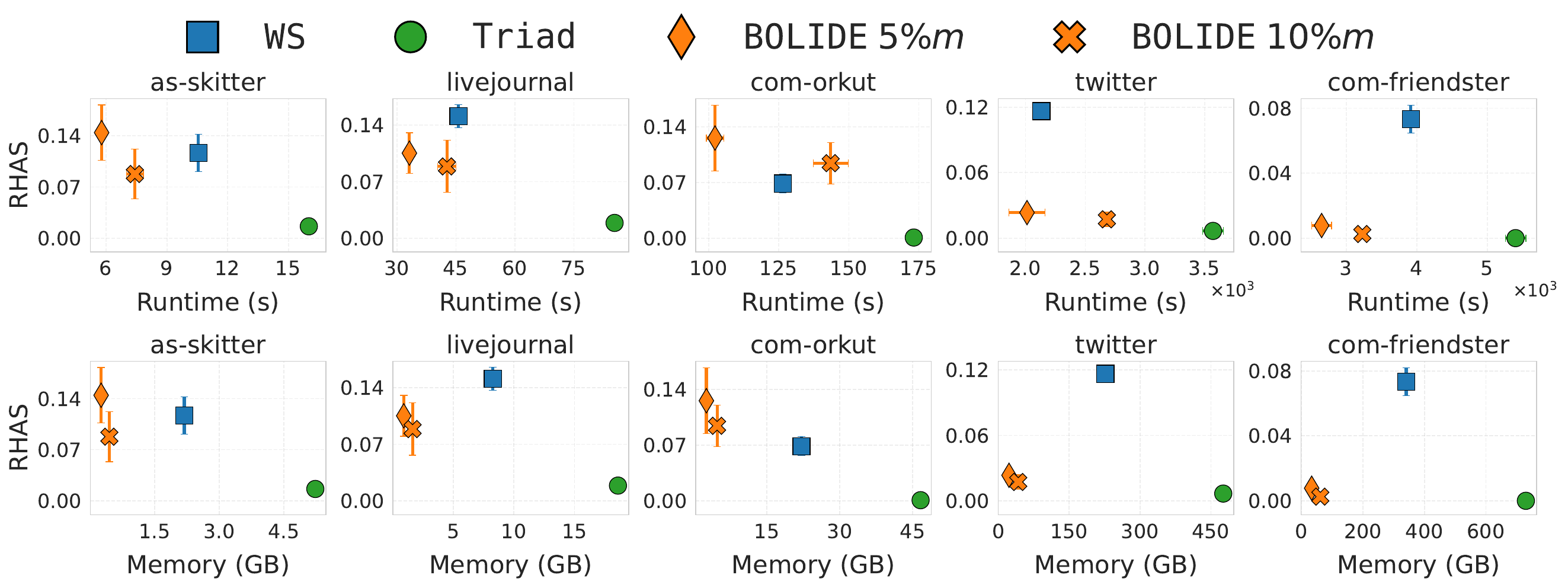}
		\caption{Comparison among \wedgesampling, \triadalg\ and \algname. For each algorithm and for each dataset (column), we report on the y-axis RHAS between exact and approximated NDCC distributions for power-two degree intervals (i.e., $D = [2^i, 2^{i + 1})$), and on the x-axis the runtime (first row) and the memory usage (second row). For all metrics, the average and 90\% confidence intervals over 10 independent repetitions are reported. }
		\label{fig:perf}
	\end{figure*}
}

\newcommand{\TableOvercount}{
	\begin{table}[t]
		\caption{Number $c = W / N$ of extra counters stored by our one-pass algorithm, where $N, N'$ are the number of nodes for which we estimate their degree, and the number of nodes for which we estimate triangles, and $W = \max(0, N' - N)$ is the number of ``wasted'' counters. 
			For each dataset and choice of auxiliary sample budget, the average and standard deviation over 10 independent repetitions are reported.}
		\label{tab:overcounts}
		\centering
		\begin{tabular}{lccc}
			& \multicolumn{3}{c}{Tail Auxiliary Sample Budget $B_{A_t}$} \\
			{Dataset} & {$0.005m$} & {$0.01m$} & {$0.015m$} \\
			\hline
			\texttt{SK}  & $0.00 \pm 0.00$ & $0.22 \pm 0.04$ & $0.55 \pm 0.03$ \\
			\texttt{LJ}  & $0.00 \pm 0.00$ & $0.00 \pm 0.00$ & $0.28 \pm 0.00$ \\
			\texttt{ORK} & $0.81 \pm 0.00$ & $1.58 \pm 0.01$ & $2.24 \pm 0.01$ \\
			\texttt{TW}  & $0.40 \pm 0.01$ & $0.73 \pm 0.01$ & $1.03 \pm 0.01$ \\
			\texttt{FR}  & $0.11 \pm 0.00$ & $0.53 \pm 0.00$ & $0.92 \pm 0.00$ \\
		\end{tabular}
	\end{table}
}

\newcommand{\TableShuffle}{
	\begin{table}[t]
		\caption{RHAS distance between exact and estimated NDCC distribution, for degree intervals of powers of 2 (i.e., $D_i = [2^i, 2^{i+1}]$), for different orderings of the stream. For each dataset and stream ordering, the average and standard deviation over 10 independent repetitions are reported.}
		\label{tab:shuffle}
		\centering
		\begin{tabular}{l cccc}
			& \multicolumn{4}{c}{Input Stream Ordering} \\
			\cmidrule(lr){2-5}
			Dataset & \texttt{Rand} & \texttt{Adj-Incr} & \texttt{Adj-Decr} & \texttt{Adj-Rand} \\
			\midrule
			\texttt{SK}  & $0.064 {{\scriptstyle \pm 0.046}}$ & $0.072 {{\scriptstyle \pm 0.056}}$ & $0.067 {{\scriptstyle \pm 0.040}}$ & $0.085 {{\scriptstyle \pm 0.064}}$ \\
			\texttt{LJ}  & $0.086 {{\scriptstyle \pm 0.054}}$ & $0.076 {{\scriptstyle \pm 0.052}}$ & $0.090 {{\scriptstyle \pm 0.057}}$ & $0.087 {{\scriptstyle \pm 0.049}}$ \\
			\texttt{ORK} & $0.115 {{\scriptstyle \pm 0.080}}$ & $0.086 {{\scriptstyle \pm 0.040}}$ & $0.100 {{\scriptstyle \pm 0.056}}$ & $0.143 {{\scriptstyle \pm 0.073}}$ \\
		\end{tabular}
	\end{table}
}

\newcommand{\TableKPComparison}{
	\begin{table*}[t]
		\caption{Comparison between \kpalg\ and \algname. For each dataset, we report: minimum degree $d$ of nodes estimated by \kpalg; degree threshold $\tau$ for \algname; number $|H|$ of nodes having degree $\ge \max\left(d, \tau\right)$ and estimated by both \kpalg\ and \algname; average relative error, average absolute error and Pearson correlation coefficient for both algorithms. For each dataset, the average and standard deviation over 10 independent repetitions are reported. 
		We highlight in bold the algorithm achieving the best performance for each considered metric.}
		\label{tab:lcc_comparison}
		\centering
		\setlength{\tabcolsep}{6pt}
		\renewcommand{\arraystretch}{1.2}
		\resizebox{\textwidth}{!}{
			\begin{tabular}{c c c c cc cc cc}
				\multirow{2}{*}{Dataset}
				& \multirow{2}{*}{$d$}
				& \multirow{2}{*}{$\tau$}
				& \multirow{2}{*}{$|H|$}
				& \multicolumn{2}{c}{Avg. Relative Error $\downarrow$}
				& \multicolumn{2}{c}{Avg. Absolute Error $\downarrow$}
				& \multicolumn{2}{c}{Pearson Corr. $\uparrow$} \\
				\cmidrule(lr){5-6} \cmidrule(lr){7-8} \cmidrule(lr){9-10}
				& & & & \kpalg & \algname & \kpalg & \algname & \kpalg & \algname \\
				\midrule
				\texttt{SK}  & $3109 {\scriptstyle \pm 49}$ & $311 {\scriptstyle \pm 49}$   & $233 {\scriptstyle \pm 4}$   & $1.13 {\scriptstyle \pm 0.07}$ & $\mathbf{0.80} {\scriptstyle \pm 0.13}$ & $(2.2 {\scriptstyle \pm 0.1}) \!\times\! 10^{-3}$ & $(\mathbf{1.6} {\scriptstyle \pm 0.1}) \!\times\! 10^{-3}$ & $0.38 {\scriptstyle \pm 0.09}$            & $\mathbf{0.73} {\scriptstyle \pm 0.09}$ \\
				\texttt{LJ}  & $3259 {\scriptstyle \pm 139}$ & $574 {\scriptstyle \pm 48}$   &  $72 {\scriptstyle \pm 5}$   & $1.03 {\scriptstyle \pm 0.11}$ & $\mathbf{0.27} {\scriptstyle \pm 0.05}$ & $(4.0 {\scriptstyle \pm 0.3}) \!\times\! 10^{-3}$ & $(\mathbf{1.0} {\scriptstyle \pm 0.1}) \!\times\! 10^{-3}$ & $0.34 {\scriptstyle \pm 0.12}$            & $\mathbf{0.92} {\scriptstyle \pm 0.02}$ \\
				\texttt{ORK} & $3082 {\scriptstyle \pm 71}$ & $926 {\scriptstyle \pm 29}$   & $493 {\scriptstyle \pm 4}$   & $0.87 {\scriptstyle \pm 0.03}$ & $\mathbf{0.45} {\scriptstyle \pm 0.14}$ & $(6.2 {\scriptstyle \pm 0.2}) \!\times\! 10^{-3}$ & $(\mathbf{4.3} {\scriptstyle \pm 1.4}) \!\times\! 10^{-3}$ & $0.54 {\scriptstyle \pm 0.03}$  & $\mathbf{0.56} {\scriptstyle \pm 0.21}$ \\
				\texttt{TW}  & $2867 {\scriptstyle \pm 56}$ & $3366 {\scriptstyle \pm 195}$ & $22k {\scriptstyle \pm 433}$ & $0.79 {\scriptstyle \pm 0.01}$ & $\mathbf{0.40} {\scriptstyle \pm 0.01}$ & $(3.1 {\scriptstyle \pm 0.1})  \!\times\!  10^{-2}$ & $(\mathbf{1.4} {\scriptstyle \pm 0.1}) \!\times\! 10^{-2}$ & $\mathbf{0.96} {\scriptstyle \pm 0.01}$ & $0.93 {\scriptstyle \pm 0.01}$          \\
				\texttt{FR}  & $2636 {\scriptstyle \pm 92}$ & $2996 {\scriptstyle \pm 4}$   &  $85 {\scriptstyle \pm 16}$   & $1.13 {\scriptstyle \pm 0.15}$ & $\mathbf{0.79} {\scriptstyle \pm 0.23}$ & $(3.5 {\scriptstyle \pm 0.5 }) \!\times\! 10^{-3}$ & $(\mathbf{1.7} {\scriptstyle \pm 0.3}) \!\times\! 10^{-3}$ & $0.48 {\scriptstyle \pm 0.22}$            & $\mathbf{0.75} {\scriptstyle \pm 0.11}$ \\
		\end{tabular}}
	\end{table*}
}

\newcommand{\TableHeadAblation}{
	\begin{table*}[t]
		\caption{Performance of \algname\ considering configurations \whead\ vs \wohead. For each dataset and configuration, we report average relative error, average absolute error and Pearson correlation coefficient, together with running time and peak memory consumption. Average and standard deviation over 10 independent repetitions are reported.}
		\label{tab:head_ablation}
		\centering
		\setlength{\tabcolsep}{5pt}
		\renewcommand{\arraystretch}{1.2}
		\resizebox{\textwidth}{!}{
			\begin{tabular}{c cc cc cc cc cc}
				\multirow{2}{*}{Dataset}
				& \multicolumn{2}{c}{Avg. Relative Error $\downarrow$}
				& \multicolumn{2}{c}{Avg. Absolute Error $\downarrow$}
				& \multicolumn{2}{c}{Pearson Corr. $\uparrow$}
				& \multicolumn{2}{c}{Running Time (s) $\downarrow$}
				& \multicolumn{2}{c}{Memory (GB) $\downarrow$} \\
				\cmidrule(lr){2-3} \cmidrule(lr){4-5} \cmidrule(lr){6-7} \cmidrule(lr){8-9} \cmidrule(lr){10-11}
				& \whead & \wohead
				& \whead & \wohead
				& \whead & \wohead
				& \whead & \wohead
				& \whead & \wohead \\
				\midrule
				\texttt{SK}  & $0.80 {\scriptstyle \pm 0.13}$ & $0.72 {\scriptstyle \pm 0.06}$ & $(1.6 {\scriptstyle \pm 0.1}) \!\times\! 10^{-3}$ & $(1.6 {\scriptstyle \pm 0.2}) \!\times\! 10^{-3}$ & $0.73 {\scriptstyle \pm 0.09}$ & $0.75 {\scriptstyle \pm 0.07}$ & $7.4 {\scriptstyle \pm 0.7}$   & $3.4 {\scriptstyle \pm 0.3}$  & $0.44$ & $0.07$ \\
				\texttt{LJ}  & $0.27 {\scriptstyle \pm 0.05}$ & $0.31 {\scriptstyle \pm 0.10}$ & $(1.0 {\scriptstyle \pm 0.1}) \!\times\! 10^{-3}$ & $(1.1 {\scriptstyle \pm 0.1}) \!\times\! 10^{-3}$ & $0.92 {\scriptstyle \pm 0.02}$ & $0.91 {\scriptstyle \pm 0.02}$ & $42.7 {\scriptstyle \pm 3.8}$  & $17.2 {\scriptstyle \pm 2.7}$ & $1.64$ & $0.24$ \\
				\texttt{ORK} & $0.45 {\scriptstyle \pm 0.14}$ & $0.46 {\scriptstyle \pm 0.14}$ & $(4.3 {\scriptstyle \pm 1.4}) \!\times\! 10^{-3}$ & $(4.4 {\scriptstyle \pm 1.4}) \!\times\! 10^{-3}$ & $0.56 {\scriptstyle \pm 0.21}$ & $0.54 {\scriptstyle \pm 0.24}$ & $143.5 {\scriptstyle \pm 10.2}$ & $54.9 {\scriptstyle \pm 4.9}$ & $4.76$ & $0.54$ \\
				\texttt{TW}  & $0.40 {\scriptstyle \pm 0.01}$ & $0.40 {\scriptstyle \pm 0.004}$ & $(1.4 {\scriptstyle \pm 0.1}) \!\times\! 10^{-2}$ & $(1.4 {\scriptstyle \pm 0.02}) \!\times\! 10^{-2}$ & $0.93 {\scriptstyle \pm 0.01}$  & $0.93 {\scriptstyle \pm 0.003}$ & $2680 {\scriptstyle \pm 21}$  & $852 {\scriptstyle \pm 46}$  & $42.24$ & $4.78$ \\
				\texttt{FR}  & $0.79 {\scriptstyle \pm 0.23}$ & $0.91 {\scriptstyle \pm 0.96}$ & $(1.7 {\scriptstyle \pm 0.3}) \!\times\! 10^{-3}$  & $(1.7 {\scriptstyle \pm 1.0}) \!\times\! 10^{-3}$  & $0.75 {\scriptstyle \pm 0.11}$  & $0.76 {\scriptstyle \pm 0.34}$  & $3226 {\scriptstyle \pm 42}$  & $1335 {\scriptstyle \pm 51}$ & $61.77$ & $8.05$ \\
		\end{tabular}}
	\end{table*}
}

\newcommand{\TableDelta}{
\begin{table*}[t]
	\centering
	\caption{Table reporting, for each degree interval $D = [2^i, 2^{i+1})$, the quantities $C(D)$ (number of nodes with degree in $D$), $\Triangles{D}{}$ (number of triangles incident to nodes with degree in $D$), $\Delta_V(D)$ (maximum number of triangles incident to a node with degree in $D$), $\Triangles{D^-}{}$ and $\Delta_E(D^+)$ (maximum number of triangles incident to an edge containing a node with degree in $D^+$). We set $\varepsilon = 0.05$ for computing intervals $D^- = [L(1 + \varepsilon), U(1 - \varepsilon))$ and $D^+ = [L(1 - \varepsilon), U(1 + \varepsilon) )$.
	}
	\label{tab:delta_ratio_perturbed}
	\setlength{\tabcolsep}{4pt}
	\resizebox{0.9\textwidth}{!}{
		\begin{tabular}{@{}l r r r r r r r r r r r r r r r@{}}
			$D$ & \multicolumn{5}{c}{{as-skitter}} & \multicolumn{5}{c}{{livejournal}} & \multicolumn{5}{c}{{com-orkut}} \\
			\cmidrule(lr){1-1} \cmidrule(lr){2-6} \cmidrule(lr){7-11} \cmidrule(lr){12-16}
			$[2^i, 2^{i+1})$
			& $C(D)$ & $\Triangles{D}{}$ & $\Triangles{D^-}{}$ & $\Delta_V(D)$ & $\Delta_E(D^+)$
			& $C(D)$ & $\Triangles{D}{}$ & $\Triangles{D^-}{}$ & $\Delta_V(D)$ & $\Delta_E(D^+)$
			& $C(D)$ & $\Triangles{D}{}$ & $\Triangles{D^-}{}$ & $\Delta_V(D)$ & $\Delta_E(D^+)$ \\
			\midrule
			
			$[2^{1}, 2^{2})$
			& 447k & 291k & 194k & 3 & 1.4k
			& 919k & 756k & 501k & 3 & 218
			& 81.6k & 65.8k & 46.2k & 3 & 102 \\
			
			$[2^{2}, 2^{3})$
			& 456k & 1.4M & 1.2M & 21 & 8.5k
			& 834k & 3.5M & 2.9M & 21 & 188
			& 148k & 595k & 517.7k & 21 & 117 \\
			
			$[2^{3}, 2^{4})$
			& 319k & 4.0M & 3.6M & 105 & 9.8k
			& 743k & 11M & 10.1M & 105 & 223
			& 300k & 4.3M & 4.0M & 105 & 493 \\
			
			$[2^{4}, 2^{5})$
			& 150k & 6.7M & 5.9M & 429 & 14k
			& 617k & 27.7M & 25M & 465 & 452
			& 555k & 25.0M & 22.1M & 433 & 1.7k \\
			
			$[2^{5}, 2^{6})$
			& 69k & 8.8M & 7.7M & 1.5k & 18k
			& 404k & 51.5M & 45.7M & 2.0k & 872
			& 772k & 100M & 87.1M & 1.6k & 3.2k \\
			
			$[2^{6}, 2^{7})$
			& 21.8k & 7.0M & 6.0M & 4.3k & 952
			& 185k & 70.9M & 62.8M & 7.9k & 1.3k
			& 685k & 247M & 221M & 5.3k & 3.0k \\
			
			$[2^{7}, 2^{8})$
			& 9.4k & 7.7M & 6.7M & 10.1k & 524
			& 62.3k & 110M & 97.8M & 31.8k & 1.2k
			& 333k & 326M & 298M & 16.0k & 2.6k \\
			
			$[2^{8}, 2^{9})$
			& 3.7k & 8.4M & 7.6M & 18.3k & 520
			& 19.4k & 158M & 144M & 114k & 715
			& 96.2k & 239M & 210M & 30.2k & 2.2k \\
			
			$[2^{9}, 2^{10})$
			& 1.4k & 9.0M & 7.8M & 34.3k & 1.0k
			& 4.9k & 91.8M & 81.7M & 214k & 1.0k
			& 29.7k & 199M & 186M & 57.9k & 2.4k \\
			
			$[2^{10}, 2^{11})$
			& 536 & 5.7M & 4.7M & 69.3k & 2.1k
			& 704 & 18.4M & 15.1M & 229k & 1.4k
			& 3.0k & 75.4M & 65.5M & 118k & 1.7k \\
			
			$[2^{11}, 2^{12})$
			& 211 & 4.5M & 4.2M & 107k & 4.1k
			& 161 & 4.6M & 3.9M & 123k & 2.5k
			& 735 & 48.8M & 41.7M & 266k & 3.1k \\
			
			$[2^{12}, 2^{13})$
			& 85 & 2.7M & 2.0M & 154k & 8.2k
			& 29 & 1.8M & 1.7M & 123k & 4.3k
			& 250 & 34.9M & 31.8M & 591k & 5.9k \\
			
			$[2^{13}, 2^{14})$
			& 60 & 3.6M & 3.2M & 283k & 16.2k
			& 11 & 1.3M & 1.2M & 368k & 4.5k
			& 99 & 28.1M & 24.8M & 1.0M & 9.1k \\
			
			$[2^{14}, 2^{15})$
			& 24 & 3.5M & 3.3M & 531k & 28.7k
			& 1 & 59k & 0 & 59k & 840
			& 27 & 16.8M & 16.3M & 1.7M & 9.1k \\

		\end{tabular}}
\end{table*}	
}

\renewcommand\vldbdoi{XX.XX/XXX.XX}
\renewcommand\vldbpages{XXX-XXX}
\renewcommand\vldbavailabilityurl{https://github.com/CristianBold4/BOLIDE}

\begin{document}
	
	\title{Single-Pass Estimation of the Clustering Coefficient Distribution in Graph Streams}
	
	%%
	%% The "author" command and its associated commands are used to define the authors and their affiliations.
	\author{Cristian Boldrin}
	%\authornote{Work done in part while visiting University of California, Santa Cruz.}
	%\authornote{Both authors contributed equally to this research.}
	%\orcid{1234-5678-9012}
	%\author{Cristian Boldrin}
	%\authornotemark[1]
	%\email{webmaster@marysville-ohio.com}
	\affiliation{%
		\institution{University of Padova}
		\city{Padova}
		\country{Italy}
	}
	\email{boldrincri@dei.unipd.it}

	\author{C. Seshadhri}
	\affiliation{%
		\institution{University of California, Santa Cruz}
		\city{Santa Cruz}
		\state{California}
		\country{USA}}
	\email{sesh@ucsc.edu}
	
	%%
	%% The abstract is a short summary of the work to be presented in the
	%% article.
	\begin{abstract}
		Triangle counting is one of the most fundamental problems in network analysis. 
		Given the massive sizes of real-world graphs, there is a long history of small-space streaming algorithms providing accurate estimates for this problem. 
		However, most of the results focus on estimating the total triangle count or the number of triangles incident to individual nodes. 
		In practice, one often wants fine-grained information to understand how triangles are distributed, as captured by \emph{clustering coefficients}. 
		In particular, a standard network analysis task requires computing the binned degree-wise clustering coefficient distribution, which provides a rich and informative summary of the structure of the graph.
		
		In this work we present \algname, the first efficient and practical algorithm for estimating binned degree-wise clustering coefficients in streaming. Our algorithm makes a \emph{single pass} over the edge stream, and is allowed to store only a \emph{small fraction} of the total number of edges.
		\algname\ carefully combines different sampling strategies to efficiently gather degree and triangle information across sets of nodes. 
		As a result, our algorithm \emph{provably} approximates the binned degree-wise clustering coefficients, and provides guarantees on the amount of memory used. 
		Our experimental evaluation shows that \algname\ accurately estimates clustering coefficient distributions while efficiently processing large datasets with \emph{billions} of edges and triangles.
		
	\end{abstract}
	
	\maketitle
	
	%%% do not modify the following VLDB block %%
	%%% VLDB block start %%%
	\pagestyle{\vldbpagestyle}
	\begingroup\small\noindent\raggedright\textbf{PVLDB Reference Format:}\\
	\vldbauthors. \vldbtitle. PVLDB, \vldbvolume(\vldbissue): \vldbpages, \vldbyear.\\
	\href{https://doi.org/\vldbdoi}{doi:\vldbdoi}
	\endgroup
	\begingroup
	\renewcommand\thefootnote{}\footnote{\noindent
		This work is licensed under the Creative Commons BY-NC-ND 4.0 International License. Visit \url{https://creativecommons.org/licenses/by-nc-nd/4.0/} to view a copy of this license. For any use beyond those covered by this license, obtain permission by emailing \href{mailto:info@vldb.org}{info@vldb.org}. Copyright is held by the owner/author(s). Publication rights licensed to the VLDB Endowment. \\
		\raggedright Proceedings of the VLDB Endowment, Vol. \vldbvolume, No. \vldbissue\ %
		ISSN 2150-8097. \\
		\href{https://doi.org/\vldbdoi}{doi:\vldbdoi} \\
	}\addtocounter{footnote}{-1}\endgroup
	%%% VLDB block end %%%
	
	%%% do not modify the following VLDB block %%
	%%% VLDB block start %%%
	\ifdefempty{\vldbavailabilityurl}{}{
		\vspace{.3cm}
		\begingroup\small\noindent\raggedright\textbf{PVLDB Artifact Availability:}\\
		The source code, data, and/or other artifacts have been made available at \url{\vldbavailabilityurl}.
		\endgroup
	}
	%%% VLDB block end %%%
	
	\section{Introduction} \label{sec:intro}
	
	Graphs are a ubiquitous representation of relationships among objects in real-world problems. 
	In many applications, understanding structural properties of the network is essential for downstream analyses.
	Among these properties, one of the most interesting is the \emph{local clustering coefficient}~\cite{watts1998collective}, which measures the fraction of connected pairs among the neighbors of a node.
	In other words, the local clustering coefficient of a node is the ratio between the number of triangles incident to the node and the number of triangles it could potentially participate in, given its degree.
	In several applications~\cite{zhang2017efficient,seshadhri2013fast,seshadhri2013triadic}, rather than considering local clustering coefficients for any node in the graph, the quantity of interest is the average clustering coefficient of nodes within a specific degree bin. 
	To this end, one can define the \emph{binned degree-wise clustering coefficient} (see Section~\ref{sec:preliminaries} for a formal definition).
	This measure reveals information about the triangles structure across the degree bins, and is a central tool in analyzing and modeling graphs~\cite{kaiser2008mean,sala2010measurement,ugander2011anatomy,seshadhri2012community,li2017clustering}, database benchmarking~\cite{ciglan2012benchmarking}, social networks~\cite{hardiman2013estimating}, graph embeddings~\cite{bianchi2020spectral}, and link prediction~\cite{WU20161}.
	For instance, binned degree-wise clustering coefficients are directly used to compare synthetic graphs
	to real ones~\cite{sala2010measurement,ugander2011anatomy}, and to fit or assess real-world graph generators~\cite{seshadhri2012community,pfeiffer2012fast,pfeiffer2014attributed}.

	Computing clustering coefficients typically requires triangle enumeration, a fundamental primitive in the domain of network science. 
	However, exact triangle computation becomes infeasible on modern graphs, with common public datasets having edges and triangles in the order of billions.
	For this reason, it is necessary to resort to efficient algorithms that operate under \emph{limited space} and provide \emph{high-quality} approximations.
	
	One of the most important models for handling massive data is the \emph{streaming} model. 
	The input graph is observed as a stream of edges arriving in arbitrary order, and the algorithm must process the stream by storing only a small fraction of the edges of the stream (e.g., via sampling).
	
	Approximating triangle counts in graph streams has been extensively studied, to the point of becoming a research area by its own.
	Since its introduction in~\cite{bar2002reductions}, the problem has received sustained attention for over two decades~\cite{jowhari2005new,buriol2007estimating,tsourakakis2009doulion,pavan2013counting,ahmed2014graph,mcgregor2016better,bera2017towards,kallaugher2017hybrid,stefani2017triest,turkoglu2017edge,turk2019revisiting,shin2017wrs}.
	Despite such a rich history, most results focus on estimating the total triangle count or triangles incident to individual nodes, while surprisingly little attention has been dedicated to estimating clustering coefficients. 
	This gap motivates our studied problem: given an input graph stream, can we design an efficient algorithm for provably computing high-quality estimates of binned degree-wise clustering coefficients?
	
	\PrintFigureIntro{NDCC}{1.2}
	
	\emph{Problem definition.} 
	We provide a brief description of the problem we consider, deferring technical details to Section~\ref{sec:preliminaries}. 
	We consider an input graph observed as an insertion-only edge stream, i.e., a sequence of edges arriving in arbitrary order. 
	%A streaming algorithm has no a priori knowledge of the graph; at each step, the algorithm receives the arriving edge, and must take irrevocably a decision.
	A streaming algorithm has no a priori knowledge of the graph, and observes edges one at a time. Each incoming edge must be processed irrevocably, and past edges are no longer accessible.
	The algorithm works with a limited amount of memory, which in practice we parameterize by a small fraction of the total number of edges (i.e., the stream length).
	We consider \emph{single-pass} algorithms, meaning that the stream can be processed only once.
	
	We study Node-averaged and Wedge-averaged binned Degree-wise Clustering Coefficient (henceforth, $\NodeAvgCC{D}{}$ and $\WedgeAvgCC{D}{}$ respectively), for any given degree \emph{interval} $D$.
	Conventionally, these measures define distributions $\{ \NodeAvgCC{D_i}{} \}_i$ and $\{ \WedgeAvgCC{D_i}{} \}_i$ when evaluated over a collection of degree intervals $\{ D_1, D_2, \dots \}$.
	Our aim is to design an efficient one-pass streaming algorithm to provably compute high-quality estimates of the binned degree-wise clustering coefficients distributions. 
	
	% -- our contributions
	\subsection{Our contributions} 
	\label{sec:contrib}
	We introduce algorithm \algname\ (\emph{B}inned degree \emph{O}ne-pass c\emph{L}ustering coef\-fic\emph{I}ents \emph{D}istribution \emph{E}stimation). 
	\begin{myitemize}
		\item \algname\ is the \emph{first} one-pass streaming algorithm for estimating binned degree-wise clustering coefficients for arbitrary degree intervals. 
		Our approach combines multiple sampling strategies for degree and triangle estimations, to derive a \emph{head} estimator (for any interval $D$ of low degrees) and a \emph{tail} estimator (for any interval $D$ of high degrees) for $\NodeAvgCC{D}{}$ and $\WedgeAvgCC{D}{}$.
		By carefully combining these estimators, \algname\ efficiently provides accurate approximations under arbitrary edge stream ordering.
		\item We analyze \algname\ and provide theoretical guarantees on the quality of estimates for the binned degree-wise clustering coefficients, as well as on the space used by our algorithm. We also show that \algname\ implicitly stores high-quality estimates of \emph{local clustering coefficients} for sufficiently high-degree nodes.
		\item We conduct an extensive experimental evaluation on very large graph streams, with up to $1.8B$ edges and $34.8B$ triangles. 
		Figure~\ref{fig:intro_figure_top} shows an example of our results. 
		The estimates from \algname\ are highly accurate across all datasets, with only a few outliers: the trends in clustering coefficients across degree bins are captured nearly perfectly, despite only storing 10\% of the edge stream. 
		\item We compare with the algorithm from Kutzkov and Pagh~\cite{kutzkov2013streaming}, which is (to the best of our knowledge) the only previous one-pass streaming algorithm that estimates local clustering coefficients for sufficiently high-degree nodes, showing that \algname\ achieves consistently lower errors.
		Moreover, we compare \algname\ with related state-of-the-art sampling approaches~\cite{seshadhri2013fast,sarpe2025efficient}, demonstrating scalability and efficiency in both runtime and memory usage, while providing high-quality estimates. 
		
	\end{myitemize}
	
	\section{Preliminaries}
	\label{sec:preliminaries}
	
	%Let $G = (V, E)$ be an undirected graph with no self-loops and no multiple edges, where $V$ and $E$ denote the set of nodes and the set of edges, respectively, with $|V| = n$ and $|E| = m$.
	Let $G = (V, E)$ be a simple and undirected graph, where $V$ and $E$ denote the set of nodes and the set of edges, respectively, with $|V| = n$ and $|E| = m$. Edge insertions are observed in arbitrary order through the graph stream $\Sigma = \{ e^{(1)}, \dots, e^{(m)}\}$. Each edge $e^{(t)} = \{u, v\} \in E$ is an unordered pair of nodes.
	%Henceforth, we omit the temporal superscript as we consider quantities at the end of the stream. 
	
	We say that $\Delta = \{ u, v, w \}$ is a triangle in $G$ if distinct edges $\{u, v\}, \{v, w\}$, and $\{u, w\}$ all appear in $E$. We also say that a triangle $\Delta$ is incident to $v$ if $v \in \Delta$; in this case, we refer to $\Delta$ as a local triangle for $v$.
	For a node $v$, let $\Deg{v}{}$ denote its degree, and let $\LocalTrianglesCard{v}{}$ be the number of triangles incident to $v$.
	The local clustering coefficient (LCC) of $v$ is defined as $\LocalCC{v}{} = \LocalTrianglesCard{v}{}/{\Deg{v}{} \choose 2}$, and corresponds to the ratio of triangles to the number of wedges ($2$-paths)
	centered at $v$. 
	
	Let $\DegreeSet{}$ be the set of degrees of the graph $G^{}$, i.e., $\DegreeSet{} = \{ d \colon  \exists v \in V \textit{ such that } \Deg{v}{} = d, \text{ with } d > 1\}$. 
	We ignore nodes with degree zero and one, as they cannot participate in triangles. 
	For a fixed set $D \subseteq \DegreeSet{}$ of degrees, we define $\InducedNodes{D}{} = \left\{ v \;|\; \Deg{v}{} \in D\right\}$. 
	Let $\NumNodes{D}{} = \CardSet{\InducedNodes{D}{}}$, and let $\Triangles{D}{}$ be the number of triangles incident to nodes having degree in $D$, i.e., $\Triangles{D}{} = \sum_{v \in V_D} \LocalTrianglesCard{v}{}$.
	We further denote with $\Delta_V(D)$ the maximum number of triangles incident to a node in $\InducedNodes{D}{}$, and with $\Delta_E(D)$ the maximum number of triangles (among those incident to nodes in $\InducedNodes{D}{}$) adjacent to an edge.
	We are now ready to introduce the fundamental binned degree-wise clustering coefficient, common in previous works~\cite{zhang2017efficient,seshadhri2013fast,seshadhri2013triadic}, and central in our paper. 
	
	\begin{definition}[{Node-averaged binned Degree-wise Clustering Coefficient}]
		\label{def:degree_cc_node}
		Let $D \subseteq \DegreeSet{}$ be a set of degrees of the graph. The node-average binned degree-wise clustering coefficient (NDCC) with respect to $D$ is the average local clustering coefficient on all nodes in $\InducedNodes{D}{}$, that is:
		\begin{equation*}
			\NodeAvgCC{D}{} = \frac{1}{\left| \InducedNodes{D}{} \right|} \sum_{v \in \InducedNodes{D}{}} \LocalCC{v}{}.
		\end{equation*}
	\end{definition}
	
	Similarly, we define the binned degree-wise clustering coefficient averaged over wedges. 
	\begin{definition}[{Wedge-averaged binned Degree-wise Clustering Coefficient}]
		\label{def:degree_cc_wedge}
		Let $D \subseteq \DegreeSet{}$ be a set of degrees of the graph. The wedge-average binned degree-wise clustering coefficient (WDCC) with respect to $D$ is the average (over wedges) clustering coefficient on all nodes in $V_D$, that is:
		\begin{equation*}
			\WedgeAvgCC{D}{} = \frac{\sum_{v \in \InducedNodes{D}{}} \LocalTrianglesCard{v}{}}{ \sum_{v \in \InducedNodes{D}{}} {\Deg{v}{} \choose 2}}.
		\end{equation*}
	\end{definition}
	
	Notice that, if the set $D$ contains all graph degrees (i.e., $D = \DegreeSet{}$), then $\WedgeAvgCC{D}{}$ corresponds to the global clustering coefficient, also known as transitivity of the graph.
	Throughout the manuscript, we restrict the set $D$ to meaningful and contiguous \emph{intervals}, i.e., $D = [L, U) \cap \mathbb{Z}^+$ for some $L < U$ with $U = \Theta(L)$.
	
	Henceforth, ``with high probability'' means probability at least $1 - 1 / n^c$, for some constant $c \ge 1$. We use $\bigOmegatilde{n}$ to hide poly-logarithmic factors in $n$, i.e., $\bigOmega{n \cdot \log^c n}$.

	\section{Related Work}
	\label{sec:related_work}

	Since the introduction of the clustering coefficients~\cite{watts1998collective}, a substantial line of work has investigated both its algorithmic computation and its variants. Prior work has considered the average local clustering coefficient over nodes~\cite{buriol2007estimating,schank2005approximating}, over wedges~\cite{etemadi2017bias}, as well as weighted generalizations~\cite{lattanzi2016efficient}.
	
	When the input graph is fully available, sampling-based methods are among the most widely used approaches for estimating binned degree-wise clustering coefficients. Two prominent families are edge-based sampling methods~\cite{de2022estimating,sarpe2025efficient} and wedge sampling methods~\cite{seshadhri2013triadic,etemadi2017bias,seshadhri2014wedge}.
	
	Within the one-pass streaming setting, which is the focus of this paper, only~\cite{kutzkov2013streaming} (to the best of our knowledge) has investigated a variant of our studied problem. Specifically, the algorithm proposed in~\cite{kutzkov2013streaming} is restricted to estimating the local clustering coefficients of sufficiently high-degree vertices, and cannot adapt to the entire degree distribution.
	To address this limitation, our approach to estimating binned degree-wise clustering coefficients in one-pass combines degree estimation with approximate triangle counting. 
	
	% -- ccdh
	The problem of degree estimation under restricted resources has been studied through the task of approximating the complementary cumulative degree histogram (CCDH), which asks for each degree $d$ the number of vertices of degree at least $d$. 
	This problem has been investigated in both sublinear query models~\cite{eden2018provable,bishnu2025towards} and streaming~\cite{simpson2015catching,bishnu2025towards}. 
	Our degree estimator is essentially equivalent to \texttt{headtail}~\cite{simpson2015catching}, a streaming CCDH estimator which combines different sampling mechanisms to obtain accurate estimates across head and tail regions of the degree distribution.
	
	% -- triangle counting
	Triangle counting in streaming has been intensely studied over the last few decades~\cite{bar2002reductions,jowhari2005new,buriol2007estimating,tsourakakis2009doulion,pagh2012colorful,pavan2013counting,kutzkov2013streaming,JhSePi13,lim2015mascot,mcgregor2016better,bera2017towards,kallaugher2017hybrid,stefani2017triest,shin2017wrs,shin2018think,jayaram2021optimal,chen2022triangle,boldrin2024fast}. 
	In single-pass streams, many algorithms have focused on sampling-based triangle estimation, especially via edge sampling and reservoir-based methods~\cite{ahmed2017sampling,stefani2017triest,shin2017wrs,shin2018think,boldrin2024fast,chen2022triangle,lim2015mascot}. These methods have been studied from both theoretical and practical perspectives, achieving high-quality estimates of the global and local triangle counts on real-world datasets.
	However, all existing streaming algorithms are designed to estimate local or global triangle counts, rather than binned degree-wise clustering coefficients. Moreover, since they operate by uniformly sampling edges (and consequently nodes are kept with probability proportional to their degrees), they induce a degree bias toward higher-degree nodes. 
	As a result, adapting them to obtain accurate estimates over arbitrary degree intervals, especially intervals containing only low-degree vertices, is nontrivial in one-pass streams and requires additional correction. This is precisely the challenge we address in our paper.

	\section{Two-Pass Algorithm}
	\label{sec:two_pass_algorithm}
	For ease of presentation, we first introduce a two-pass streaming algorithm. 
	Approximating the binned degree-wise clustering coefficients essentially boils down to estimating degrees and triangle counts, and carefully combining these estimates. 
	Hence, the first pass focuses on sampling nodes to estimate their degrees, while the second pass estimates the number of triangles incident to the sampled nodes only.
	Finally, in Section~\ref{sec:one_pass_algo}, we show how to ``merge'' the two passes into a single pass without losing guarantees on the quality of the returned estimates.

	% first pass
	\begin{algorithm}[b]
		\caption{\algnamefirstpass $\left( \Sigma, p_h, p_t\right)$}
		\label{alg:first_pass}
		\LinesNumbered
		\kwInput{Graph stream $\Sigma$; \\
			Probability $p_h$ of sampling head nodes; \\
			Probability $p_t$ of sampling tail nodes.}
		\kwOutput{Estimate of degree for sampled head and tail nodes.}
		Initialize empty maps $S_h, S_t \gets \emptyset$\label{line:fp_init-first}\;
		%Init. edge sample $E_t \gets \emptyset$\label{line:fp_init-second}\;
		\For{\textup{each incoming edge} $e = \{u, v\}$ \textup{in the stream} $\Sigma$}{\label{line:fp_for-stream}
			%\If{$u \notin S_h$}{ \label{line:fp_head_first}
				%	\lIf{$hash[u] < p_h$}{$S_h[u] \gets 1$} \label{line:fp_head_last}
				%}
			%\lElse{{$S_h[u] \gets S_h[u] + 1$\label{line:fp_head_else}}}
			%Repeat the above steps~\ref{line:fp_head_first}-\ref{line:fp_head_else} also for $v$\label{line_fp-repeat}\;
			%$E_t \gets E_t \cup \{ \{ u, v\} \}$ with probability $p_t$\label{line:fp_tail}\;
			\If{$u \notin S_h$}{\label{line:fp_head_first}
				\lIf{$hash[u] < p_h$}{$S_h[u] \gets 1$\label{line:fp_head_last}}}
			\lElse{$S_h[u] \gets S_h[u] + 1$}\label{line:fp_head_else}
			\lIf{$u \notin S_t$}{set $S_t[u] \gets 1$ with probability $p_t$} \label{line:fp_tail_first}
			\lElse{$S_t[u] \gets S_t[u] + 1$} \label{line:fp_tail_last}
			Repeat the above steps~\ref{line:fp_head_first}-\ref{line:fp_tail_last} also for $v$\;\label{line:fp_repeat}
			
		}
		Let $\ell(r) = \lceil \frac{1 - p_t - (1 - p_t)^{r+1} - r p_t (1 - p_t)^r}{p_t (1 - (1 - p_t)^r)} \rceil$\;\label{line:fp_deg_correction_first}
		Find $r$ s.t. $S_t[v] = r - \ell(r)$ and set $S_t[v] \gets r$, for each $v \in S_t$\; \label{line:fp_deg_correction_last}
		\Return head node-degrees $S_h$, tail (estimated) node-degrees $S_t$\label{line:fp_return}\;
	\end{algorithm}
	
	\AuxiliarySampleMotivation
	
	\emph{First Pass.}
	The first pass (Algorithm~\ref{alg:first_pass}) is inherited from~\cite{simpson2015catching}, whose focus is estimating complementary cumulative degree histogram in graph streams. We borrow the same techniques for estimating degrees of sampled nodes.
	The high-level idea is to maintain a uniform random sample of nodes via hashing (for low-degree nodes, \emph{head} of degree distribution), and a separate sample of nodes picked with probability proportional to their degrees (for high-degree nodes, \emph{tail} of degree distribution).
	
	\algnamefirstpass\ (Alg.~\ref{alg:first_pass}) takes as input the graph stream $\Sigma$ and probabilities $p_h, p_t$ of sampling head and tail nodes, respectively. 
	We assume the existence of a hash function $hash : V \to [0,1]$ that maps nodes uniformly to the interval $[0,1]$.
	First, maps $S_h, S_t$ for storing degrees of (sampled) head and tail nodes are initialized (line~\ref{line:fp_init-first}). 
	For each incoming edge $\Edge{u}{v}$, if node $u \notin S_h$ then \algnamefirstpass\ starts tracking exactly its degree if $hash[u] < p_h$; otherwise, it increments its tracked degree (line~\ref{line:fp_head_first}-\ref{line:fp_head_else}).
	(The probability that $hash[u] < p_h$ is exactly $p_h$.) 
	Next, if $u \notin S_t$ we set $S_t[u] = 1$ with probability $p_t$. Otherwise, if $u$ is already in $S_t$, we increment its degree. 
	Concretely, once we sample a node, we start tracking its degree in $S_t$. This procedure is also referred to as ``sample-and-hold''. 
	The same steps are then applied also for node $v$.
	
	At the end of the stream, degree estimates for (sampled) tail nodes are computed (lines~\ref{line:fp_deg_correction_first}-\ref{line:fp_deg_correction_last}).
	Given that degrees of nodes in $S_t$ follow a (truncated) geometric distribution, we need to consider the expected loss (denoted as $\ell(r)$ for degree $r$) of such a distribution to correct estimates. 
	Intuitively, for each node in $S_t$, we need to account for a correction term that quantifies how late we started tracking its degree during the stream. 
	(See~\cite{simpson2015catching} for further details.)
	
	As we show in Section~\ref{sec:additional_proofs} of the Appendix, it turns out that the corrected degrees are approximately close to the true degrees, i.e., $S_t[v] \approx \Deg{v}{}$ for any ``high-degree'' node $v$, with high probability. (``High-degree'' nodes are defined by the choice of $p_t$.)
	
	% second pass
	\emph{Second Pass.}
	The second pass (Algorithm~\ref{alg:second_pass}) estimates triangles via edge sampling with fixed probability (common in prior work~\cite{lim2015mascot,shin2018think}).
	For our purpose, we are interested only in counting triangles incident to nodes tracked in the first pass (i.e., nodes in $S_h, S_t$, returned by Alg.~\ref{alg:first_pass}), for which we know the degree (or estimated degree).
	To this end, we maintain separate samples dedicated to incoming edges incident to the nodes in $S_h$ and $S_t$, which we denote as head and tail \emph{main} samples, respectively. 
	For each incoming edge $e = \Edge{u}{v}$, triangle counting is performed by checking the number of wedges closed by $e$ in the main samples (i.e., a pair of edges that form a triangle with $e$). 
	Next, if we have either $u$ or $v$ contained in $S_h \cup S_t$, then we independently sample $e$ in the respective head or tail main sample; otherwise, we simply ignore edge $e$. 
	The estimates of triangles incident to nodes in $S_h \cup S_t$ can then be rescaled by an appropriate factor to obtain unbiased estimates.
	
	Unfortunately, this strategy fails. 
	Depending on the arrival order of the triangle edges, we cannot guarantee unbiased estimates using \emph{only} the main samples. 
	For instance, consider a triangle $\{u, v, w\}$ with $w \in S_h \cup S_t$ and $u, v \notin S_h \cup S_t$.
	If $\Edge{u}{v}$ is the first (or second) edge incoming in the stream (relative to edges $\Edge{u}{v}, \Edge{v}{w}, \Edge{w}{u}$ forming the triangle), then we would not include it in our main sample, as neither $u$ nor $v$ belongs to $S_h \cup S_t$. 
	Consequently, we lose the triangle $\{u, v, w\}$, inducing a systematic undercount in estimating triangles incident to node $w$. 
	For this reason, we introduce \emph{auxiliary} samples to store edges with no endpoints in $S_h \cup S_t$.
	Intuitively, the auxiliary samples help catch the triangles missed by the main samples alone, guaranteeing unbiased estimates.
	A visualization is given in Figure~\ref{fig:aux_sample}.

	\begin{algorithm}[b]
		\caption{\algnamesecondpass $\left( \Sigma, S_h, S_t, p_{M_h}, p_{M_t}, p_{A_h}, p_{A_t} \right)$}
		\label{alg:second_pass}
		\LinesNumbered
		%	\DontPrintSemicolon
		\kwInput{Graph stream $\Sigma$; \\
			Head node-degree map $S_h$ (from previous pass); \\
			Tail node-degree map $S_t$ (from previous pass); \\
			Prob. $p_{M_h}$ for sampling edges incident to head nodes; \\
			Prob. $p_{M_t}$ for sampling edges incident to tail nodes; \\
			Prob. $p_{A_h}$ for sampling edges \emph{not} incident to head nodes; \\
			Prob. $p_{A_t}$ for sampling edges \emph{not} incident to tail nodes;
		}
		\kwOutput{Estimate of triangles incident to nodes in $S_h$ and $S_t$.}
		%\tcp{Initialization}
		Initialize edge sets $M_h, M_t, A_h, A_t \gets \emptyset$\;\label{line:sp_init_first}
		Initialize node-triangle maps $\hat{T}_h, \hat{T}_t \gets \emptyset$\;\label{line:sp_init_last}
		%Initialize number $h$ of head edges, number $t$ of tail edges: $h, t \longleftarrow 0$\;
		%nitialize number $a_h$ of auxiliary head edges, number $a_t$ of auxiliary tail edges: $a_h, a_t \longleftarrow 0$\;
		\For{\textup{each incoming edge} $e = \{u, v\}$ \textup{in the stream} $\Sigma$}{\label{line:sp_for-stream}
			\texttt{CountTriangles}($\{u, v\}, M_h, p_{M_h}, A_h, p_{A_h}, \hat{T}_h$)\label{line:sp_ct_head}\;
			\texttt{CountTriangles}($\{u, v\}, M_t, p_{M_t}, A_t, p_{A_t}, \hat{T}_t$)\label{line:sp_ct_tail}\;
			\lIf{$u \in S_h$ or $v \in S_h$}{\texttt{SampleEdge}($\{u, v\}, M_h, p_{M_h}$)\label{line:sp_se_h}}
			\lElse{\texttt{SampleEdge}($\{u, v\}, A_h, p_{A_h}$)\label{line:sp_se_ah}}
			\lIf{$u \in S_t$ or $v \in S_t$}{\texttt{SampleEdge}($\{u, v\}, M_t, p_{M_t}$)\label{line:sp_se_t}}
			\lElse{\texttt{SampleEdge}($\{u, v\}, A_t, p_{A_t}$)\label{line:sp_se_at}}
		}
		\Return head triangle estimates $\hat{T}_h$, tail triangle estimates $\hat{T}_t$.\label{line:sp_return}
	\end{algorithm}
	
	\begin{algorithm}[t]
		\caption{\texttt{CountTriangles}$ \left( \{u, v\}, M, p_M, A, p_A, \hat{T} \right)$}
		\label{alg:count_triangles}
		\LinesNumbered
		\kwInput{
			edge $\{u, v\}$; \\
			main sample $M$; prob.\ $p_M$ for sampling edges in $M$; \\
			auxiliary sample $A$; prob.\ $p_A$ for sampling edges in $A$; \\
			Node-triangle map $\hat{T}$.
		}
		%\tcp{Find all neighbors of u and v across both samples}
		Let $Adj_M(u), Adj_M(v)$ be the adjacency lists of  $u, v$ in sample $M$\;
		Let $Adj_A(u), Adj_A(v)$ be the adjacency lists of  $u, v$ in sample $A$\;
		Let $\widehat{N}(u) \gets Adj_M(u) \cup Adj_A(u)$\;
		Let $\widehat{N}(v) \gets Adj_M(v) \cup Adj_A(v)$\;
		% \tcp{Iterate only through valid triangles}
		\For{each node $w \in \widehat{N}(u) \cap \widehat{N}(v)$}{
			\uIf{$w \in Adj_M(u)$ and $w \in Adj_M(v)$}{
				Initialize $\hat{T}[u]$, $\hat{T}[v]$, $\hat{T}[w]$ to zero if not set yet\;
				Increment $\hat{T}[u]$, $\hat{T}[v]$, $\hat{T}[w]$ by $1 / p_M^2$\;
			}
			\ElseIf{$w \in Adj_M(u)$ or $w \in Adj_M(v)$}{
				Initialize $\hat{T}[u]$, $\hat{T}[v]$, $\hat{T}[w]$ to zero if not set yet\;
				Increment $\hat{T}[u]$, $\hat{T}[v]$, $\hat{T}[w]$ by $1 / \left( p_M \cdot p_A \right)$\;
			}
		}
	\end{algorithm}

	We now describe \algnamesecondpass\ (Algorithm~\ref{alg:second_pass}). 
	The algorithm takes as input the graph stream $\Sigma$, the degree maps $S_h, S_t$ of nodes sampled in the first pass (Alg.~\ref{alg:first_pass}), the probabilities $p_{M_h}, p_{M_t}$ for sampling edges in the head and tail main samples, and the probabilities $p_{A_h}, p_{A_t}$ for sampling edges in the head and tail auxiliary samples. 
	First, main samples $M_h$ and $M_t$, auxiliary samples $A_h$ and $A_t$, and node-triangles maps $\hat{T}_h$ and  $\hat{T}_t$ are initialized (lines~\ref{line:sp_init_first}-\ref{line:sp_init_last}).
	For each incoming edge $e = \Edge{u}{v}$ in the stream, \algnamesecondpass\ counts each occurrence of triangles containing $\{u, v\}$ (lines~\ref{line:sp_ct_head},~\ref{line:sp_ct_tail}).  
	Then, if the incoming edge has at least one endpoint in the set $S_h$ of head nodes, it is included in the main sample $M_h$ with probability $p_{M_h}$ (line~\ref{line:sp_se_h}); otherwise, it is included in the auxiliary sample $A_h$ with probability $p_{A_h}$ (line~\ref{line:sp_se_ah}). An analogous procedure is then repeated with respect to the set $S_t$ of tail nodes (lines~\ref{line:sp_se_t}-\ref{line:sp_se_at}). 
	Finally, triangle estimates $\hat{T}_h, \hat{T}_t$ (if any) for nodes in $S_h$ and $S_t$ are returned (line~\ref{line:sp_return}). 
	
	\algnamesecondpass\ makes use of two subroutines, \texttt{CountTriangles} and \texttt{SampleEdge}.
	\texttt{CountTriangles} (Algorithm~\ref{alg:count_triangles}) counts a triangle closed by the current edge in the stream if \emph{at least one} of its two edges (when the closing one arrives) is in the \emph{main} sample. The probability that such a triangle has been sampled is used to update the relevant counts; this ensures that the final estimate is unbiased. This step is crucial to avoid the problems of Figure~\ref{fig:aux_sample}.
	%By doing that, only triangles incident to nodes for which we know their degrees (or estimated degrees) are counted, allowing for the downstream estimation of clustering coefficients. 
	\texttt{SampleEdge} (see Section~\ref{sec:additional_algo} of the Appendix for pseudocode) adds the given edge independently to the given sample with the given probability.

	\emph{Estimates.}
	By maintaining independence across head and tail (main and auxiliary) samples, our algorithm can be treated as computing two different estimators for clustering coefficients. 
	The head estimator uses triangle estimates in $\hat{T}_h$ and degrees in $S_h$, while the tail estimator uses $\hat{T}_t$ and $S_t$.
	We now describe algorithm \algnameestimate\ (Algorithm~\ref{alg:estimate}), which computes estimates of NDCC and WDCC after the two passes.
	The algorithm takes as input the degree interval $D = [L, U)$, the probability $p_t$, node-degree maps $S_h, S_t$ from the first pass (Alg.~\ref{alg:first_pass}), and node-triangles map $\hat{T}_h, \hat{T}_t$ from the second pass (Alg.~\ref{alg:second_pass}).
	We assume $\hat{T}_h[u] = 0$ for node $u \in S_h, u \notin \hat{T}_h$, and equivalently for $\hat{T}_t$.
	First, \algnameestimate\ computes the degree threshold $\tau$ for distinguishing the use of the head and tail estimator (line~\ref{line:tau}). For simplicity, if $L < \tau < U$, then we set $\tau = U$. 
	If $\tau \ge U$, then estimates are computed by using the head estimator (lines~\ref{line:head_first}-\ref{line:head_last}); otherwise (i.e., $\tau < U$, and $\tau \le L$ from the choice of $\tau$) by using the tail estimator (lines~\ref{line:tail_first}-\ref{line:tail_last}). At the end, the algorithm returns estimates $\widehat{NDCC}(D)$ and $\widehat{WDCC}(D)$ for degree interval $D$ (line~\ref{line:return}). 
	
	\begin{algorithm}[h]
		\caption{\texttt{Estimate}$\left( D, p_t, S_h, \hat{T}_h, S_t, \hat{T}_t \right)$}
		\label{alg:estimate}
		\LinesNumbered
		\kwInput{Degree interval $D = [L, U)$; \\ %Prob. $p_h$; Prob. $p_t$\\
			Probability $p_t$ of sampling tail nodes; \\
			Node-degree map $S_h$ with $\Deg{u}{}$ for each $u \in S_h$;\\
			Node-triangles map $\hat{T}_h$ with $\hat{T}_h[u]$ for each $u \in S_h$;\\
			Node-degree map $S_t$ with $\EstDeg{v}{}$ for each $v \in S_t$;\\
			Node-triangles map $\hat{T}_t$ with $\hat{T}_t[v]$ for each $v \in S_t$.
		}
		Set $\tau = \frac{1}{\varepsilon p_t}$ for any $\varepsilon \in (0, 1/2)$. If $L < \tau < U$, then let $\tau = U$\;\label{line:tau}
		\If{$\tau \ge U$}{\label{line:head_first}
			$\widehat{NDCC}(D) \gets \frac{1}{ \left| \left\{ u \in S_h, \Deg{u}{} \in D \right\} \right| } \sum_{u \in S_h, \Deg{u}{} \in D} \frac{\hat{T}_h[u]}{{\Deg{u}{} \choose 2}}$\;
			$\widehat{WDCC}(D) \gets \frac{\sum_{u \in S_h, \Deg{u}{} \in D} \hat{T}_h[u]}{\sum_{u \in S_h, \Deg{u}{} \in D} {\Deg{u}{} \choose 2}}$\;\label{line:head_last}
		}
		\Else{\label{line:tail_first}
			$\widehat{NDCC}(D) \gets \frac{1}{ \left| \left\{ v \in S_t, \lceil \EstDeg{v}{} \rceil \in D \right\} \right| } \sum_{v \in S_t, \lceil \EstDeg{v}{} \rceil \in D} \frac{\hat{T}_t[v]}{{\lceil \EstDeg{v}{} \rceil \choose 2}}$\;
			$\widehat{WDCC}(D) \gets \frac{\sum_{v \in S_t, \lceil \EstDeg{v}{} \rceil \in D} \hat{T}_t[v]}{\sum_{v \in S_t, \lceil \EstDeg{v}{} \rceil \in D} {\lceil \EstDeg{v}{} \rceil \choose 2}}$\;\label{line:tail_last}
		}
		\Return $\widehat{NDCC}(D), \widehat{WDCC}(D) $\;\label{line:return}
	\end{algorithm}

	\subsection{Analysis}
	\label{sec:analysis}
	In the following, we first present our main theoretical contributions, which provably demonstrate the quality of our returned estimates. 
	We fix a degree interval $D = [L,U)$, and analyze the quality of estimates from the head and tail estimator separately, with respect to the exact quantities $\NodeAvgCC{D}{}$ and $\WedgeAvgCC{D}{}$.
	We recall the parameters used: $p_h$, $p_t$ are the probabilities of sampling head and tail nodes; $p_{M_h}, p_{M_t}$ are the probabilities for sampling edges within the head and tail main samples; $p_{A_h}, p_{A_t}$ are the probabilities for sampling edges within the head and tail auxiliary samples.
	In our analysis, we give lower bounds for all these parameters, \emph{except} for $p_t$, which implicitly determines whether the head or tail estimator should be used (i.e., $p_t$ defines the degree threshold $\tau$, line~\ref{line:tau} of Algorithm~\ref{alg:estimate}).
	Further details on the choice of $p_t$ are given in Section~\ref{sec:one_pass_algo}.
	We use parameters $\eps$ for multiplicative error and $\eta$ for additive error.
	% For the theoretical analysis, as in standard in the streaming literature, we do not optimize dependencies on $\poly(\eps^{-1}\eta^{-1}\log n)$ factors.
	All the proofs are in Section~\ref{sec:additional_proofs} of the Appendix.
	
	\emph{Head Estimator.}
	We start by focusing on the head estimator. Note that the first pass counts exactly the degrees of all vertices in $S_h$, since nodes are sampled (with probability $p_h$) via hashing.
	The estimator error comes from the estimate of the number $\NumNodes{D}{}$ of nodes having degree in $D$, and from the estimates of triangles contained in $\hat{T}_h$. 
	We state the following theorem.
	\begin{theorem}[Head Estimator]
		\label{thm:head_estimator}
		Consider degree interval $D = [L, U)$ such that $L < U \le \tau$. For any $\varepsilon \in (0, 1/2), \eta > 0$, if
		$$
		p_h = \bigOmegatilde{ \frac{1}{\eta^2} \left( \frac{1}{\NumNodes{D}{}} + \frac{\Delta_V(D)}{\Triangles{D}{}} \right)}  \text{ and }
		$$
		$$
		p_{M_h} p_{A_h} = \bigOmegatilde{\frac{1}{\varepsilon^2} \left( \eta^2 + \frac{\Delta_E(D)}{ \Triangles{D}{}} \right) },
		$$
		then Algorithm~\ref{alg:estimate} outputs values $\widehat{NDCC}(D), \widehat{WDCC}(D)$ such that:
		\begin{align*}
			& \left| \widehat{NDCC}(D) - \NodeAvgCC{D}{} \right| \le \varepsilon \cdot \NodeAvgCC{D}{} +  \eta, \text{ and} \\
			& \left| \widehat{WDCC}(D) - \WedgeAvgCC{D}{} \right| \le \varepsilon \cdot \WedgeAvgCC{D}{} +  \eta,
		\end{align*}
		with high probability.
	\end{theorem}

	The result in Theorem~\ref{thm:head_estimator} highlights the trade-off between \emph{node} sampling (with probability $p_h$) and triangle counting via \emph{edge} sampling with probabilities $p_{M_h}, p_{A_h}$, captured by the dependence on $\eta$. 
	Concretely, when $\eta^2 \ge \Delta_E(D) / \Triangles{D}{}$, decreasing $\eta$ lowers the required edge-space but raises $p_h$, increasing the node-space.

	Recall that $\Delta_V(D)$ is the maximum number of triangles incident to a node in $\InducedNodes{D}{}$, and $\Delta_E(D)$ is the maximum number of triangles (among those incident to nodes in $V_D$) adjacent to an edge.
	The dependence on such quantities is typically required in the streaming setting (see~\cite{jayaram2021optimal,pagh2012colorful}), since the quality of the estimates depends on how concentrated the triangles are around a few nodes or edges.

	\emph{Tail Estimator.}
	Given that the tail estimator does not observe the exact node degrees, performing the degree binning introduces errors. To capture this effect, we express the guarantees with respect to ``shrunk'' and ``enlarged'' degree intervals. 
	Define $S_{LCC}(D)$ as the sum of local clustering coefficients for nodes in $\InducedNodes{D}{}$, $C(D) = |\InducedNodes{D}{}|$, $T(D)$ as the number of triangles incident to nodes in $\InducedNodes{D}{}$ and $W(D)$ as the number of wedges centered at nodes in $\InducedNodes{D}{}$. 
	Note that we can write $\NodeAvgCC{D}{} = S_{LCC}(D) / C(D)$, and $\WedgeAvgCC{D}{} = T(D)/ W(D)$. 
	We prove the following result. 
	
	\begin{theorem}[Tail Estimator]
		\label{thm:tail_estimator}
		Consider degree interval $D = [L, U)$ such that $\tau \le L < U$. 
		For any $\varepsilon \in (0, 1/2)$, define $D^+ = [L (1 - \varepsilon), U(1 + \varepsilon)] $, and $D^- = [L (1 + \varepsilon), U(1 - \varepsilon)]$.
		For the quantities $\NodeAvgCC{D}{} = \frac{S_{LCC}(D)}{C(D)}, \WedgeAvgCC{D}{} = \frac{T(D)}{W(D)}$, if
		$$
		\tau = \bigOmegatilde{ \frac{1}{\varepsilon p_t} }  \text{ and } 
		p_{M_t} p_{A_t} = \bigOmegatilde{\frac{\Delta_E \left(D^+\right)}{ \varepsilon^2 \Triangles{D^-}{}}},
		$$
		then Algorithm~\ref{alg:estimate} outputs values $\widehat{NDCC}(D), \widehat{WDCC}(D)$ such that:
		\begin{align*}
			& \widehat{NDCC}(D) \in \left[ \left(1 - \varepsilon \right) \; \frac{S_{LCC}\left(D^-\right)}{\NumNodes{D^+}{}},  \left( 1 + \varepsilon \right) \; \frac{S_{LCC}(D^+)}{\NumNodes{D^-}{}} \right], \text{ and}\\
			& \widehat{WDCC}(D) \in \left[ \left(1 - \varepsilon \right) \; \frac{T(D^-)}{W(D^+)},  \left( 1 + \varepsilon \right) \; \frac{T(D^+)}{W(D^-)} \right],
		\end{align*}
		with high probability.
	\end{theorem}
	
	Theorem~\ref{thm:tail_estimator} shows that the tail estimator obtains multiplicative approximation guarantees with respect to not the original ratios defined by $\NodeAvgCC{D}{}$ and $\WedgeAvgCC{D}{}$, but with respect to slightly distorted ratios, which comprise ``shrunk'' and ``enlarged'' intervals (compared to $D$, and depending on $\varepsilon$).  
	Note that $D^- \subseteq D \subseteq D^+$.
	In Section~\ref{sec:empirical_bounds_tail} of the Appendix, we show that the bounds expressed in Theorem~\ref{thm:tail_estimator} are empirically very precise, even if they consider the slacked intervals $D^-, D^+$. 
	
	Moreover, the tail estimator contains all the high-degree nodes (specifically, nodes with degree $\ge \tau$) in the graph with high probability (see Section~\ref{sec:additional_proofs} for the proof). 
	Considering such nodes, we quantify the error of each of their estimated local clustering coefficient. 
	We have the following result. 
	
	\begin{theorem}[Tail estimator for LCC]
		\label{thm:tail_LCC}
		Consider $\tau = \widetilde{\Omega}\!\left(\frac{1}{\varepsilon p_t}\right)$.
		Let $v \in S_t$ be a node such that $d_v \ge \tau$ and $\LocalCC{v}{} > 0$. 
		Let $\widehat{LCC}(v) = {\hat{T}_t[v]} / {{\lceil \EstDeg{v}{} \rceil \choose 2 }}$, with $\EstDeg{v}{} = S_t[v]$. 
		For any $\varepsilon \in (0,1/2)$ and any $\eta \in (0,\LocalCC{v}{}]$, if
		$$p_{M_t}p_{A_t} = \widetilde{\Omega}\!\left(\frac{\Delta_E(v)\LocalCC{v}{}}{\eta^2 \tau^2} \right), $$
		then, with high probability,
		$$
		\left|\widehat{LCC}(v)-\LocalCC{v}{}\right| \le \varepsilon \LocalCC{v}{} + \eta.
		$$
	\end{theorem}
	
	Above, $\Delta_E(v)$ is the maximum number of triangles (among those incident to $v$) adjacent to an edge. 
	Theorem~\ref{thm:tail_LCC} therefore shows that the error in estimating LCC decomposes into a multiplicative term (due to binning approximate degrees) and an additive term (due to triangle estimation), for any node $v$ with degree at least $\tau$.
	Moreover, Theorem~\ref{thm:tail_LCC} formalizes desirable properties of the estimator: accurate LCC estimation depends only on the structure around $v$ (i.e., from $\Deg{v}{} \ge \tau, \LocalCC{v}{}$ and $\Delta_E(v)$), and not on global graph properties. We also see that LCC estimates are more accurate for sufficiently high-degree nodes, whose incident triangles do not concentrate excessively around the same edge. 
	%  (i.e., triangles causing depencence in our triangle estimator, which uses edge sampling).
	%Theorem~\ref{thm:tail_LCC} provides a multiplicative plus additive approximation for local clustering coefficient estimates for sufficiently high-degree nodes. 
	
	%\emph{Choice of degree threshold $\tau$.} 
	%As we discuss in the next section, the algorithm is run on a collection of intervals
	%$D_1, D_2, \ldots$.
	%The degree threshold $\tau$ determines whether the head estimator or the tail estimator provides an accurate output. Note that $\tau$ is basically $1/p_t$ (ignoring $\poly(\eps^{-1} \eta^{-1} \log n)$ factors). 
	%The ideal choice of $p_t$ is to balance the space required by the head and tail estimators.
	%We explain this choice after Theorem~\ref{thm:space}.
	
	% Ideally, the threshold $\tau$ should be picked to minimize the bias-variance tradeoff between our head estimator and our tail estimator. 
	% 	In practice, this quantity cannot be selected precisely a priori. 
	% 	Instead, as discussed, a practical strategy is to fix the probability $p_t$ for sampling edges in the the first pass (e.g., $p_t = 0.1$, so that we sample in expectation the 10\% of edges in the first pass), and set $\tau$ as reported in Alg.~\ref{alg:estimate}.}
% 
\section{One-Pass Algorithm}
\label{sec:one_pass_algo}
In the following, we show how to combine the two passes presented in Section~\ref{sec:two_pass_algorithm} into a single pass.
We analyze the head and tail estimators separately.

For the head estimator, the aggregation is immediate. The head sample $S_h$ is defined via hashing: for a node $u$, its membership in $S_h$ is decided irrevocably upon the first occurrence of $u$ in the stream. From that time onward, we maintain both its exact degree counter (lines~\ref{line:fp_head_first}-\ref{line:fp_head_else} of Alg.~\ref{alg:first_pass}), and subsequently execute the head triangle-counting routine (lines~\ref{line:sp_ct_head},~\ref{line:sp_se_h},~\ref{line:sp_se_ah} of Alg.~\ref{alg:second_pass}). This process yields the same exact degrees tracked in $S_h$ that would be obtained by running Alg. \algnamefirstpass, and the same triangle estimator that would be obtained by running Alg. \algnamesecondpass\ on $S_h$, after the first pass. 

Handling the tail estimator requires care, since a node may enter $S_t$ after multiple edges incident to it have already passed. 
%The second pass is necessary to count the triangles incident to these missed edges. A direct implementation in a single pass leads to systematic undercount and biased estimates.
% , given the nature of the edge sample $E_t$, which eventually defines the tail nodes (line~\ref{line:fp_tail} of Alg.~\ref{alg:first_pass}).
% If we naively address the problem as done for head nodes (i.e., start counting triangles incident to a node only after it first appears in the current $E_t$), then the resulting estimates can be biased. 
Addressing the problem as done for the head estimator can induce systematic errors in the estimates. 
Consider the arrival of edges $\Edge{v}{w}, \Edge{w}{u}, \Edge{u}{v},\Edge{x}{w}$, with $S_t, M_t, A_t = \emptyset$. 
Suppose nodes $u, v, w$ from the first two edges are not sampled in $S_t$, but such edges are sampled in the auxiliary sample $A_t$. When $\Edge{u}{v}$ arrives, we do not observe any triangle since we have $M_t = \emptyset$. Recall that triangles are counted if at least one edge is in $M_t$ when the closing one arrives (see Alg.~\ref{alg:count_triangles}). However, the triangle $\{u, v, w\}$ could be observed from $A_t = \{ \Edge{v}{w}, \Edge{w}{u} \}$ when $\Edge{u}{v}$ arrives.
Later, if node $w$ is sampled into $S_t$ after seeing $\Edge{x}{w}$, then $w$ becomes tracked (tail node), but the earlier triangle $\{u, v, w\}$ has been permanently missed, inducing a systematic undercount (biased estimates). 

To preserve the unbiasedness of triangles for the tail estimator, we decouple triangle counting from tail membership. 
Concretely, after updating $S_t$ (as in Alg.~\ref{alg:first_pass}), we run tail triangle-counting routine (lines~\ref{line:sp_ct_tail},~\ref{line:sp_se_t},~\ref{line:sp_se_at} of Alg.~\ref{alg:second_pass}) with the modification of accounting triangles regardless of whether the triangle edges belong to the set $M_t$. 
In our two-pass algorithm, we only counted triangles with at least one edge in the (tail) main sample. Now we remove this requirement and allow counting triangles even if \emph{both edges} of the triangle (when the arriving one closes the triangle) belong to the \emph{auxiliary} sample.
Thus, implementing the one-pass algorithm requires a slight modification to \texttt{CountTriangles} (Alg.~\ref{alg:count_triangles}) when dealing with tail samples. It needs to check for triangles in $Adj_A(u) \cap Adj_A(v)$ closed by edge $\Edge{u}{v}$, 
with the corresponding increment being $1 / p_{A_t}^2$ since both edges are detected in the auxiliary sample $A_t$. 

This process yields the same set $S_t$ as the one produced by running \algnamefirstpass, which determines the tail nodes.
Moreover, we prove that when $p_{M_t} = p_{A_t}$, the resulting triangle estimator coincides with the one obtained by running \algnamesecondpass\ on $S_t$ after the first pass. 
\begin{lemma}
	\label{lemma:triangle_tail_equiv}
	The estimators produced by the one-pass and two-pass algorithms are identical for the node-degree map $S_t$ and, if $p_{M_t}=p_{A_t}$, also for the node-triangle map $\hat{T}_t$.
\end{lemma}
The proof of Lemma~\ref{lemma:triangle_tail_equiv} follows from the principle of deferred decisions~\cite{mitzenmacher2017probability}, and is given in Section~\ref{sec:additional_proofs}.

% 
% -- overcounts
Let $\bar{n} = |S_t|$ be the number of nodes in $S_t$ at the end of the stream, i.e., the number of tail nodes for which we estimate their degrees.
Our single-pass algorithm \emph{may} maintain triangle counters for more than $\bar{n}$ nodes. The reason is that we are counting triangles even if the two edges are in the auxiliary samples; for the purpose of our algorithm, such ``wasted'' counters will be discarded if we end up without knowing the degrees of the nodes forming such triangles. 
We observe empirically that the number of ``wasted" counters remains within a small constant factor of $\bar{n}$, where this constant ranges from 0 to 2.2 for representative settings (see Section~\ref{sec:extra_counters}).

Remarkably, this implies that theoretical guarantees for the two-pass algorithm (Section~\ref{sec:analysis}) directly extend to our presented one-pass algorithm, with minimal bookkeeping overhead. We now combine these results to derive the space and time complexity of our one-pass algorithm.

\begin{theorem}
	\label{thm:space}
	Let $\mathcal{D} = \{ D_1, D_2, \dots \}$ be a set of degree intervals.
	There exists a \emph{one-pass} streaming algorithm that, for all $D = [L, U) \in \mathcal{D}$ \emph{simultaneously}, computes approximations of $\NodeAvgCC{D}{}$ and of $\WedgeAvgCC{D}{}$ with the guarantees of Theorems~\ref{thm:head_estimator} and~\ref{thm:tail_estimator}, using \emph{expected} space:
	$$
	\bigOtilde{\frac{n}{\eta^2} \; \max_{D \colon U \le \tau } \left( \frac{1}{\NumNodes{D}{}} + \frac{\Delta_V(D)}{\Triangles{D}{}} \right) + \frac{m}{\varepsilon \tau}} \text{ nodes, and}
	$$
	$$
	\bigOtilde{\frac{m}{\varepsilon} \left( \eta + \max_D \sqrt{\frac{\Delta_E(D^+)}{\Triangles{D^-}{}}} \; \right) } \text{edges, }
	$$
	where $\tau = \frac{1}{\varepsilon p_t}$.
\end{theorem}

\emph{The sublinearity of space and choice of $p_t$.} 
Let us unpack the bounds of Theorem~\ref{thm:space} to better understand the right choice for $p_t$.
For simplicity, consider degree intervals of the form $D_i = [b^i, b^{i+1})$, for some $b > 1$, so that the intervals cover the entire range of degrees. 
Ignoring the dependencies on $\eps$ and $\eta$, as well as the triangle terms, define $F_\tau = \min_{D \colon U \le \tau} \NumNodes{D}{}$. 
Then Theorem~\ref{thm:space} implies that, in expectation, our one-pass algorithm stores $n/F_\tau$ head nodes plus $m/\tau$ tail nodes, where $\tau \approx 1/ p_t$.
Observe that as $\tau$ increases (so tail nodes storage decreases), the minimization defining $F_\tau$ considers more intervals, and $F_\tau$ decreases (so head nodes storage increases). 
This is the primary tradeoff that the choice of $p_t$ (and consequently $\tau$) brings in.

In a typical heavy-tailed degree distribution, $\NumNodes{D_i}{}$ is usually monotonically decreasing in $i$. 
That is, higher-degree intervals contain fewer vertices. Therefore, the ideal choice of $\tau$ is a moderate degree threshold such that every interval below $\tau$ still contains many vertices. 
This trade-off occurs in many sublinear results for estimating the degree distribution~\cite{simpson2015catching,eden2018provable,bishnu2025towards}, where the optimal choice is often characterized by a variant of the $h$-index\footnote{The $h$-index is the largest $h$ such that at least $h$ nodes have degree at least $h$.} of the graph.
In practice, these optimal trade-off points are not possible to determine without prior knowledge of the input graph.

A practical approach is then to set $p_t$ to a small constant, e.g., $p_t = 0.005$, in order to store (in expectation) 1\% of the total edges defining tail nodes. 
However, note that this is a worst-case bound that considers every node. A tighter bound would be $\Exp{\CardSet{S_t}} \le \sum_v \left( 1 - \left(1 - p_t \right)^{\Deg{v}{}} \right) \le \sum_v \min \left(1, p_t \Deg{v}{} \right)$, which limits the contribution of high-degree nodes in the summation.
Again, $\eta$ governs the node-edge space tradeoff: for sufficiently large $\eta$, decreasing its value reduces the edge-space at the cost of increased node-space.

As shown in Section~\ref{sec:bound_deltas} of the Appendix, in real-world graphs $\Delta_V(D)$ and $\Delta_E(D)$ are orders of magnitude smaller than $\Triangles{D}{}$ for representative intervals. Thus, in practice, the triangle-dependent terms in the space bounds of Theorem~\ref{thm:space} are small and contribute only marginally to the total space.

For completeness, we analyze the time complexity of our one-pass algorithm. 
In the analysis, we also account for how our sets of nodes and edges are implemented. 
The worst-case time complexity is dominated by counting the triangles closed by each arriving edge. 
\begin{proposition}
	\label{prop:time}
	Given an input stream $\Sigma$ for graph $G$ with $m$ edges and maximum degree $d_{max}(G)$, our one-pass algorithm processes the stream in $\bigO{m \cdot d_{max}(G)}$ time.  
\end{proposition}
Proposition~\ref{prop:time} provides a worst-case bound. In practice, in real graph streams, the cost of computing common neighbors for an edge is much smaller than $d_{max}(G)$, since it depends on the structure of the subgraph induced by the maintained samples, rather than on the maximum degree of $G$ alone.

\subsection{\algname\ Algorithm}
\label{sec:algo_practice}
In practice, the implementation of our proposed one-pass algorithm differs slightly, as we apply the following modification involving the triangle counting phase. 
Instead of sampling edges independently with fixed probability (which enables easier concentration bounds in our theoretical analysis), we implement all our samples via reservoir sampling~\cite{vitter1985random}. 
Reservoir sampling approach is widely used in popular and practical methods for counting triangles in one-pass graph streams~\cite{stefani2017triest,shin2017wrs,wu2025great,boldrin2024fast}, as it keeps space bounded \emph{with certainty} given a memory budget, and not only in expectation (as in Theorem~\ref{thm:space}). 
Reservoir sampling guarantees~\cite{vitter1985random} that the maintained sample of edges is a uniform sample over all the possible sets of edges seen so far, at any time in the stream. 
Therefore, instead of parameters $p_{M_h}, p_{A_h}, p_{M_t}, p_{A_t}$ (used in Alg.~\ref{alg:second_pass}), we consider reservoir samples having memory budgets $B_{M_h}, B_{A_h}$ for head samples, and $B_{M_t}, B_{A_t}$ for tail samples. This approach \emph{guarantees} that our algorithm stores at most $B_{M_h} + B_{M_t} + B_{A_h} + B_{A_t}$ edges at \emph{any time} in the stream. 

Let $m'$ be the number of edges incident to sampled head nodes. Notice that by fixing $B_{M_h} = p_{M_h} \cdot m'$, then edge sampling described in Section~\ref{sec:two_pass_algorithm} samples in expectation $B_{M_h}$ edges in the main head sample $M_h$. 
Since we are unaware of $m'$, from a practitioner's point of view it is easier to set a fixed budget and implement a strategy that promises to maintain the memory bounded. 
The same applies for other edge samples. 

Note that \algname\ does not need $\tau$ to compute estimates ($\tau$ is used only at the end of the stream). A practitioner can then try multiple values of $\tau$ and inspect the resulting estimates.
As proposed in the implementation of~\cite{simpson2015catching}, we pick the degree threshold $\tau$ as the largest degree such that there are at least 10 nodes having that degree in the sample $S_h$. 
%This implies that there are roughly at least $10 / p_h$ nodes of degree $\tau$ in the original graph (for sufficiently large $p_h$). 
This choice empirically~\cite{simpson2015catching} allows for a smooth transition between head and tail estimators.

\begin{algorithm}[t]
	\caption{\algname $\left( \Sigma, p_h, p_t, B_{M_h}, B_{M_t}, B_{A_h}, B_{A_t} \right)$}
	\label{alg:bolide}
	\LinesNumbered
	%	\DontPrintSemicolon
	\kwInput{Graph stream $\Sigma$; \\
		Probabilities $p_h, p_t$ of sampling head and tail nodes; \\
		Budgets $B_{M_h}, B_{M_t}$ for head and tail main samples; \\
		Budgets $B_{A_h}, B_{A_t}$ head and tail auxiliary samples.
	}
	\kwOutput{Estimate of degrees and estimates of triangles for head and tail nodes; degree threshold.}
	% \tcp{Initialization}
	Init head and tail node-degree maps $S_h, S_t \gets \emptyset$\label{line:bolide_init_first}\;
	Init head and tail node-triangle maps: $\hat{T}_h, \hat{T}_t  \gets \emptyset$\;
	Init reservoir of edges $M_h, M_t, A_h, A_t \gets \emptyset$\;
	Init number $h, t$ of head and tail edges: $h, t \gets 0$\;
	Init number $a_h, a_t$ of aux head and aux tail edges: $a_h, a_t \gets 0$\label{line:bolide_init_last}\;
	\For{\textup{each incoming edge} $e = \{u, v\}$ \textup{in the stream} $\Sigma$}{\label{line:bolide_for_each}
		% \tcp{Nodes Sampling}
		\If{$u \notin S_h$}{\label{line:bolide_first}
			\lIf{$hash[u] < p_h$}{$S_h[u] \gets 1$}}
		\lElse{$S_h[u] \gets S_h[u] + 1$}\label{line:bolide_second}
		\lIf{$u \notin S_t$}{set $S_t[u] \gets 1$ with probability $p_t$} \label{line:bolide_third}
		\lElse{$S_t[u] \gets S_t[u] + 1$} \label{line:bolide_last}
		Repeat the above steps~\ref{line:bolide_first}-\ref{line:bolide_last} also for $v$\;\label{line:bolide_repeat}
		% \tcp{Triangle Counting}
		\texttt{CTReservoir}($\{u, v\}, M_h, B_{M_h}, h, A_h, B_{A_h}, a_h, \hat{T}$)\;\label{line:bolide_ct_head}
		\texttt{CTReservoir}($\{u, v\}, M_t, B_{M_t}, t, A_t, B_{A_t}, a_t, \hat{T}$)\;\label{line:bolide_ct_tail}
		% \tcp{Edge Sampling}
		\lIf{$u \in S_h$ or $v \in S_h$}{\texttt{SEReservoir}($\{u, v\}, M_h, B_{M_h}, h$)}\label{line:bolide_se_head_main}
		\lElse{\texttt{SEReservoir}($\{u, v\}, A_h, B_{A_h}, a_h$)}\label{line:bolide_se_head_aux}
		\lIf{$u \in S_t$ or $v \in S_t$}{\texttt{SEReservoir}($\{u, v\}, M_t, B_{M_t}, t$)}\label{line:bolide_se_tail_main}
		\lElse{\texttt{SEReservoir}($\{u, v\}, A_t, B_{A_t}, a_t$)}\label{line:bolide_se_tail_aux}
	}
	% \tcp{Correct tail degrees by expected loss}
	Let $\ell(r) = \lceil \frac{1 - p_t - (1 - p_t)^{r+1} - r p_t (1 - p_t)^r}{p_t (1 - (1 - p_t)^r)} \rceil$\;\label{line:bolide_correct_first}
	Find $r$ s.t. $S_t[v] = r - \ell(r)$ and set $S_t[v] \gets r$, for each $v \in S_t$\; \label{line:bolide_correct_last}
	% Set $S_t[v] \gets \left( S_t[v] - \ell(S_t[v]) \right)$ for all $v \in S_t$\;
	% \tcp{Select head-tail degree threshold}
	Let $C_h(d)$ be the number of nodes with degree exactly $d$ in $S_h$\;\label{line:bolide_tau_first}
	Set $\tau$ to be the largest degree $d$ such that $C_h(d) \geq 10$\;\label{line:bolide_tau_last}
	\Return head and tail node-degree estimates $S_h$, $S_t$, head and tail triangle estimates $\hat{T}_h$, $\hat{T}_t$, threshold $\tau$.\label{line:bolide_return}
\end{algorithm}

We denote the implementation of our one-pass algorithm as \algname. Its pseudocode is reported in Algorithm~\ref{alg:bolide}, and consists in the subsequent run of the first pass and the second pass (using reservoir sampling) merged into a single pass, as described in Section~\ref{sec:one_pass_algo}.
\algname\ uses two subroutines: \texttt{CTReservoir} and \texttt{SEReservoir}, respectively the versions of \texttt{CountTrian\-gles} and \texttt{SampleEdge} adap\-ted to reservoir sampling, and described in Section~\ref{sec:additional_algo}.
Importantly, as discussed in Section~\ref{sec:one_pass_algo}, \texttt{CTReservoir} counts triangles even if both edges (when the closing one arrives) are in the tail auxiliary sample. 

Notice that our triangle estimators can be viewed as a ``black-box'': any existing edge-sampling scheme from the literature~\cite{shin2017wrs,stefani2017triest,wu2025great,boldrin2024fast} can be plugged in to derive triangles' estimates. 
In our implementation of \algname, we adopt the reservoir sampling scheme of \triestalg~\cite{stefani2017triest}.
Note that the approximation guarantees from Theorems~\ref{thm:head_estimator},~\ref{thm:tail_estimator} and~\ref{thm:tail_LCC} do not follow when using reservoir sampling, as such results are proved for fixed probability sampling.
%Since reservoir sampling creates dependencies among sampled edges, characterizing precisely the approximation guarantees from our algorithm using reservoir sampling would introduce covariance and arrival-time terms~\cite{stefani2017triest}. 
On the other hand, reservoir sampling allows to maintain the memory bounded deterministically, and not only in expectation (as in Theorem~\ref{thm:space}), which is why we adopt it in practice.

\emph{Choice of main-auxiliary split.} 
Consider a fixed probability $p_h$ of sampling head nodes, and a memory budget $k_h = B_{M_h} + B_{A_h}$ for the head samples. 
We determine the optimal main-auxiliary split  by minimizing an upper bound on the simplified variance\footnote{Simplified variance represents the total variance as just the sum of individual variances, ignoring the covariance terms. It is commonly used in previous work~\cite{lee2020temporal,wu2025great}, as it is strongly correlated ($R^2 > 0.99$) with traditional variance in real-world graphs~\cite{lee2020temporal}.} of the triangle estimates. This minimization reduces to solving a simple quadratic equation, which can be easily done upfront, before initializing \algname.
We provide the complete details in Section~\ref{sec:params_selection} of the Appendix. 
To illustrate, setting $p_h = 0.1$ yields a head main-auxiliary split $B_{M_h}, B_{A_h} \approx 0.56k_h, 0.44k_h$; increasing $p_h$ to $0.2$ shifts the split to $B_{M_h}, B_{A_h} \approx 0.65 k_h, 0.35 k_h$.
An analogous optimization can be derived for tail main-auxiliary split, which leads to the choice $B_{M_t} = B_{A_t}$ for any value of $p_t$.
Again, we defer the complete details to Section~\ref{sec:params_selection}.  

We derive these (simplified) variance upper bounds by assuming random order streaming. Crucially, this assumption is utilized only as a practical heuristic to initialize (part of) our parameters and does not impact our theoretical guarantees (i.e., our estimators do not make any assumption on the streaming order). 
While real-world graph streams rarely exhibit a perfectly uniform random order, our experimental evaluation (Section~\ref{sec:experiments}) confirms that this initialization strategy leads to highly accurate estimates in practice.

\section{Experimental Evaluation}
\label{sec:experiments}
In this section, we report our main experimental results (see Section~\ref{sec:additional_experiments} of the Appendix for the complete evaluation). 
We begin by introducing the metric used to assess the performance.

\subsection{Evaluation Metric}
\label{sec:metric}
Prior work on estimating the complementary cumulative degree histogram (CCDH)~\cite{simpson2015catching,bishnu2025towards,eden2018provable} evaluates accuracy using the Relative Hausdorff (RH) distance.\footnote{RH distance is preferred to statistical measures such as KS-distance, $\chi^2$, and $\ell_p$-norms, since these measures tend to ignore the tail, which contains only a negligible fraction of the probability mass in real-world networks.}
Informally, RH distance allows a multiplicative slack in the domain, controlled by $\delta$, and a multiplicative slack in the estimate, controlled by $\varepsilon$. 

In our setting, purely multiplicative error is too restrictive when evaluating binned degree-wise clustering coefficients. 
Values are tiny (especially for high-degree nodes), so miniscule absolute deviations induce disproportionately large relative errors. 
To this end, we introduce a modified RH-style distance with Additive Slack in the range (RHAS).
\begin{definition}
	\label{def:rhas}
	Let $f,\hat f : \mathcal D \to \mathcal R$, where $\mathcal D \subseteq \mathbb Z_{>0}$ and $\mathcal R \subseteq \mathbb R_{\ge 0}$.
	Given $\delta \in [0,1)$, $\varepsilon \ge 0$, and $\eta \ge 0$, define neighborhood
	$
	\mathcal N_\delta(x)
	=
	\mathcal D \cap [(1-\delta)x,(1+\delta)x].
	$
	We say that $f$ and $\hat f$ are $(\delta,\varepsilon,\eta)$-close in the \emph{RHAS} sense if:
	\begin{myitemize}
		\item for every $x \in \mathcal D$, there exists $x' \in \mathcal N_\delta(x)$ such that
		$
		|f(x)-\hat f(x')| \le \varepsilon f(x)+\eta,
		$ and
		\item symmetrically, for every $x \in \mathcal D$, there exists $x' \in \mathcal N_\delta(x)$ such that
		$
		|\hat f(x)-f(x')| \le \varepsilon \hat f(x)+\eta.
		$
	\end{myitemize}
\end{definition}
\noindent For fixed $\delta$ and $\eta$, the corresponding RHAS distance is:
\[
\mathrm{RHAS}_{\delta,\eta}(f,\hat f)
=
\inf\left\{
\varepsilon \ge 0
\;\middle|\;
f \text{ and } \hat f \text{ are } (\delta,\varepsilon,\eta)\text{-close}
\right\}.
\]

RHAS distance allows us to assess approximation guarantees with a combination of multiplicative and additive error, as done in theory (Theorems~\ref{thm:head_estimator},~\ref{thm:tail_estimator}).
Notice that the standard RH distance~\cite{simpson2015catching,bishnu2025towards,eden2018provable} is obtained by setting $\eta = 0$ in our RHAS definition (Def.~\ref{def:rhas}).

In our experiments, we evaluate binned degree-wise clustering coefficients (Def.~\ref{def:degree_cc_node},~\ref{def:degree_cc_wedge}) over \emph{logarithmically binned degree intervals}, e.g., powers of two: $D_i = [2^i,2^{i+1})$, which is standard in practice~\cite{seshadhri2013triadic,zhang2017efficient}. 
Specifically, we compare the exact distribution $\{\NodeAvgCC{D_i}{}\}_i$ with the approximated distribution $\{\widehat{NDCC}(D_i)\}_i$, and analogously $\{\WedgeAvgCC{D_i}{}\}_i$ with $\{\widehat{WDCC}(D_i)\}_i$, under RHAS distance (Definition~\ref{def:rhas}).
These distributions are indexed by $i$, which is the logarithmic degree scale.
A useful consequence of RHAS is that the neighborhood $\mathcal N_\delta(i)$ starts including adjacent bins only when $i\ge 1/\delta$. For $\delta=0.1$ and $D_i = [2^i,2^{i+1})$, RHAS enforces strict bin-to-bin matching for low-degree regimes (below $2^{10}$), while allowing comparison with adjacent bins only for high-degree regimes (from $2^{10}$).

\subsection{Experiments}
\label{sec:exp}

The goals of our experiments are: (i) to assess the dependence of \algname\ from parameters; (ii) to 
measure the performance of \algname\ for estimating binned degree-wise clustering coefficients compared to baselines; (iii) to assess robustness of \algname\ under different orderings of the input graph stream; (iv) to compare \algname\ with similar approaches in the literature in terms of accuracy, runtime, and memory usage.

\emph{Experimental Setup.} We implemented our algorithm in C\texttt{++}22. The code to reproduce all experiments is available at~\url{https://github.com/CristianBold4/BOLIDE}. 
All the code was compiled under \texttt{gcc} 13.3.0 and ran on a 2.20 GHz Intel Xeon CPU with 1 TB of RAM, on Ubuntu 24.04. 
%For each run, we report the average and standard deviation over 10 independent trials.

We consider degree intervals $D_i = [b^i, b^{i+1})$ for different values of $b \in \{2, 1.5, 1.2, 1.1 \}$, and for all $i \in \left[1, \lceil \log_b \left( d_{max}(G) \right) \rceil \right]$.
We run our one-pass algorithm \algname\ to compute the approximate distribution $\{ \widehat{NDCC}(D_i)\}_{i}$ and $\{ \widehat{WDCC}(D_i) \}_{i}$. 
To quantitatively evaluate the approximate distributions, we consider the RHAS distance (Def.~\ref{def:rhas}) between exact and approximate distributions, at the end of the stream, by fixing $\delta = 0.1$ and $\eta = 0.001$. 

\emph{Datasets.}
We consider datasets downloaded from~\cite{leskovec2016snap,Kwak10www}, common in previous works. Datasets' statistics are summarized in Table~\ref{tab:datasets}. From each dataset, we remove self-loops and multiple edges (if any), deriving a stream of insertion-only edges. 

\begin{table}[b]
	\caption{Datasets' statistics: number $n$ of nodes; number $m$ of edges; number $T$ of triangles; $\NodeAvgCC{D}{}$ for $D = D_G$.}
	\label{tab:datasets}
	\centering
	\resizebox{\columnwidth}{!}{
		\begin{tabular}{ c c c c c }
			Dataset & $n$ & $m$ & $T$ & NDCC$(D_G)$\\
			\toprule
			as-skitter (\texttt{SK})	      & $1.7M$ & $11.1M$ & $28.7M$ & $0.2963$ \\
			livejournal (\texttt{LJ})        & $4.8M$ & $42.8M$ & $285M$  & $0.3509$ \\
			com-orkut (\texttt{ORK})	      & $3.1M$ & $117M$  & $627M$  & $0.1704$ \\
			twitter (\texttt{TW})			  & $41M$  & $1.2B$  & $34.8B$ & $0.0761$ \\ 
			com-friendster (\texttt{FR})     & $65M$  & $1.8B$  & $4.1B$  & $0.2050$ \\
			%\bottomrule
		\end{tabular}
	}
\end{table}

\emph{Choice of Parameters.}
\algname\ depends on several parameters: probabilities $p_h$ and $p_t$ of sampling head and tail nodes, memory budgets $B_{M_h}, B_{A_h}$ for head samples, and memory budgets $B_{M_t}, B_{A_t}$ for tail samples. 
Clearly, the higher each of these parameters is, the better estimates we get, as they all govern the space used by the algorithm. 

To empirically select the best parameters, we perform a grid search over values $p_h \in \{0.05, 0.1, 0.15, 0.2, 0.25, 0.3\}$ and $p_t \in \{0.025, 0.05, 0.075\} \cdot \frac{n}{m}$. 
Moreover, we strictly bound the space used by \algname, allowing the memory budget for the number of edges stored to be $0.1 m$. 

Let $k_h = B_{M_h} + B_{A_h}$ and $k_t = B_{M_t} + B_{A_t}$ be the space allowed to the head and tail samples, respectively. 
Subject to the total budget constraint $k_h + k_t = 0.1 m$, we evaluate allocations $k_h, k_t \in \{ (0.01, 0.09), (0.05, 0.05), (0.09, 0.01) \} \cdot m$. 
For each configuration, the main-auxiliary split is determined as described at the end of Section~\ref{sec:algo_practice}.
As detailed in Section~\ref{sec:params_choice}, the configuration $k_h = 0.09m, k_t = 0.01m$ consistently yields the highest accuracy (in terms of RHAS) across our evaluated datasets. 
This allocation is justified by the typical heavy-tailed distribution of real-world networks: the majority of edges are incident to the massive number of low-degree nodes, rather than the few high-degree nodes. Consequently, an effective estimator must dedicate a substantially larger budget to edges incident to head nodes.

Finally, we pick $p_h = 0.2$ and $p_t = 0.05 \cdot \frac{n}{m}$, given that we observe a balance between accuracy and memory used around such values (see our grid search's results in Section~\ref{sec:params_choice}). 

Henceforth, if not explicitly specified, we fix parameters as explained above. That is, for each dataset, we set $p_h = 0.2, p_t = 0.05 \cdot \frac{n}{m}, B_{M_h} = 0.058 m, B_{A_h} = 0.032m, B_{M_t} = B_{A_t} = 0.005m$.
Therefore, the space used by our algorithm \algname\ is $0.3n$ nodes (worst-case, in expectation), and $0.1m$ edges, including duplicates.
Since $n$ is typically an order of magnitude less than $m$ (see Table~\ref{tab:datasets}), storing $0.3 n$ nodes is cheap, as the space used by \algname\ is mainly devoted to edges in the reservoir samples. 

\PrintFigureBaslines
\PrintFigure{WDCC}{1.5}
\emph{Results.}
We consider as baselines state-of-the-art algorithms for approximating local triangles in one-pass streams: \triestalg~\cite{stefani2017triest} and \wrsalg~\cite{shin2017wrs,lee2020temporal}. 
\triestalg\ uses a single reservoir sampling to sample incoming edges, while \wrsalg\ couples reservoir sampling with a waiting room storing the most recent edges.
Together with local triangle estimates produced by the baselines, we consider the unbiased degree estimator derived from their stored sample to produce NDCC estimates (see Section~\ref{sec:comparison_baselines} for further details). 

We run \triestalg, \wrsalg\ and \algname\ with memory budget $0.1m$. We compute estimates for degree intervals $D_i = [b^i, b^{i+1})$ for $b \in \{2, 1.5, 1.2, 1.1\}$, using Algorithm \texttt{Estimate} (Alg.~\ref{alg:estimate}) for \algname\, and an equivalent adaptation (considering only the ``tail'' branch) for the baselines. 

Given their similarity, and due to space constraints, we only report results for NDCC with $b = \{2, 1.2\}$, and for WDCC with $b = \{1.5\}$, and we defer to Section~\ref{sec:additional_results} for complete results.
We report the exact and estimated distributions, namely $\{ \NodeAvgCC{D_i}{} \}_i$ vs $\{ \widehat{NDCC}(D_i)\}_i$, and $\{ \WedgeAvgCC{D_i}{} \}_i$ vs $\{ \widehat{WDCC}(D_i)\}_i$, for each degree interval $D_i$ to qualitatively observe the accuracy of estimates.
Moreover, we show ``point-wise'' values of $\varepsilon$ for RHAS (Definition~\ref{def:rhas}), at each degree scale $i$. From Section~\ref{sec:metric}, it is easy to see that RHAS distance between exact and estimated distributions corresponds to the maximum ``point-wise'' value of $\varepsilon$.

Figures~\ref{fig:compact_NDCC_base1.2},~\ref{fig:compact_baselines} and~\ref{fig:compact_WDCC_base1.5} show results.  
We observe that \algname\ is substantially accurate in estimating both NDCC and WDCC, with values of $\varepsilon$ (``point-wise'' RHAS) typically well below 10\% for all considered intervals.
Notably, both \triestalg\ and \wrsalg\ significantly fail to provide accurate approximations for low-degree nodes. This is because \triestalg\ samples edges via reservoir sampling (nodes are kept with probability proportional to their degrees), and \wrsalg\ biases the sample towards recently observed edges on the stream. These results confirm the effectiveness of our head estimator. For the sake of visualization, we show baselines comparison only in Figure~\ref{fig:compact_baselines}.

% -- ordering of stream
\emph{Robustness across orderings of the stream.}
Next, we assess the performance of \algname\ under different orderings of the input graph stream. 
As done in~\cite{simpson2015catching}, we consider the following orderings: (i) \texttt{Rand}, shuffles the original stream uniformly at random; (ii) \texttt{Adj-Incr}, adjacency-list ordering with nodes in increasing degree; (iii) \texttt{Adj-Decr}, adjacency-list ordering with nodes in decreasing degree; (iv) \texttt{Adj-Rand}, adjacency-list ordering with nodes shuffled uniformly at random. 
For each dataset and ordering, we report RHAS distance considering NDCC distribution, for degree intervals of powers of two (i.e., $D_i = [2^i, 2^{i+1})$). 
Given that the goal of this experiment is not scalability, we do not compute results for our biggest datasets (\texttt{TW} and \texttt{FR}).

Table~\ref{tab:shuffle} shows results, with \algname\ achieving comparable and robust performance across the different orderings of the input stream. 
%For \texttt{LJ}, we observe slightly higher RHAS distance with \texttt{Adj-Incr} and \texttt{Adj-Decr} orderings; these values are acceptable, as RHAS distance is (on average) always below 12\%, with space used by \algname\ fixed to $11\%m$ edges. 
\TableShuffle

% -- Comparison with Kutzov + Pagh
\FigureLCC
\TableKPComparison

\emph{Estimate of Local Clustering Coefficients.}
We compare \algname\ with \texttt{EstimateClusteringCoefficients} algorithm~\cite{kutzkov2013streaming} (\kpalg\ from now on). 
\kpalg\ provides estimates of local clustering coefficients for nodes having degree \emph{at least $d$} (for some $d$) in a single-pass stream, and is tightly related to local clustering coefficients estimates from our tail estimator (see Theorem~\ref{thm:tail_LCC}). 

Briefly, \kpalg\ algorithm works by running $K$ independent and parallel copies of monochromatic graph sparsification, using $C$ colors. 
That is, in each copy every node is assigned a random color, and only edges connecting nodes of the same color are kept. 
For each sparsified graph, \kpalg\ then estimates the fraction of closed wedges, and the LCC for each node is reported as the overall ratio across all the parallel copies $K$, provided that such nodes have been retained a sufficient number of times. 
Note that the expected space used by \kpalg\ in terms of the number of edges is $K \cdot m / C$.

Again, we run \algname\ by fixing memory budget to $0.1$ edges.
We run authors' implementation of \kpalg\ by fixing $K = 200$,\footnote{\cite{kutzkov2013streaming} ranges $K \in [50, 100, 150, 200, 250, 300, 350, 400]$. We consider $K = 200$ as in the main experiment of~\cite{kutzkov2013streaming}, and we defer to Section~\ref{sec:comparison_lcc} for a more in-depth analysis.} and set $C = K / 0.1 = 2000$, in order to account for the same (expected) memory space of \algname\ (in terms of number of edges stored). 
Notice that \kpalg\ is, by design, estimating local clustering coefficients of nodes with degree at least some $d$, which is comparable to what our tail estimator is implicitly doing (for nodes with degree $\ge \tau$ ). 
%To make sure that \algname\ is using the tail estimator, we compute the following. 
Let $\LocalCC{v}{}$ and $\EstLocalCC{v}{}$ be the true and estimated local clustering coefficient of node $v$, respectively. 
Let $H$ be the set of nodes of degree at least $\max \left( d, \tau \right)$ for which both \kpalg\ and \algname\ estimate their LCCs. 
As in~\cite{kutzkov2013streaming}, we consider the following metrics:
(i) Average relative error: 
$
\frac{1}{|H|} \sum_{v \in H} | \LocalCC{v}{} - \EstLocalCC{v}{} | / \LocalCC{v}{};
$
(ii) Average absolute error:
$
\frac{1}{|H|} \sum_{v \in H} | \LocalCC{v}{} - \EstLocalCC{v}{} |;
$
(iii) Pearson correlation coefficient between $\{\LocalCC{v}{} \}_{v \in H}$ and $\{\EstLocalCC{v}{} \}_{v \in H}$, measuring the linear correlation between true and estimated values (see~\cite{pearson1895vii,kutzkov2013streaming} for details).
We do not compare runtimes given that, as admitted in~\cite{kutzkov2013streaming}, \kpalg\ is implemented by running copies sequentially (and not in parallel). That is, \kpalg\ is making $K$ passes of data and is always slower than \algname. 
For instance, \kpalg\ takes on average 16 and 25 hours on \texttt{TW} and \texttt{FR}, respectively, while \algname\ processes them in less than 1 hour. 

Figure~\ref{fig:lcc_comparison} shows distributions of local clustering coefficients for nodes in $H$. We see that \algname\ tracks almost perfectly the trend of LCC (barring a few outliers in \texttt{ORK} and \texttt{FR}), as opposed to \kpalg, which often provides estimates that are far off their true value. 
In Table~\ref{tab:lcc_comparison}, we report analytical results. For each dataset, we show values $d$ and $\tau$. 
Note that $d$ (minimum degree of nodes for which \kpalg\ provides estimates) is similar across all datasets, while our degree threshold $\tau$ varies, showing that \algname\ effectively adapts head and tail estimators to the degree distribution of each dataset. 
Table~\ref{tab:lcc_comparison} shows that \algname\ consistently outperforms \kpalg\ in reporting high-quality estimates of LCCs for all considered metrics, but for Pearson correlation coefficient in \texttt{TW} dataset (where our algorithm is more accurate in the estimates).

Most importantly, \algname\ employs \emph{only} the \emph{tail} estimator to estimate LCCs, which consumes $B_{M_t} + B_{A_t} = 0.01m$ edge space, an order of magnitude less compared to \kpalg, which uses $0.1m$ (in expectation).
Therefore, if the final goal is to estimate LCCs, \algname\ can be configured so as to ignore the head component and to fully dedicate its allowed space to the tail estimator.
In Section~\ref{sec:comparison_lcc} of the Appendix, we report runtime and memory statistics for \algname\ instantiated with and without head samples, confirming that, when estimating LCCs, the full space budget can be dedicated to the tail estimator at no additional cost.

% -- Comparison with Wedge Sampling + Runtime
%\TableComparisonWS
\FigurePerf
\emph{Efficiency and Scalability.}
We consider state-of-the-art approaches for estimating NDCC: \texttt{Wedge-Sampling}~\cite{seshadhri2013triadic} (\wedgesampling\ from now on) and \triadalg~\cite{sarpe2025efficient}.
Briefly, \wedgesampling\ works by sampling wedges uniformly at random from the set of wedges centered at nodes in $\InducedNodes{D}{}$, where $\InducedNodes{D}{}$ is the set of nodes having degree in $D$. For each sampled wedge, \wedgesampling\ checks whether it is closed: the fraction of closed wedges over the total number of sampled wedges is reported as an estimate of $\NodeAvgCC{D}{}$.
\triadalg~\cite{sarpe2025efficient} works instead by repeatedly and adaptively sampling sets of edges, based on solving an optimization problem for minimizing the empirical variance. Next, estimates are computed for each set of sampled edges, and then corrected and aggregated to get the final values of NDCC for each interval.  

Notice that \wedgesampling\ and \triadalg\ \emph{are not} streaming algorithms, and require storing the whole graph. Both are aware of exact degree of nodes, as opposed to \algname: the binning of nodes is performed exactly (while for our tail estimator, we need to account for errors in the degree bins). 
%Therefore, the only source of error for \wedgesampling\ stems from sampling wedges, controlled by the number $N_W$ of wedges the algorithm is allowed to sample. 
Despite all being sampling algorithms, \wedgesampling, \triadalg\ and \algname\ differ substantially in both computational model and implementation, and establishing an apples-to-apples comparison is challenging. \wedgesampling\ and \triadalg\ store the entire input, can perform multiple passes over the data, and assume node ids in $[0,n-1]$ to leverage optimizations unavailable to \algname. By contrast, \algname\ is one-pass, uses limited memory, and requires no prior knowledge of $n$ or $m$, nor any preprocessing step such as renormalizing ids.

To demonstrate the power of the streaming model, our experimental goal is to compare the three different methods in accuracy, runtime and memory required.
We run \wedgesampling\ with $N_W = 32k$, which is the maximum value considered in~\cite{seshadhri2013triadic}, and \triadalg\ with sample size fixed to $m / 500$, as suggested in~\cite{sarpe2025efficient}. Note that such parameters have a negligible impact on runtime and memory consumption, which are dominated by the cost of loading and storing the entire graph.
We run \algname\ with memory budget fixed to $5\% m$ and $10\% m$ (in terms of edges stored), setting all the parameters as explained above. 
We consider the RHAS distance between the exact and estimated NDCC distributions for degree intervals of powers of two (i.e., $D_i=[2^i,2^{i+1})$), together with runtime (including time for loading the graph for \wedgesampling\ and \triadalg) and memory (measured as Peak RAM memory, in GB) required by the algorithms. 

Figure~\ref{fig:perf} shows results. 
We see that \algname\ achieves better accuracy than \wedgesampling\ on \texttt{LJ}, \texttt{TW} and \texttt{FR}, and comparable accuracy on \texttt{SK} and \texttt{ORK}, while being faster (except for \texttt{ORK} and \texttt{TW}, in which \algname\ with budget $10\%m$ is slightly slower) and requiring less memory.
\triadalg\ algorithm is consistently the most accurate, due to the sophisticated techniques used. However, it is computationally unfeasible: in our biggest graphs \texttt{TW} and \texttt{FR} it consumes 476GB and 729GB respectively, whereas \wedgesampling\ requires 226GB and 340GB. 
In contrast, these graphs are efficiently processed in less than 50 minutes by \algname, requiring one order of magnitude less memory than \wedgesampling\ and \triadalg\ (precisely, 22GB for \texttt{TW} and 33GB for \texttt{FR} when allowed for $5\%m$ budget), thanks to its small-space one-pass streaming design which avoids to store the whole graph.
In summary, \algname\ is the best algorithm when considering tradeoffs amongst accuracy, runtime, and memory consumption.

\section{Conclusion}
\label{sec:conclusion}
We studied the problem of estimating binned degree-wise clustering coefficients in one-pass graph streams.  
Our proposed method is, to the best of our knowledge, the first algorithm to provide accuracy guarantees for this problem. 
We implemented \algname\ and showed that it yields accurate estimates while being fast and space-efficient, especially in comparison with related methods from the literature.

There are several promising directions for future work.
Since edge triangle counts are a common feature in clustering analysis, it would be compelling to estimate the related distribution with a small-space streaming algorithm.
Another interesting direction is the design of streaming algorithms for higher-order clustering coefficients, which involve large cliques.

\begin{acks}
	Work supported by “National Centre for HPC, Big Data and Quantum Computing” project CN00000013, approved under call M42C – Investimento 1.4 – Avviso “Centri Nazionali”.% – D.D. n. 3138 of 16.12.2021, admitted to funding with MUR decree n. 1031 of 06.17.2022.
	CS acknowledges the support of NSF DMS-2023495, CCF-1740850, 2402572.
\end{acks}

\balance

\bibliographystyle{ACM-Reference-Format}
\bibliography{bibliography}

\clearpage
\nobalance 
\appendix

\section{Tools}
\label{sec:tools}
In the following, we present concentration inequalities used to derive our results. 

\begin{lemma}[Multiplicative Chernoff Bound]
	\label{lemma:mul_chernoff_bound}
	Let $X_1, \dots, X_n$ be independent random variables taking values in $[0, 1]$. Let $X = \sum_{i=1}^{n} X_i$, and let $\mu = \Exp{X}$. Then, for any $\delta \in [0, 1]$, 
	$$
	\Prob{ \left| X - \mu \right| \ge \delta \mu } \leq 2 \exp \left( \frac{-\delta^2 \mu}{3} \right).
	$$
\end{lemma}

\begin{lemma}[Hoeffding's Inequality]
	\label{lemma:hoeffding_ineq}
	Let $X_1, \dots, X_n$ be random variables with $X_i \in [0, 1]$ that are either independent or negatively associated. Let $\bar{X} = \frac{1}{n} \sum_{i = 1}^n X_i$, and $\mu = \Exp{\bar{X}}$. Then, for any $t \ge 0$, 
	$$
	\Prob{ \left| \bar{X} - \mu \right| \ge t } \le 2 \exp \left( -2n t^2 \right)
	$$
\end{lemma}

\begin{lemma}[Bernstein's Inequality]
	\label{lemma:bernstein_ineq}
	Let $X_1, \dots, X_n$ be independent random variables such that $\left| X_i - \Exp{X_i} \right| \le M$ for all $i$. Let $S = \sum_{i = 1}^n X_i$, and $V = \sum_{i = 1}^n \Var{X_i}$. Then, for any $t \ge 0$,
	$$
	\Prob{ \left| S - \Exp{S} \right| \ge t } \le 2 \exp \left( - \frac{t^2}{2V + \frac{2}{3} M t} \right)
	$$
\end{lemma}

\begin{lemma}[One Sided Chernoff Like Bound]
	\label{lemma:one_sided_chernoff_bound}
	Let $X_1, ..., X_n$ be independent random variables such that each $X_i \sim Bern(p)$ with $p < \tau (1 - \varepsilon)$ for some $\varepsilon \in (0, 1/2]$. Let $X = \sum_{i=1}^n X_i$. Then, we have:
	$$
	\Prob{X \ge n \tau} \le \exp \left( \frac{- \varepsilon^2 n \tau}{10} \right).
	$$
\end{lemma}
\begin{proof}
	The proof follows immediately from Lemma 22 of~\cite{bishnu2025towards}, where instead of showing concentration on the average $\frac{1}{n} \sum_{i = 1}^n X_i$ (as done in~\cite{bishnu2025towards}), we consider the sum $X = \sum_{i = 1}^n X_i$.
\end{proof}

In the following, we leverage a result from~\cite{janson2004large} in the context of dependency graphs. For random variables $X_1, \dots, X_n$, a graph $H$ with vertex set $V = \{1, \dots n\}$ is called a \emph{dependency graph}~\cite{dubhashi2009concentration} for variables $X_i$'s if, for any two disjoint subsets $A, B \subset V$ such that there are no edges connecting a vertex in $A$ to a vertex in $B$, the collections $\{ X_i \}_{i \in A}$ and $\{X_j\}_{j \in B}$ are mutually independent. 
\begin{lemma}[Bernstein-type ineq. for dependency graphs {\cite[Thm.~2.3]{janson2004large}}]
	\label{lemma:bernstein_frac_chromatic}
	Let $X_1,\dots, X_n$ be random variables with a dependency graph $H$ and $\mathbb{E}[X_i]=0$ for all $i$. Assume $|X_i|\le b$ for all $i$, and let $V=\sum_{i=1}^n \Var{X_i}$. 
	Let $\chi^*(H)$ be the fractional chromatic number of $H$. Then for all $t\ge 0$,
	\[ 
	\Prob{ \left|\sum_{i=1}^n X_i\right| \ge t } \le 2\exp\!\left(-\frac{8 t^2}{25 \chi^*(H)\left(V+\frac{bt}{3}\right)} \right).
	\]
\end{lemma}
Above, the fractional chromatic number $\chi^*(H)$ is defined as the linear programming relaxation of the standard chromatic number $\chi(H)$. Intuitively, it represents the minimum weight needed to cover the vertices of $H$ with independent sets. 
It holds $\chi^*(H) \le \chi(H) \le d_{max}(H) + 1$, where $d_{max}(H)$ is the maximum degree for graph $H$.

\section{Additional algorithms}
\label{sec:additional_algo}
In this section, we report the pseudocode of subroutines missing from the main text.  

\texttt{SampleEdge} (Algorithm~\ref{alg:sample_edge}) adds the given edge independently to the given sample with the given probability.

\texttt{CTReservoir} (Algorithm~\ref{alg:count_triangles_reservoir}) is given the incoming edge $\Edge{u}{v}$, the main sample $S$ with its memory budget $M_s$ and counter $s$, the auxiliary sample $A$ with its budget $M_a$ and counter $a$, and the node-triangle map to eventually update. 
First, the algorithm computes the adjacency lists of nodes $u, v$ in both samples $S$ and $A$, deriving the list of neighbors (lines~\ref{line:ct_res_first}-\ref{line:ct_res_last}). Then, for every node $w$ forming a triangle with $u$ and $v$ (line~\ref{line:ct_res_for}), it derives the probability of sampling such triangle based on whether the two edges (when the incoming one arrives) forming the triangle belong both to sample $S$ (line~\ref{line:ct_res_SS}), both to sample $A$ (line~\ref{line:ct_res_AA}), or one in $S$ and the other in $A$ (line~\ref{line:ct_res_SA}). 
Finally, such probability is used to increment the local triangles count (line~\ref{line:ct_res_incr}). 

\texttt{SEReservoir} (Algorithm~\ref{alg:sample_edge_reservoir}) is given the edge to sample, the reservoir sample $S$ together with its memory budget $M$ and the number $s$ of elements seen so far (belonging to $S$), and applies the reservoir sampling update: if $S$ is not full yet, it adds the incoming edge to $S$ (line~\ref{line:se_res_not_full}), otherwise it adds it to $S$ with probability $M / s$ and by replacing an edge in $S$ uniformly at random (lines~\ref{line:se_res_first}-\ref{line:se_res_last}). Then, it increments the counter $s$ (line~\ref{line:se_res_incr}).
We assume that \texttt{FlipBiasedCoin}(p) returns $HEAD$ with probability $p$. 

Notice that the correctness of our triangle-counter algorithm follows from the use of reservoir sampling (see~\cite{stefani2017triest,vitter1985random}), and since we are using independent reservoir samples for (head and tail) main and auxiliary samples (see Section~\ref{sec:one_pass_algo}). 

\begin{algorithm}[!b]
	%\footnotesize
	\caption{\texttt{SampleEdge}$\left(\{u, v\}, S, p \right)$}
	\label{alg:sample_edge}
	\LinesNumbered
	\KwIn{edge $\{u, v\}$; edge sample $S$; probability $p$ of sampling; \\
	}
	Let  $S \gets S \cup \{\{u, v\}\}$ with probability $p$\;
\end{algorithm}

\begin{algorithm}[!b]
	\caption{\texttt{CTReservoir}($\{u, v\}, S, M_s, s, A, M_a, a, \hat{T}$)}
	\label{alg:count_triangles_reservoir}
	\LinesNumbered
	\kwInput{
		edge $\{u, v\}$; \\
		main sample $S$; memory budget $M_s$ for sample $S$;\\
		number $s$ of edges belonging to $S$; \\
		auxiliary sample $A$; memory budget $M_a$ for sample $A$; \\
		number $a$ of edges belonging to $A$; \\
		Node-triangle map $\hat{T}$.
	}
	
	%\tcp{Find all neighbors of $u$ and $v$ across both samples}
	Let $Adj_S(u), Adj_S(v)$ be the adjacency lists of  $u, v$ in sample $S$\label{line:ct_res_first}\;
	Let $Adj_A(u), Adj_A(v)$ be the adjacency lists of  $u, v$ in sample $A$\;
	Let $N(u) \gets Adj_S(u) \cup Adj_A(u)$\;
	Let $N(v) \gets Adj_S(v) \cup Adj_A(v)$\;\label{line:ct_res_last}
	
	%\tcp{Iterate over triangles}
	\For{each node $w \in N(u) \cap N(v)$}{ \label{line:ct_res_for}
		Initialize $\hat{T}[u]$, $\hat{T}[v]$, $\hat{T}[w]$ to zero if not set yet\;
		
		\uIf{$w \in Adj_S(u)$ \textnormal{{AND}} $w \in Adj_S(v)$}{
			$p \gets \min \left( 1, \frac{M_s}{s} \cdot \frac{M_s - 1}{s - 1} \right)$\label{line:ct_res_SS}\;
		}
		\uElseIf{$w \in Adj_A(u)$ \textnormal{AND} $w \in Adj_A(v)$}{
			$p \gets \min \left( 1, \frac{M_a}{a} \cdot \frac{M_a - 1}{a - 1} \right)$ \label{line:ct_res_AA}\;
		}
		\lElse{$p \gets \min \left( 1, \frac{M_s}{s} \right) \cdot \min \left( 1, \frac{M_a}{a} \right)$ \label{line:ct_res_SA}}
		Increment $\hat{T}[u]$, $\hat{T}[v]$, $\hat{T}[w]$ by $1/p$ \label{line:ct_res_incr}\;
	}
\end{algorithm}

\begin{algorithm}[t]
	\caption{\texttt{SEReservoir}$\left(\{u, v\}, S, M, s \right)$}
	\label{alg:sample_edge_reservoir}
	\LinesNumbered
	\kwInput{edge $\{u, v\}$; edge sample $S$; \\
		memory budget $M$ for sample $S$; \\
		number $s$ of edges belonging to $S$ seen so far.
	}
	\lIf{$\CardSet{S} < M$}{$S \longleftarrow S \cup \{ \{u, v\}\}$}\label{line:se_res_not_full}
	\ElseIf{\textnormal{\texttt{FlipBiasedCoin}}$(M / s) == HEAD$}{\label{line:se_res_first}
		$\{\bar{u}, \bar{v}\} \longleftarrow$ edge sampled uniformly at random from $S$\;
		$S \gets (S \setminus \{\{\bar{u}, \bar{v}\}\}) \cup \{\{u, v\}\}$\;\label{line:se_res_last}
	}
	$s \gets s + 1$; \label{line:se_res_incr}\\
\end{algorithm}

\section{Additional Proofs}
\label{sec:additional_proofs}
We report all our derived results together with their proofs. 

\begin{lemma}
	\label{lemma:est_degrees_tail}
	Let $\EstDeg{u}{}$ be the estimated degree for each node $u$ contained in node-degree maps $S_h$ and $S_t$ returned by \algnamefirstpass\ algorithm, using fixed probability $p_t$ for sampling edges. 
	Then:
	\begin{enumerate}[label=(\roman*)]
		\item $\EstDeg{u}{} = \Deg{u}{}$ for each node $u \in S_h$\label{item:lemma_one_first},
	\end{enumerate}
	and, for any $\varepsilon \in (0, 1/2)$, if
	$
	\tau = \bigOmegatilde{\frac{1}{\varepsilon p_t}},
	$
	we have:
	\begin{enumerate}[label=(\roman*)]
		\setcounter{enumi}{1}
		\item if $\Deg{v}{} \leq (1 - \varepsilon) \tau$, then $\EstDeg{v}{} \leq \tau$ for each node $v \in S_t$;\label{item:lemma_one_second}
		\item if $\Deg{v}{} \ge \tau$, then $\EstDeg{v}{} \in (1 \pm \varepsilon) \Deg{v}{}$ for each node $v \in S_t$ with high probability;\label{item:lemma_one_third}
		\item if $\Deg{v}{} \ge \tau$, then $v \in S_t$ with high probability.\label{item:lemma_one_fourh}
	\end{enumerate}
\end{lemma}
\begin{proof}
	\emph{Item~\ref{item:lemma_one_first}.} By design, \algnamefirstpass\ is tracking exactly degrees of nodes in $S_h$, which are sampled via hashing. 
	
	\emph{Item~\ref{item:lemma_one_second}.} 
	Fix a vertex $v$. Let $C_v$ be the ``raw'' counter for node $v$, i.e., the value in $S_t[v]$ before adding back the expected loss. We can express $\Deg{v}{} = C_v + L_v$, where $L_v$ is the number of edges incident to $v$ that we missed. For any $a \ge 0$, we have $\Prob{L_v \ge a} \le (1 - p_t)^a \le \exp \left( -p_t a \right)$. 
	The expected loss $\ell(r)$ is such that $0 \le \ell(r) \le \lceil \frac{1 - p_t}{p_t} \rceil \le \frac{1}{p_t}$. For $\Deg{v}{} \le \left( 1 - \varepsilon \right) \tau$ and the suggested choice of $\tau$:
	$$
	\EstDeg{v}{} = C_v + \ell \left( \EstDeg{v}{} \right) \le d_v + \frac{1}{p_t} \le  \left( 1 - \varepsilon \right) \tau + \varepsilon \tau \le \tau.
	$$
	
	\emph{Item~\ref{item:lemma_one_third}.} 
	Fix a vertex $v$. For the lower tail, we have that if $L_v \le \varepsilon d_v$, which happens with probability $1 - \exp \left( -p_t \varepsilon \Deg{v}{} \right) \ge 1 - \exp \left( -p_t \varepsilon \tau \right)$,  then:
	$$
	\EstDeg{v}{} = C_v + \ell \left( \EstDeg{v}{} \right) \ge C_v = \Deg{v}{} - L_v \ge \Deg{v}{} - \varepsilon \Deg{v}{} = \left( 1 - \varepsilon \right) \Deg{v}{}.
	$$
	Therefore, the lower tail follows by a union bound and the suggested choice of $\tau$. The upper tail holds deterministically:
	$$
	\EstDeg{v}{} = C_v + \ell \left( \EstDeg{v}{} \right)  \le \Deg{v}{} + \frac{1}{p_t} \le \Deg{v}{} + \varepsilon \tau \le \left( 1 + \varepsilon \right) \Deg{v}{}.
	$$
	
	\emph{Item~\ref{item:lemma_one_fourh}.} Fix a vertex $v$ such that $\Deg{v}{} \ge \tau$. Then:
	$$
	\Prob{v \notin S_t} = \Prob{L_v \ge \Deg{v}{}} \le \exp \left(- p_t \Deg{v}{} \right) \le \exp \left( - p_t \tau \right) \le \frac{1}{n^c}. 
	$$
	The item follows by a union bound over nodes $v$ with $\Deg{v}{} \ge \tau$.
	
\end{proof}

\begin{fact}
	\label{fact:degree_bounds}
	Given $\EstDeg{}{}$ such that $\EstDeg{}{} \in \left(1 \pm \varepsilon \right) \Deg{}{}$, and $d \ge \tau$, it holds:
	$$
	{\EstDeg{}{} \choose 2} \in \left[ \left( 1 - \varepsilon \right)^2 - \frac{\varepsilon (1 - \varepsilon)}{\tau - 1}, \left( 1 + \varepsilon \right)^2 + \frac{\varepsilon (1 + \varepsilon) }{\tau - 1} \right] {\Deg{}{} \choose 2}.
	$$
\end{fact}

\begin{corollary}
	\label{cor:degree_bounds_simplified}
	Following Fact~\ref{fact:degree_bounds}, for any $\varepsilon \in (0, 1/2], \tau \ge 4$:
	$$
	{\EstDeg{}{} \choose 2} \in \left( 1 \pm 3 \varepsilon \right) {\Deg{}{} \choose 2}.
	$$
\end{corollary}

\begin{proof}
	From Fact~\ref{fact:degree_bounds}, we analyze the upper and lower bounds, separately. 
	For the upper bound, since $\varepsilon \le 1/2$, we have $\varepsilon^2 \le \varepsilon/2$. Thus:
	$$
	\left( 1 + \varepsilon \right)^2 + \frac{\varepsilon (1 + \varepsilon) }{\tau - 1} 
	= 1 + 2\varepsilon + \varepsilon^2 + \frac{\varepsilon(1 + \varepsilon)}{\tau - 1} 
	\le 1 + 2.5\varepsilon + \frac{1.5\varepsilon}{\tau - 1}.
	$$
	For $\tau \ge 4$, it follows an upper bound of $1 + 3\varepsilon$.
	For the lower bound:
	$$
	\left( 1 - \varepsilon \right)^2 - \frac{\varepsilon (1 - \varepsilon)}{\tau - 1} 
	= 1 - 2\varepsilon + \varepsilon^2 - \frac{\varepsilon(1 - \varepsilon)}{\tau - 1} 
	\ge 1 - 2\varepsilon - \frac{\varepsilon}{\tau - 1}.
	$$
	For $\tau \ge 4$, the expression is $> 1 - 3\varepsilon$. 
\end{proof}
Notice that in our experiments, we have values of $\tau$ from 253 (as-skitter) to 4448 (twitter).

\emph{Quality of Estimates.}
In the following, we report proofs of results analyzing the quality of the estimates for binned-degree clustering coefficients, returned by our proposed (two-pass) approach. Note that such results extend to our one-pass approach (see Section~\ref{sec:one_pass_algo}).

We start by analyzing the head estimator. We prove Theorem~\ref{thm:head_estimator} by showing bounds for NDCC and WDCC separately in Theorem~\ref{thm:ndcc_head_appendix} and~\ref{thm:wdcc_head_appendix}. 
First, we provide concentration for estimates of the number of nodes and the number of triangles. 
Throughout, we assume degree interval $D = [L, U)$, with $U = \Theta(L)$.

\begin{lemma}
	\label{lemma:concentration_C_h}
	Consider a degree interval $D$ and any $\varepsilon \in (0, 1) $. Let $C_h(D)$ be the number of nodes with degree in $D$ in the sample $S_h$. Then, if $p_h = \bigOmegatilde{\frac{1}{\varepsilon^2 \NumNodes{D}{}}}$, we have, with high probability, $C_h(D) \in (1 \pm \varepsilon) p_h \NumNodes{D}{}$.
\end{lemma}
\begin{proof}
	Let $X_u = 1$ if $u \in S_h$ and 0 otherwise. Note that $X_i$'s are independent. We have $C_h(D) = \sum_{u, \Deg{u}{} \in D} X_u$, and $\Exp{C_h(D)} = p_h \NumNodes{D}{}$.
	Result follows from an application of Lemma~\ref{lemma:mul_chernoff_bound}.
\end{proof}

We show concentration bounds also for the number of triangles incident to sampled nodes. 
Recall that $\Delta_V(D)$ is the maximum number of triangles incident to a node with degree in $D$.
\begin{lemma}
	\label{lemma:concentration_T_h}
	Consider degree interval $D$ and any $\varepsilon \in (0, 1) $. Let $T_h(D)$ be the number of triangles incident to nodes with degree in $D$ in the sample $S_h$. Then, if $p_h = \bigOmegatilde{\frac{\Delta_V(D)}{\varepsilon^2 \Triangles{D}{}}}$, we have, with high probability, $T_h(D) \in (1 \pm \varepsilon) p_h \Triangles{D}{}$.
\end{lemma}
\begin{proof}
	Let $X_u = 1$ if $u \in S_h$ and 0 otherwise. Note that $X_i$'s are independent. We have $T_h(D) = \sum_{u, \Deg{u}{} \in D} X_u \LocalTrianglesCard{u}{}$, and $\Exp{T_h(D)} = p_h \Triangles{D}{}$.
	By definition, it holds $X_u t_u \le \Delta_V(D)$, for each $u$. 
	Result follows by an application of Lemma~\ref{lemma:mul_chernoff_bound} to the normalized variables $X_u t_u / \Delta_V(D)$.
\end{proof}

\begin{theorem}
	\label{thm:ndcc_head_appendix}
	Consider degree interval $D = [L, U)$ such that $L < U \le \tau$. For any $\varepsilon, \eta \in (0, 1)$, if
	$$
	p_h = \bigOmegatilde{ \frac{1}{\eta^2} \left( \frac{1}{\NumNodes{D}{}} + \frac{\Delta_V(D)}{\Triangles{D}{}} \right)}  \text{ and }
	$$
	$$
	p_{M_h} p_{A_h} = \bigOmegatilde{\frac{1}{ \varepsilon^2 } \left( \eta^2 + \frac{\Delta_E(D)}{\Triangles{D}{} } \right) },
	$$ 
	then Algorithm \algnameestimate\ outputs a value $\widehat{NDCC}(D)$ such that:
	\begin{align*}
		& \left| \widehat{NDCC}(D) - \NodeAvgCC{D}{} \right| \le \varepsilon \cdot \NodeAvgCC{D}{} +  \eta.
	\end{align*}
	with high probability.
\end{theorem}
\begin{proof}
	Recall that $\EstDeg{u}{} = \Deg{u}{}$ for any $u \in S_h$, since degrees are tracked exactly (Lemma~\ref{lemma:est_degrees_tail}, item~\ref{item:lemma_one_first}).
	For head estimator, Alg. \algnameestimate\ computes estimates as:
	\begin{equation}
		\widehat{NDCC}(D) = \frac{1}{{ \left| \left\{ u \in S_h, \Deg{u}{} \in D \right\} \right| }} \; \sum_{u \in S_h, \Deg{u}{} \in D} \frac{\EstLocalTrianglesCard{u}{}}{{\Deg{u}{} \choose 2} }.
	\end{equation}
	Let $\widehat{C}(D) = \left| \left\{ u \in S_h, \Deg{u}{} \in D \right\} \right|$. 
	Define the quantity: 
	$$
	Y(D) = \frac{1}{ \widehat{C}(D)} \; \sum_{u \in S_h, \Deg{u}{} \in D} \frac{ \LocalTrianglesCard{u}{} } {{\Deg{u}{} \choose 2} }
	$$
	We bound the overall error as follows:
	\allowdisplaybreaks
	\begin{align}
		\label{eq:bound_ndcc}
		& | \widehat{NDCC}(D) - \NodeAvgCC{D}{}| \nonumber \\
		&\; \le \underbrace{| \NodeAvgCC{D}{} - Y(D) |}_{(i): \text{ node sampling error}} + \underbrace{| \widehat{NDCC}(D) - Y(D)|}_{(ii): \text{ triangle estimator error}}.
	\end{align}
	
	First, we bound term $(i)$ of Eq.~\ref{eq:bound_ndcc}. 
	Let $Z_u = \LocalTrianglesCard{u}{} / {\Deg{u}{} \choose 2}$ if $u \in S_h$, and 0 otherwise; we have $Z_u \in [0, 1]$ and $Y(D) = \sum_{u, \Deg{u}{} \in D} Z_u / \widehat{C}(D)$. %Note that $Z_u$'s are independent. 
	Conditioning on the event $\widehat{C}(D) = C \in (1 \pm \eta) p_h \NumNodes{D}{}$ (see Lemma~\ref{lemma:concentration_C_h}), we have that variables $Z_i$'s are negatively associated. That is, $\Prob{Z_u \ne 0 \;|\; \widehat{C}(D) = C} = C / \NumNodes{D}{}$. Therefore:
	\begin{align*}
		& \Exp{Y(D) \;\middle\vert\; \widehat{C}(D) = C} \\
		& = \frac{1}{C} \sum_{u, \Deg{u}{} \in D} \frac{\LocalTrianglesCard{u}{}}{{\Deg{u}{} \choose 2}} \; \Exp{u \in S_h, \Deg{u}{} \in D \;\middle\vert\; \widehat{C}(D) = C } \\ 
		&=  \frac{1}{C} \sum_{u, \Deg{u}{} \in D} \frac{\LocalTrianglesCard{u}{}}{{\Deg{u}{} \choose 2}} \; \frac{C}{\NumNodes{D}{}} = \NodeAvgCC{D}{}.
	\end{align*}
	Applying Hoeffding's inequality (Lemma~\ref{lemma:hoeffding_ineq}) for negative associated variables:
	\begin{align*}
		& \Prob{ \left| Y(D) - \NodeAvgCC{D}{} \right| \ge \eta \;\middle\vert\; \widehat{C}_h = C} \\
		&= \Prob{ \left| \frac{\sum_{u, \Deg{u}{} \in D} Z_u}{C} - \NodeAvgCC{D}{} \right| \ge \eta \;\middle\vert\; \widehat{C}_h = C} \\
		& \le 2 \exp \left( -2 C \eta^2 \right) \le 2 \exp \left( -2 (1 - \eta) p_h \Ccount(D) \eta^2 \right),
	\end{align*}
	which is $\le 1 / n^c$ for $p_h = \bigOmegatilde{ \frac{1}{\eta^2 \NumNodes{D}{}} }$, which bounds the first term of Eq.~\ref{eq:bound_ndcc}.
	Applying a union bound over both events concludes the proof for case 1.
	
	Now, we bound term $(ii)$ of Eq.~\ref{eq:bound_ndcc}, which controls the error of our triangle-estimator. Our goal boils down to bounding:
	\begin{equation}
		\label{eq:item_ii_case_1_proof_NDCC}
		\frac{1}{\widehat{C}(D)} \left| \sum_{u \in S_h, \Deg{u}{} \in D} \frac{\EstLocalTrianglesCard{u}{} - \LocalTrianglesCard{u}{}}{{\Deg{u}{} \choose 2}}  \right|.
	\end{equation}
	Define $\widehat{A}(D), A(D)$ as $\sum_{u \in S_h, \Deg{u}{} \in D} \frac{\EstLocalTrianglesCard{u}{}}{{\Deg{u}{} \choose 2}}$ and $\sum_{u \in S_h, \Deg{u}{} \in D} \frac{\LocalTrianglesCard{u}{}}{{\Deg{u}{} \choose 2}}$, respectively.
	Eq.~\ref{eq:item_ii_case_1_proof_NDCC} can be expressed as $\frac{1}{\widehat{C}(D)} \left| \widehat{A}(D) - A(D) \right|$.
	Let $\mathcal{T}_h(D)$ be the set of triangles incident to nodes in $S_h$, having degree in $D$. We have $|\mathcal{T}_h(D)| = \Triangles{D}{h}$. 
	For each triangle $\Delta$, define:
	$$
	w_\Delta = \sum_{u \in \Delta \cap S_h, \Deg{u}{} \in D} \frac{1}{{\Deg{u}{} \choose 2}} \in \left[ \frac{1}{{U \choose 2}}, \frac{3}{{L \choose 2}} \right]. 
	$$
	We have $\sum_{\Delta \in \mathcal{T}_h(D)} w_\Delta = A(D)$.
	Let $I_\Delta = 1$ if triangle $\Delta$ is counted by Alg. \algnamesecondpass, and 0 otherwise. We can write: 
	$$
	\widehat{A}(D) = \sum_{\Delta \in \mathcal{T}_h(D)} \frac{w_\Delta}{p_\Delta} I_\Delta, \textup{ and } \Exp{\widehat{A}(D)} = A(D).
	$$
	Let $X_\Delta = \frac{w_\Delta I_\Delta}{p_\Delta}$.
	We have $\Var{X_\Delta} \le \frac{9 p_\Delta (1 - p_\Delta)}{{L \choose 2}^2 p_\Delta^2} \le \frac{9}{{L \choose 2}^2 p_A p_M}$, assuming $p_A \le p_M$.
	Define centerd variables $Y_\Delta = X_\Delta - \Exp{X_\Delta}$. We have $Y_\Delta = w_\Delta \left( \frac{I_\Delta}{p_\Delta} - 1 \right) \le \frac{3}{{L \choose 2} p_M p_A}$.
	Note that $\sum_\Delta X_\Delta = \widehat{A}(D)$.
	
	Consider the dependency graph $H$ for variables $X_1, \dots X_{T_h(D)}$. That is, place a vertex in $H$ for every variable $X_i$, and place an edge between $X_{\Delta_i}$ and $X_{\Delta_j}$ if they are dependent, i.e., if triangle $\Delta_i$ and triangle $\Delta_j$ share a sampled edge during the streaming algorithm. 
	
	We provide a refined bound for the fractional chromatic number $\chi^*(H) \le \chi$, which corresponds to the maximum number of triangles (among those in $\mathcal{T}_h(D)$) sharing an edge in the original graph. 
	Fix an edge $e = \{a, b\}$. We have the following cases:
	\begin{myitemize}
		\item if $a \in V_D$ and is sampled in $S_h$, then every triangle adjacent to edge $e$ is incident to $a$. Hence, $\chi \le \Delta_V(D)$. The same holds when $b \in V_D$ and $b \in S_h$.
		\item if $a, b \notin V_D$, then a triangle incident to $e$ is counted when the third endpoint is in $V_D$ and sampled in $S_h$. Hence, the number of distinct triangles through $e$ corresponds to the number of nodes (in $S_h \cap V_D$) forming triangles $e$. Let $t^{3rd}_e$ be the number of triangles adjacent to edge $e$ with the third endpoint in $V_D$. We express $X_e \sim Bin(t^{3rd}_e, p_h)$, as each node is sampled i.i.d. via hashing. Let $\mu_e = t^{3rd}_e p_h$. If $\mu_e = \Omega(\log n)$, by a Chernoff bound we get $\Prob{X_e \ge 2 \mu_e} \le \exp(- \mu_e / 3) \le n^{-c - 2}$. Otherwise, if $\mu_e = o(\log n)$, we have: $\Prob{X_e \ge a} \le \sum_{k \ge a} {t_e^{3rd} \choose k} p_h^k \le  \sum_{k \ge a} \left( \frac{e \mu_e}{k} \right)^k \le 2 \left( \frac{e \mu_e}{a} \right)^a$ by solving the geometric series for $a \ge 2e \mu_e$. Picking $a = \max \left(2e \mu_e ,  O(\log n) \right)$ makes this probability small enough. 
	\end{myitemize}
	Given that $t^{3rd}_e \le \Delta_E(D)$ and by a union bound over $\Theta(n^2)$ many edges, we get $\chi \le \Delta_V(D) + p_h \Delta_E(D) + \log n$ (with high probability)

	Now, we have $\sum_\Delta \Var{Y_\Delta} =  \sum_\Delta \Var{X_\Delta} \le \Triangles{D}{h} \frac{9}{{L \choose 2}^2 p_A p_M}$.
	By defining $Y = \sum_\Delta Y_\Delta$, we can apply Lemma~\ref{lemma:bernstein_frac_chromatic}:
	\allowdisplaybreaks
	\begin{align*}
		& \Prob{ \left| \widehat{A}(D) - A(D) \right| \ge \varepsilon A(D)} = \Prob{ \left| Y \right| \ge \varepsilon A(D)} \\
		& \le 2 \exp \left( \frac{-8 \varepsilon^2 A^2(D) }{25 \chi \left( \frac{9 \Triangles{D}{h}}{{L \choose 2}^2 p_M p_A} + \frac{ \varepsilon A(D)}{{L \choose 2} p_M p_A} \right)} \right),
	\end{align*}
	which is $\le 1/n^c$ given that $A(D) \in T_h(D) \left[ 1/{U \choose 2}, 3/{L \choose 2} \right]$, $U/L = \Theta(1)$, for $p_M p_A = \bigOmegatilde{\frac{\chi}{\varepsilon^2 T_h(D)}}$. 
	We then use Lemma~\ref{lemma:concentration_T_h} to obtain concetration for $T_h(D)$ around $p_h \Triangles{D}{}$. Finally, we use the fact that $1 / \NumNodes{D}{} \le \Delta_V(D) / \Triangles{D}{}$ to get $p_M p_A = \bigOmegatilde{\frac{1}{\varepsilon^2} \left( \eta^2 + \frac{\Delta_E(D)}{\varepsilon^2 \Triangles{D}{}}} \right)$.
	Thus, we showed:
	\begin{align*}
		& \left| \widehat{NDCC}(D) - Y(D) \right| = \frac{1}{\widehat{C}(D)} \left| \widehat{A}(D) - A(D) \right| \le \frac{\varepsilon A(D)}{\widehat{C}(D)} = \varepsilon Y(D),
	\end{align*}
	with high probability. Combining item $(i), (ii)$ we proved, for $D = [L, U)$ with $L \le U \le \tau$, that: 
	\begin{align*}
		& \left| \widehat{NDCC}(D) - \NodeAvgCC{D}{} \right| \le \eta + \varepsilon Y(D) \le \eta + \varepsilon (\NodeAvgCC{D}{} + \eta) \\
		& = \varepsilon \cdot \NodeAvgCC{D}{} + (1 + \varepsilon) \eta.
	\end{align*}
	The final result follows then by rescaling $\eta$.
	
\end{proof}

\begin{theorem}
	\label{thm:wdcc_head_appendix}
	Consider a degree interval $D = [L, U)$ such that $L < U \le \tau$. For any $\varepsilon, \eta \in (0, 1)$, if
	$$
	p_h = \bigOmegatilde{ \frac{1}{\eta^2} \left( \frac{1}{\NumNodes{D}{}} + \frac{\Delta_V(D)}{\Triangles{D}{}} \right)}  \text{ and }
	$$
	$$
	p_{M_h} p_{A_h} = \bigOmegatilde{\frac{1}{ \varepsilon^2 } \left( \eta^2 + \frac{\Delta_E(D)}{\Triangles{D}{} } \right) },
	$$ 
	then Algorithm \algnameestimate\ outputs a value $\widehat{WDCC}(D)$ such that:
	\begin{align*}
		& \left| \widehat{WDCC}(D) - \WedgeAvgCC{D}{} \right| \le \varepsilon \cdot \WedgeAvgCC{D}{} +  \eta.
	\end{align*}
\end{theorem}
\begin{proof}
	Recall that $\EstDeg{u}{} = \Deg{u}{}$ for any $u \in S_h$, since degrees are tracked exactly (Lemma~\ref{lemma:est_degrees_tail}, item~\ref{item:lemma_one_first}).
	For the head estimator, Alg. \algnameestimate\ computes:
	\begin{equation}
		\widehat{WDCC}(D) = \frac{ \sum_{u \in S_h, \Deg{u}{} \in D} \EstLocalTrianglesCard{u}{}} {\sum_{u \in S_h, \Deg{u}{} \in D} {{\Deg{u}{} \choose 2}}}.
	\end{equation}
	Define the quantity using true triangle counts: 
	$$
	Y(D) = \frac{ \sum_{u \in S_h, \Deg{u}{} \in D} \LocalTrianglesCard{u}{}} {\sum_{u \in S_h, \Deg{u}{} \in D} {{\Deg{u}{} \choose 2}}}.
	$$
	We bound the overall error as follows:
	\begin{align}
		\label{eq:bound_wdcc}
		& | \widehat{WDCC}(D) - \WedgeAvgCC{D}{}| \nonumber \\
		&\; \le \underbrace{| \WedgeAvgCC{D}{} - Y(D) |}_{(i): \text{ node sampling error}} + \underbrace{| \widehat{WDCC}(D) - Y(D)|}_{(ii): \text{ triangle estimator error}}.
	\end{align}
	
	First, we bound term $(i)$ of Eq.~\ref{eq:bound_wdcc}. 
	Let $X_u = 1$ if $u \in S_h$ and 0 otherwise. 
	For the upper tail, we have:
	\begin{align*}
		& \Prob{Y(D) \ge \WedgeAvgCC{D}{} + \eta} \\
		& = \Prob{\sum_{u, d_u \in D} X_u \LocalTrianglesCard{u}{} \ge \left( \WedgeAvgCC{D}{} + \eta \right) \sum_{u, d_u \in D} X_u {\Deg{u}{} \choose 2}} \\
		& = \Prob{\sum_{u, d_u \in D} Z_u \ge 0},
	\end{align*}
	by defining $Z_u = X_u \left( \LocalTrianglesCard{u}{} - \left(  \WedgeAvgCC{D}{} + \eta \right) {d_u \choose 2} \right)$. Note that $Z_u$'s are independent. 
	Let $Z = \sum_{u, d_u \in D} Z_u$. We have:
	\begin{align*}
		& \Exp{Z} = p_h \sum_{u, d_u \in D} \left( \LocalTrianglesCard{u}{} - \WedgeAvgCC{D}{} {d_u \choose 2} - \eta {d_u \choose 2} \right) \\
		& = - \eta \sum_{u, d_u \in D} p_h {d_u \choose 2} \ge - \eta \NumNodes{D}{} p_h {U \choose 2}, 
	\end{align*}
	as $D = [L, U)$.
	Moreover, each $|Z_u|$ is upper bounded by ${d_u \choose 2} \le {U \choose 2}$, and $\Var{Z} \le p_h \NumNodes{D}{} {U \choose 2}^2$. By Bernstein's inequality (Lemma~\ref{lemma:bernstein_ineq}): 
	\begin{align*}
		\Prob{Z - \Exp{Z} \ge -\Exp{Z}} \le 2 \exp \left( \frac{- p_h^2 \eta^2 {L \choose 2}^2 \left(\NumNodes{D}{}\right)^2} {2 p_h \NumNodes{D}{} {U \choose 2}^2 + \frac{2}{3} \NumNodes{D}{} {U \choose 2}^2 p_h \eta }  \right),
	\end{align*}
	which is $\le 1/n^c$ for $p_h = \bigOmegatilde{\frac{1}{\eta^2 \NumNodes{D}{}}}$, when $U = \Theta(L)$. An equivalent proof applies for the lower tail $Y(D) \le \WedgeAvgCC{D}{} - \eta$, after which we apply a union bound over the two tail bounds. 
	
	Now, we bound term $(ii)$ of Eq.~\ref{eq:bound_wdcc}, which controls the error of our triangle-estimator. 
	The proof is equivalent to the one for Theorem~\ref{thm:ndcc_head_appendix} by defining variables accordingly. 
	Thus:
	\begin{align*}
		& \left| \widehat{WDCC}(D) - Y(D) \right| = \frac{1}{\sum_{u, d_u \in D} {d_u \choose 2}} \left| \widehat{A}(D) - A(D) \right| \le \varepsilon Y(D),
	\end{align*}
	with high probability. Combining item $(i), (ii)$ we proved, for $D = [L, U)$ with $L \le U \le \tau$, that: 
	\begin{align*}
		& \left| \widehat{WDCC}(D) - \WedgeAvgCC{D}{} \right| \le \eta + \varepsilon Y(D) \le \eta + \varepsilon (\WedgeAvgCC{D}{} + \eta) \\
		& = \varepsilon \cdot \WedgeAvgCC{D}{} + (1 + \varepsilon) \eta.
	\end{align*}
	The final result from the theorem follows then by rescaling $\eta$.
	
\end{proof}

Next, we prove bounds for our tail estimator, which provides estimates when, for $D = [L, U)$, we have $\tau \le L$. Combining Theorems~\ref{thm:ndcc_tail_appendix} and~\ref{thm:wdcc_tail_appendix} yields Theorem~\ref{thm:tail_estimator}. 
\begin{theorem}
	\label{thm:ndcc_tail_appendix}
	Consider a degree interval $D = [L, U)$ such that $\tau \le L < U$. 
	For any $\varepsilon \in (0, 1/2)$, define $D^+ = [L (1 - \varepsilon), U(1 + \varepsilon)]$, and $D^- = [L (1 + \varepsilon), U(1 - \varepsilon)]$.
	For the quantity $\NodeAvgCC{D}{} = \frac{S_{LCC}(D)}{C(D)}$, we have, if
	$$
	\tau = \bigOmegatilde{ \frac{1}{\varepsilon p_t} }  \text{ and } p_{M_t} p_{A_t} = \bigOmegatilde{\frac{\Delta_E \left( D^+ \right)}{ \varepsilon^2 \Triangles{D^-}{}}},
	$$ 
	then Algorithm \algnameestimate\ outputs a value $\widehat{NDCC}(D)$ such that:
	\begin{align*}
		& \widehat{NDCC}(D) \in \left[ \left(1 - \varepsilon \right) \; \frac{S_{LCC}(D^-)}{\NumNodes{D^+}{}},  \left( 1 + \varepsilon \right) \; \frac{S_{LCC}(D^+)}{\NumNodes{D^-}{}} \right],
	\end{align*}
\end{theorem}

\begin{proof}
	For the tail estimator, it holds (w.h.p.) $\EstDeg{v}{} \in (1 \pm \varepsilon) \Deg{v}{}$ by Lemma~\ref{lemma:est_degrees_tail} item~\ref{item:lemma_one_third}, for $\tau = \bigOmegatilde{\frac{1}{\varepsilon p_t}}$.
	For the tail estimator, Alg. \algnameestimate\ computes:
	\begin{equation}
		\widehat{NDCC}(D) = \frac{ \sum_{v \in S_t, \lceil \EstDeg{v}{} \rceil \in D} \frac{\EstLocalTrianglesCard{v}{}}{{\lceil \EstDeg{v}{} \rceil \choose 2}} }{ \left| \left\{ v \in S_t, \lceil \EstDeg{v}{} \rceil \in D \right\} \right| }
		\overset{(i)}{=} \frac{ \sum_{v, \lceil \EstDeg{v}{} \rceil \in D} \frac{\EstLocalTrianglesCard{v}{}}{{\lceil \EstDeg{v}{} \rceil \choose 2}} }{ \left| \left\{ v, \lceil \EstDeg{v}{} \rceil \in D \right\} \right| },
	\end{equation}
	where $(i)$ follows from Lemma~\ref{lemma:est_degrees_tail}, item~\ref{item:lemma_one_fourh} w.h.p. Define:
	$$
	Y(D) = \frac{ \sum_{v, \lceil \EstDeg{v}{} \rceil \in D} \frac{\LocalTrianglesCard{v}{}}{{\lceil \EstDeg{v}{} \rceil \choose 2}} }{ \left| \left\{ v, \lceil \EstDeg{v}{} \rceil \in D \right\} \right| },
	$$	
	and $\widehat{X}(D)$ using true degrees (in the binomial coefficients):
	\begin{equation*}
		\widehat{X}(D) = \frac{ \sum_{v, \lceil \EstDeg{v}{} \rceil \in D} \frac{\LocalTrianglesCard{v}{}}{{ \Deg{v}{} \choose 2}} }{ \left| \left\{ v, \lceil \EstDeg{v}{} \rceil \in D \right\} \right| } 
		= \frac{\widehat{S}_{LCC}(D)}{\widehat{C}(D)};.
	\end{equation*}
	From Corollary~\ref{cor:degree_bounds_simplified}, for $\tau \ge 4$, we have:
	\begin{equation}\label{eq:tail_degrees_ndcc}
		Y(D) \in \frac{1}{1 \mp 3 \varepsilon} \widehat{X}(D).
	\end{equation}
	Define $X(D)$ using true degree over sum and count as:
	$$
	{X}(D) = \frac{ \sum_{v, \Deg{v}{} \in D} \frac{\LocalTrianglesCard{v}{}}{{ \Deg{v}{} \choose 2}} }{ \left| \left\{ v, \Deg{v}{} \in D \right\} \right| } = \frac{S_{LCC}(D)}{C(D)}; 
	$$
	Now, we have the same terms in $\widehat{X}(D), X(D)$, but ranging on a different domain due to slacked (approximated) degrees. 
	To explicit the domain differences, we consider complementary cdf to leverage monotonicity; that is, define $\widehat{F}(d) = \sum_{v, \lceil \EstDeg{v}{} \rceil \ge d} \LocalTrianglesCard{v}{} / {\Deg{v}{} \choose 2}$ and $F(d) = \sum_{v, \Deg{v}{} \ge d} \LocalTrianglesCard{v}{} / {\Deg{v}{} \choose 2}$.
	We have $\widehat{S}_{LCC}([L, U)) = \widehat{F}(L) - \widehat{F}(U)$, and $S_{LCC}([L, U)) = F(L) - F(U)$.
	By Lemma~\ref{lemma:est_degrees_tail}, items~\ref{item:lemma_one_second} and~\ref{item:lemma_one_third}, and by rescaling $\varepsilon$'s, we have (w.h.p.):
	\begin{align*}
		\widehat{F}(d) & \le \sum_{v, \Deg{v}{} \ge \frac{d}{1 + \varepsilon}} \frac{\LocalTrianglesCard{v}{}}{{\Deg{v}{} \choose 2}} \le \sum_{v, \Deg{v}{} \ge (1 - \varepsilon )d} \frac{\LocalTrianglesCard{v}{}}{{\Deg{v}{} \choose 2}} = F\!\left( (1 - \varepsilon)d \right),
	\end{align*}
	and, similarly:
	\begin{align*}
		\widehat{F}(d) & \ge \sum_{v, \Deg{v}{} \ge \frac{d}{1 - \varepsilon}} \frac{\LocalTrianglesCard{v}{}}{{\Deg{v}{} \choose 2}} \ge \sum_{v, \Deg{v}{} \ge (1 + \varepsilon )d} \frac{\LocalTrianglesCard{v}{}}{{\Deg{v}{} \choose 2}} = F\!\left((1 + \varepsilon) d\right).
	\end{align*}
	Consequently, for our estimator:
	\begin{align*}
		\widehat{S}_{LCC}([L, U)) &= \widehat{F}(L) - \widehat{F}(U) \le {F}((1 - \varepsilon)L) - {F}((1 + \varepsilon)U) \\
		& = S_{LCC} \left( [ (1 - \varepsilon)L, (1 + \varepsilon)U ] \right),
	\end{align*}
	and, similarly:
	\begin{align*}
		\widehat{S}_{LCC}([L, U)) &= \widehat{F}(L) - \widehat{F}(U) \ge {F}((1 + \varepsilon)L) - {F}((1 - \varepsilon)U) \\
		& = S_{LCC} \left( [ (1 + \varepsilon)L, (1 - \varepsilon)U ] \right).
	\end{align*}
	By defining $D^+, D^-$ as in the statement of the theorem, we have $\widehat{S}_{LCC}(D) \in [S_{LCC}(D^-), S_{LCC}(D^+)]$.
	With the same exact steps, we can derive $\widehat{C}(D) \in [C(D^-), C(D^+)]$. Combining with Eq.~\ref{eq:tail_degrees_ndcc}, we obtain:
	$$
	Y(D) \in \frac{1}{ 1 \mp 3\varepsilon } \left[ \frac{S_{LCC}(D^-)}{C(D^+)}, \frac{S_{LCC}(D^+)}{C(D^-)} \right].
	$$
	Now, we bound $\widehat{NDCC}(D)$ with respect to $Y(D)$, which corresponds to the triangle-estimator error. 
	This proof is equivalent to the one in Theorem~\ref{thm:ndcc_head_appendix}. We just need to define $w_\Delta$'s accordingly and notice that $w_\Delta \in \left[ \frac{1}{ {U \left(1 + \varepsilon \right) \choose 2} }, \frac{3} {{ L \left( 1 - \varepsilon \right) \choose 2}} \right]$. Applying the same steps leads to $ \widehat{NDCC}(D) \in (1 \pm \varepsilon) \; Y(D) $ w.h.p. by setting $p_M p_A = \bigOmegatilde{\frac{\Delta_E \left(D^+ \right)}{\varepsilon^2 T(D^-)}}$. 
	Notice that $T_t(D^-) = T(D^-)$ with high probability by Lemma~\ref{lemma:est_degrees_tail} item~\ref{item:lemma_one_fourh}.
	Overall, we can write:
	\begin{align*}
		& \widehat{NDCC}(D) \in \left[ (1 - \varepsilon) \; \frac{S_{LCC}(D^-)}{\NumNodes{D^+}{}},  (1 + \varepsilon) \; \frac{S_{LCC}(D^+)}{\NumNodes{D^-}{}} \right],
	\end{align*}
	as $\varepsilon \in (0, 1/2)$, and by rescaling, concluding the proof. 
\end{proof}

\begin{theorem}
	\label{thm:wdcc_tail_appendix}
	Consider a degree interval $D = [L, U)$ such that $\tau \le L < U$. 
	For any $\varepsilon \in (0, 1/2)$, define $D^+ = [L (1 - \varepsilon), U(1 + \varepsilon)]$, and $D^- = [L (1 + \varepsilon), U(1 - \varepsilon)]$.
	For the quantity $\WedgeAvgCC{D}{} = \frac{\Triangles{D}{}}{\Wedges{D}{}}$, we have, if
	$$
	\tau = \bigOmegatilde{ \frac{1}{\varepsilon p_t} }  \text{ and } p_{M_t} p_{A_t} = \bigOmegatilde{\frac{\Delta_E \left(D^+ \right)}{ \varepsilon^2 \Triangles{D^-}{}}},
	$$ 
	then Algorithm \algnameestimate\ outputs a value $\widehat{WDCC}(D)$ such that:
	\begin{align*}
		& \widehat{WDCC}(D) \in \left[ \left(1 - \varepsilon \right) \; \frac{\Triangles{D^-}{}}{\Wedges{D^+}{}},  \left( 1 + \varepsilon \right) \; \frac{\Triangles{D^+}{}}{\Wedges{D^-}{}} \right],
	\end{align*}
\end{theorem}
\begin{proof}
	The proof follows equivalent steps of the proof of Thm.~\ref{thm:ndcc_tail_appendix}, by applying complementary cdf on $\Triangles{D}{}$ and on $\Wedges{D}{}$.
\end{proof}

Now, we prove our result on local clustering coefficients for sufficiently high-degree nodes (i.e., from our tail estimator).
\begin{manualtheorem}{Theorem~\ref{thm:tail_LCC}}
	Let $v \in S_t$ be a node such that $d_v \ge \tau$ and $\LocalCC{v}{} > 0$. For any $\varepsilon \in (0,1/2)$ and any $\eta \in (0,\LocalCC{v}{}]$, if
	$$
	\tau = \widetilde{\Omega}\!\left(\frac{1}{\varepsilon p_t}\right), \text{ and } p_{M_t}p_{A_t} = \widetilde{\Omega}\!\left(\frac{\Delta_E(v)\LocalCC{v}{}}{\eta^2 \tau^2} \right),
	$$
	then, with high probability,
	$$
	\left|\widehat{LCC}(v)-\LocalCC{v}{}\right| \le \varepsilon \LocalCC{v}{} + \eta.
	$$
\end{manualtheorem}
\begin{proof}
	First, a node $v$ with degree $d_v \ge \tau$ is included in $S_t$ for the above choice of $\tau$ with high probability, from Lemma~\ref{lemma:est_degrees_tail}.
	Our tail estimator computes, for any such node:
	$$
	\widehat{LCC}(v) = \frac{\EstLocalTrianglesCard{v}{}}{{\lceil \EstDeg{v}{} \rceil \choose 2} } \overset{(i)}{\in} \frac{1}{1 \mp \varepsilon} \frac{\EstLocalTrianglesCard{v}{}}{{ \Deg{v}{} \choose 2} } \coloneqq \frac{1}{1 \mp \varepsilon} Y(v),
	$$
	where $(i)$ follows from Corollary~\ref{cor:degree_bounds_simplified}, and by rescaling $\varepsilon$.
	For vertex $v$, let $X_i(v) = 1 / p_{\Delta_i}$ if $i$-th triangle is incident to $v$ and is counted by the algorithm. We have $\Exp{X_i(v)} = 1$, so that $\EstLocalTrianglesCard{v}{} = \sum_i X_i(v)$ and $\Exp{\EstLocalTrianglesCard{v}{}} = \LocalTrianglesCard{v}{}$. Moreover, $\Var{X_i(v)} \le 1 / p_{\Delta_i} \le 1 / p_A p_M$, assuming $p_A \le p_M$. We can build a dependency graph for variables $X_i(v)$ for all $i$, with fractional chromatic number $\le \Delta_E(v)$, which is the maximum number of triangles (among those incident to $v$) adjacent to an edge.
	Define centered variables $Z_i(v) = X_i(v) - \Exp{X_i(v)}$, and $Z = \sum_i Z_i(v)$. By Lemma~\ref{lemma:bernstein_frac_chromatic}:
	\allowdisplaybreaks
	\begin{align*}
		& \Prob{ \left| Y(v) - LCC(v) \right| \ge \eta} =  \Prob{ \left| Z \right| \ge \eta \; {d_v \choose 2}} \\
		& \le 2 \exp \left( \frac{-8 \eta^2 {d_v \choose 2}^2}{25 \Delta_E(v) \left( \frac{\LocalTrianglesCard{v}{}}{p_A p_M} + \frac{\eta {d_v \choose 2}}{3 p_A p_M} \right)} \right),
	\end{align*}
	which is $\le 1/n^c$ for $p_A p_M = \bigOmegatilde{ \frac{\Delta_E(v) LCC(v)}{\eta^2 \tau^2}}$ for $\eta \le LCC(v)$, and since $d_v \ge \tau$. Overall:
	\begin{align*}
		& \left| \widehat{LCC}(v) - LCC(v) \right| \le \left| \widehat{LCC}(v) - Y(v) \right| + \left| Y(v) - LCC(v) \right| \\
		& \le \varepsilon Y(v) + \eta \le \varepsilon LCC(v) + \eta (\varepsilon + 1).
	\end{align*}
	Main result follows by rescaling $\varepsilon$ and $\eta$. 
	
\end{proof}

Now, we report the proof of the equivalence of the tail triangle estimator in the two-pass and in the one-pass algorithms. 
\begin{manualtheorem}{Lemma~\ref{lemma:triangle_tail_equiv}}
	The estimators produced by the one-pass and two-pass algorithms are identical for the node-degree map $S_t$ and, if $p_{M_t}=p_{A_t}$, also for the node-triangle map $\hat{T}_t$.
\end{manualtheorem}

\begin{proof}
	
	We prove the lemma by coupling the randomness of the one-pass and two-pass algorithms. 
	For each edge $e = \{u, v\}$ in the stream, independently pre-generate Bernoulli random variables $X_u, X_v$ with parameter $p_t$, used to determine whether $u, v \in S_t$, and a Bernoulli random variable $Y_e$ with parameter $p_{M_t}=p_{A_t}$, used to determine whether $e$ is sampled for tail triangle counting. 
	Consider the same variables $\{ X_u, X_v, Y_e \}_{e = \{u, v\}}$ for both algorithms. 
	
	Since both algorithms use the same Bernoulli variables $X_u, X_v$ to decide whether nodes $u$ and $v$ belong to $S_t$, the resulting set $S_t$ is identical in the two algorithms.
	
	Now, fix any node $v \in S_t$, and consider any triangle of the type
	$
	\Delta = \{\Edge{v}{a}, \Edge{v}{b}, \Edge{a}{b}\}
	$
	incident to $v$. Let $e_1,e_2$ be the two edges of $\Delta$ that appear before the closing edge in the stream.
	
	We first analyze the two-pass algorithm. After the first pass, the set $S_t$ is fixed. An edge $e$ is inserted into the tail main sample if it is incident to a node of $S_t$ and $Y_e=1$; otherwise, if $Y_e=1$, it is inserted into the tail auxiliary sample. Since $v \in S_t$, at least one of the two edges $e_1,e_2$ is incident to $v$. Hence, whenever that edge is sampled, it must belong to the tail main sample. Therefore, the event that $\Delta$ is counted by the two-pass algorithm is exactly $ Y_{e_1}=1$ and $Y_{e_2}=1.$
	
	Consider now the one-pass algorithm. At the arrival time of an edge $e$, the algorithm decides, based on the current stream prefix, whether $e$ is eligible for the tail main sample or for the tail auxiliary sample. Since $p_{M_t}=p_{A_t}$, in both cases the decision is made using the same variable $Y_e$. Thus, independently of whether $e$ is currently being sampled in the main or auxiliary sample, the event that $e$ is sampled is $Y_e=1$.
	The one-pass algorithm counts a triangle whenever its two non-closing edges are sampled, even if both are stored in the auxiliary sample. Hence, for the same triangle $\Delta$, the event that it is counted is again exactly $Y_{e_1}=1$ and $Y_{e_2}=1.$
	
	Thus, for every triangle incident to $v \in S_t$, the counting event is identical in the one-pass and two-pass algorithms. Summing over all such triangles shows that the number of counted triangles incident to $v$ is the same in the two algorithms, for every $v \in S_t$.
\end{proof}

Finally, we present the proofs about space and time complexities of our proposed one-pass algorithm.

\begin{manualtheorem}{Theorem~\ref{thm:space}}
	Let $\mathcal{D} = \{ D_1, D_2, \dots \}$ be a set of degree intervals.
	There exists a \emph{one-pass} streaming algorithm that, for all $D = [L, U) \in \mathcal{D}$ \emph{simultaneously}, computes approximations of $\NodeAvgCC{D}{}$ and of $\WedgeAvgCC{D}{}$ using \emph{expected} space:
	$$
	\bigOtilde{\frac{n}{\eta^2} \; \max_{D \colon U \le \tau } \left( \frac{1}{\NumNodes{D}{}} + \frac{\Delta_V(D)}{\Triangles{D}{}} \right) + \frac{m}{\varepsilon \tau}} \text{ nodes, and}
	$$
	$$
	\bigOtilde{\frac{m}{\varepsilon} \left( \eta + \max_D \sqrt{\frac{\Delta_E(D^+)}{\Triangles{D^-}{}}} \; \right) } \text{edges, }
	$$
	where $\tau = \frac{1}{\varepsilon p_t}$.
\end{manualtheorem}

\begin{proof}
	The overall pipeline is as follows: first, run our one-pass algorithm to gather node-degree and node-triangle estimates. 
	Parameters of the one-pass algorithm should be fixed accordingly to $\mathcal{D}$: if all $D = [L, U) \in \mathcal{D}$ are such that $\tau \le L$, then all intervals are estimated by using the tail estimator (i.e., there is no point in using the head estimator and hence one can set $p_h = 0$). In general, $p_h$ should consider only degrees $D = [L, U) \in \mathcal{D} \colon U \le \tau$.  
	Note that $\tau = 1 / (\varepsilon p_t)$ can be computed upfront, given our fixed parameter $p_t$. 
	Upon running the one-pass algorithm, for every degree interval $D = [L, U) \in \mathcal{D}$ run Algorithm \algnameestimate\ to obtain estimates $\widehat{NDCC}(D)$ and $\widehat{WDCC}(D)$, either using the head estimator (if $U \le \tau$) or the tail estimator (if $\tau \ge L $). 
	
	Consider a fixed interval $D \in \mathcal{D}$. 
	We plug bounds derived in Thm.~\ref{thm:ndcc_head_appendix},~\ref{thm:wdcc_head_appendix},~\ref{thm:ndcc_tail_appendix},~\ref{thm:wdcc_tail_appendix}.
	The space used by the algorithm accounts for the node-degree map $S_h$ (head) and $S_t$ (tail), and the subgraph stored for counting triangles.
	The expected number of nodes in $S_h$ is $n \cdot p_h$, and the expected number of nodes in $S_t$ is $\sum_v \left( 1 - \left( 1 - p_t \right)^{\Deg{v}{}} \right) \le \sum_v \min \left( 1, p_t \Deg{v}{} \right)$. We consider the looser $m p_t = m / (\varepsilon \tau)$ for the sake of interpretability of our bound.
	The expected number of edges for counting triangles is bounded by $m \left( p_{M_h} + p_{A_h} + p_{M_t} + p_{A_t} \right)$, subject to constraints
	$
	p_{M_h} p_{A_h} = \bigOmegatilde{\frac{1}{ \varepsilon^2 } \left( \eta^2 + \frac{\Delta_E(D)}{\Triangles{D}{} } \right) },
	$ 
	and $p_{M_t} p_{A_t} = \bigOmegatilde{\frac{\Delta_E(D^+)}{\varepsilon^2 \Triangles{D^-}{}}}$.
	Thus, the quantity $p_{M_h} + p_{A_h} + p_{M_t} + p_{A_t}$ is minimised when $p_{M_h} = p_{A_h} = \bigOtilde{\sqrt{\frac{1}{\varepsilon^2} \left( \eta^2 + \frac{\Delta_E(D)}{\Triangles{D}{}}\right)}}$ and $p_{M_t} = p_{A_t} = \bigOtilde{\sqrt{\frac{\Delta_E(D^+)}{\varepsilon^2 \Triangles{D^-}{}}}}$. We have $p_h = \bigOmegatilde{\frac{1}{\eta^2} \left( \frac{1}{\NumNodes{D}{}} + \frac{\Delta_V(D)}{\Triangles{D}{}} \right)}$, and $\frac{\Triangles{D}{}}{\NumNodes{D}{}} \le \Delta_V(D)$. Plugging and re-arranging, yields the final bound for a single degree interval $D \in \mathcal{D}$. The result follows by a union bound on $|\mathcal{D}| = \bigO{n^2}$ and considering the interval $D$ which maximizes each quantity depending on it. Again, by prioritizing clarity of our bounds, we do not make sepation between head and tail intervals in our edge space.
	
\end{proof}

\begin{manualtheorem}{Proposition~\ref{prop:time}}
	Given an input stream $\Sigma$ for graph $G$ with $m$ edges and maximum degree $d_{max}(G)$, our one-pass algorithm processes the stream in $\bigO{m \cdot d_{max}(G)}$ time.  
\end{manualtheorem}
\begin{proof}
	For each incoming edge $\Edge{u}{v}$ in the stream $\Sigma$, the most expensive step is the computation of common neighbors when counting triangles. 
	For edge $\Edge{u}{v}$, computing common neighbors takes $\bigO{\min(\Deg{u}{}, \Deg{v}{})}$, where adjacency lists of subgraph are implemented by hash sets, so we can iterate over the smaller set and use set membership in the largest one. 
	Iterating over all the edges, it trivially follows $\sum_{\Edge{u}{v} \in \Sigma} \bigO{\min(\Deg{u}{}, \Deg{v}{})} = \bigO{m \cdot d_{max}(G)}$, where $d_{max}(G)$ is the maximum degree for graph $G$, and is a rough upper bound as we need to consider degrees in the subgraph induced by sampled edges. 
	All remaining operations, including hash computations, accesses to the node-degree maps $S_h, S_t$, and insertions into $M_h, M_t, A_h,$ and $A_t$, take constant time per operation under standard hashing assumptions. The final step for correcting tail nodes takes at most $\bigO{m}$ time, and is implemented by guessing the value $r$ through binary search.
\end{proof}

\section{Selecting Parameters for \algname }
\label{sec:params_selection}
In the following, we provide analytical insights into the parameter selection strategy used in our experiments. 
Our algorithm \algname\ allocates memory for counting triangles across four different reservoirs: $B_{M_h}$ for head main sample, $B_{A_h}$ for head auxiliary sample, $B_{M_t}$ for tail main sample, and $B_{A_t}$ for tail auxiliary sample. 

\smallskip
\emph{Head Main and Auxiliary Sample}. Given the probability $p_h$ of sampling head nodes, which determines the sample $S_h$, let $m_{M_h} = \left| \{ \{u, v\} \in E \colon u \in S_h \text{ or } v \in S_h \}\right|$ (number of candidates edges for head main sample), and $m_{A_h} = \left| \{ \{u, v\} \in E \colon u \notin S_h \text{ and } v \notin S_h \}\right|$ (number of candidate edges for head auxiliary sample). 
In expectation, we have $m_{M_h} = \left( 2p_h - p_h^2 \right) m$ and $m_{A_h} = \left(1 - p_h \right)^2 m$.

Consider the counting of a triangle $\{ u, v, w\}$ with at least one vertex (say $v$) belongs to $S_h$. We can have:
\begin{myitemize}
	\item $u, w \notin S_h$: depending on the arrival order, the two edges (when the closing one arrives) can be both in the main sample or one in the main and the other in the auxiliary sample. Assuming random stream ordering, the former happens w.p. 1/3, and the latter w.p. 2/3.
	\item $u \in S_h$ or $w \in S_h$: no matter the arrival order, both edges (when the closing one arrives) are in the main sample. 
\end{myitemize}
Thus, for the given triangle, the probability that both edges are in the main sample is $\frac{1}{3} \left( 1 - p_h \right)^2 + \left( 1 - \left(1 - p_h \right)^2 \right) = 1 - \frac{2}{3} \left(1 - p_h \right)^2$, while the probability that one edge is in the main and the other in the auxiliary sample is $ \frac{2}{3} \left( 1 - p_h \right)^2$.

Following previous work~\cite{lee2020temporal,wu2025great}, we consider the simplified variance (i.e., ignoring covariance), as it has been shown that it is strongly correlated with the variance in real-world graphs ($R^2 > 0.99$). 
Define the ``effective'' sampling probabilities $q_M = B_{M_h} / m_{M_h}$ and $q_A = B_{A_h} / m_{A_h}$ from our reservoir samples.
Minimizing the simplified variance implies minimizing the following upper bound on each individual triangle variance\footnote{Technically, the probability is $B_{M_h} \left( B_{M_h} - 1\right) / \left( m_{M_h} \left( m_{M_h} - 1\right) \right)$ rather than $q_M^2$; we assume equivalence for sufficiently large numbers.}:
\begin{equation} \label{eq:variance_head_bound}
	V = \frac{1}{q_M^2} \left( 1 - \frac{2}{3} \left( 1 - p_h \right)^2 \right) + \frac{1}{q_M q_A} \left( \frac{2}{3} \left( 1 - p_h \right)^2 \right).
\end{equation}

Our objective is to minimze $V$ subject to a fixed memory budget for head samples $k_h =  B_{M_h} + B_{A_h}$. 
We know that, in expectation, $B_{M_h} = q_M \cdot m_{M_h}$ and $B_{A_h} = q_A \cdot m_{A_h}$, for fixed probability sampling. We plug in the expected quantity\footnote{The derivation is not formal for the ease of presentation. Standard concentration inequalities (e.g., Hoeffding's, see~\ref{lemma:hoeffding_ineq}) guarantee that the empirical count matches expectations with high probability.} in~\eqref{eq:variance_head_bound}

Let $y = B_{M_h} / k_h$ be the budget allocation for head main sample. Setting $\frac{dV}{dy} = 0$ we obtain a quadratic equation $ay^2 + by +c = 0$ with:
$
a = -\frac{4}{3} + \frac{20}{3}p_h - \frac{10}{3}p_h^2, \
b = \frac{2}{3} - \frac{16}{3}p_h - \frac{16}{3}p_h^2 + 8p_h^3 - 2p_h^4, \
c = \frac{4}{3}p_h + \frac{14}{3}p_h^2 - \frac{16}{3}p_h^3 + \frac{4}{3}p_h^4.
$
% We can obtain a clean formula depending on $p_h$ for setting the budgets $B_{M_h}$ and $B_{A_h}$, given $k_h = B_{M_h} + B_{A_h}$.
Consequently, the optimal parameter selection can be computing deterministically given $p_h$ by solving the above quadratic equation. 
For instance, when $p_h = 0.1$, this selection yields $B_{M_h} \approx 0.56 k_h$ and $B_{A_h} \approx 0.44 k_h$; for $p_h = 0.2$, it shifts to $B_{M_h} \approx 0.65 k_h$ and $B_{A_h} \approx 0.35 k_h$.
%At the extremes, for $p_h = 1$ it sets $B_{M_h} = k_h$ and for $p_h = 0$ it yields $B_{M_h} = 0.5 k_h$

\smallskip
\emph{Tail Main and Auxiliary Sample}.
Given that tail node sampling is degree-dependent, providing an analytical closed-form budget allocation (as done for the head component) is intractable. 
Remember that tail triangles are counted even when their edges arrive prior to the sampling of tail nodes, in which case both edges belong to the tail auxiliary sample (see Section~\ref{sec:one_pass_algo}). 

To resolve the tail budget allocation, we assume that a triangle is equally likely to be captured in one of the three following configurations: (i) entirely within the main sample, (ii) entirely within the auxiliary sample, or (iii) split across both.
By similar calculations, an upper bound on the simplified variance is $\frac{1}{3 q_M^2} + \frac{1}{3 q_M q_A} + \frac{1}{3 q_A^2}$ for $q_M = B_{M_t} / k_t$ and $q_A = B_{A_t} / k_t$ and subject to $k_t = B_{M_t} + B_{A_t}$, which is minimized setting $B_{M_t} = B_{A_t}$, i.e., equally balancing the tail main and auxiliary budgets.

Notice that we set equal memory budgets $B_{M_t} = B_{A_t} = 0.005m$.
Since \algname\ is implemented via reservoir sampling, equal budgets do not correspond to equal sampling probabilities. As motivated in Section~\ref{sec:one_pass_algo}, our choice of parameters is regarded as a heuristic initialization, and is not intended as an instantiation of the condition $p_{M_t} = p_{A_t}$ of Lemma~\ref{lemma:triangle_tail_equiv}.

\section{Additional Experiments}
\label{sec:additional_experiments}

In this section, we first assess empirically the quality of bounds obtained by our tail estimator. 
Moreover, we analyze the extra number of counters necessary for guaranteeing correct estimates for our final one-pass algorithm. 
Then, we expand on the empirical choice of parameters for our algorithm \algname, and finally we report results missing in the main text due to space issues. 

\subsection{Empirical Bounds for Tail Estimator}
\label{sec:empirical_bounds_tail}
\EmpiricalBoundsTail
In the following, we empirically assess the quality of the bounds reported by our tail estimator. 
For degree interval $D$, let $\NodeAvgCC{D}{} = S_{LCC}(D) / \NumNodes{D}{}$. 
From Theorem~\ref{thm:ndcc_tail_appendix}, Algorithm \algnameestimate\ outputs estimate $\widehat{NDCC}(D)$ such that:
\begin{equation}
	\label{eq:ndcc_tail_interval}
	\widehat{NDCC}(D) \in \left[ (1 - \varepsilon) \frac{S_{LCC}(D^-)}{\NumNodes{D^+}{}}, (1 + \varepsilon) \frac{S_{LCC}(D^+)}{\NumNodes{D^-}{}}   \right], 
\end{equation}
for any degree interval $D = [L, U)$ such that $L \ge \tau$, where $D^- = [(1 + \varepsilon)L, (1 - \varepsilon)U]$ and $D^+ = [(1 - \varepsilon)L, (1 + \varepsilon)U]$.

To empirically quantify such bounds, we compute absolute error $|\widehat{NDCC}(D) - \NodeAvgCC{D}{}|$ considering degree intervals of powers of two (i.e., $D = [2^i, 2^{i+1})$). We consider values of $\widehat{NDCC}(D)$ as the middle point with respect to the interval guarantees from Equation~\ref{eq:ndcc_tail_interval}.
We pick 20 intervals equally spaced, making sure that for each interval $D = [L, U)$ we have $L \ge \tau$ (i.e., when we employ our tail estimator), for values of $\tau$ computed by our algorithm \algname\ in the experiments (see main text for some representative values of degree threshold $\tau$).

Figure~\ref{fig:empirical_bounds_tail} shows results for $\varepsilon \in \{ 0.01, 0.05, 0.1\}$. While for $\varepsilon = 0.1$ we observe large absolute errors and high variance, lower values of $\varepsilon$ (i.e., tighter intervals $D^-, D^+$) allow for very accurate and precise stable bounds, across all the datasets and the considered intervals. 
This means that in real-world networks, the distribution of degree-binned clustering coefficient is stable when allowing small slack in the binned intervals (for high degrees), and our bounds in Theorem~\ref{thm:ndcc_tail_appendix} are empirically accurate.  

\TableOvercount

\subsection{One Pass Algorithm: Extra Counters}
\label{sec:extra_counters}

In this section, we investigate the number of extra counters needed to ensure unbiasedness of estimates for our one-pass algorithm. 
Let $N$ be the number of nodes for which we know the degree (i.e., $|S_t|$), and $N'$ be the number of nodes for which we know the number of local triangles (i.e., $|\hat{T}_t|$).
In our two-pass algorithm, we have $N \ge N'$, given that our second pass \algnamesecondpass\ is estimating triangles only incident to nodes for which we know the degree from the first pass \algnamefirstpass.
Conversely, our one-pass algorithm also counts triangles incident to nodes for which (currently) we are unaware of their degrees. Such nodes belong to edges in the auxiliary sample; counting such triangles allows for the unbiasedness (see Section~\ref{sec:one_pass_algo}). 
That is, for our one-pass algorithm, we \emph{may} have $N \le N'$. 
Note that the extra counters for triangles will be wasted, as we cannot provide an estimate of clustering coefficients without knowing degrees (or estimates of degrees) for such nodes. 

\TableDelta

Table~\ref{tab:overcounts} shows the fraction of extra counters stored by our one-pass algorithm. We compute the number $W$ of ``wasted'' counters as $W = \max(0, N' - N)$ for each dataset, and for different size of tail auxiliary sample $B_{A_t} = f \cdot m$ for values of $f \in \{0.005, 0.01, 0.015 \}$.
The other parameters are as explained in Section~\ref{sec:params_choice}.
We observe that, across all datasets, the wasted counters range from 0 to 2.24 times the number of counters effectively needed.

\subsection{Empirical Values: $\Delta_V(D)$, $\Delta_E(D)$, $\Triangles{D}{}$}
\label{sec:bound_deltas}
In the following, we inspect values of $\Delta_V(D), \Delta_E(D)$ and $\Triangles{D}{}$ in real-world datasets.
%Recall that, for a degree interval $D$, $\Delta_V(D)$ is the maximum number of triangles incident to a node having degree in $D$, and $\Delta_E(D)$ is the maximum number of triangles (among those incident to nodes in $V_D$) adjacent to an edge.

Table~\ref{tab:delta_ratio_perturbed} shows values which appear in our space bound, considering intervals of power of two (i.e., $D = [2^i, 2^{i+1})$) on \texttt{SK}, \texttt{LJ} and \texttt{ORK}. 
We compute slacked intervals $D^- = [L(1 + \varepsilon), U(1 - \varepsilon))$ and $D^+ = [L(1 - \varepsilon), U(1 + \varepsilon))$ for any $D = [L, U)$ reported in the first column of the table, considering $\varepsilon = 0.05$.
We observe that the ratio $\Delta_E(D^+) / \Triangles{D^-}{}$, which our edge-space depends on, is consistently very small across all datasets: its maximum is below $0.009$ for all the considered intervals and datasets.

\subsection{Choice of Parameters} 
\label{sec:params_choice}
\GridSearchParams
We explain the choice of parameters for our algorithm \algname. 
%Overall, \algname\ depends on several parameters: probabilities $p_h$ and $p_t$ of sampling head and tail nodes, memory budgets $B_{M_h}, B_{M_t}$ for head and tail main samples, and memory budgets $B_{A_h}, B_{A_t}$ for the auxiliary samples. 
%Clearly, the higher each of these parameters is, the better the estimates, as they all govern the space used. 

We perform a grid search to identify the best choice of parameters. 
We consider values of $p_h \in \{0.05, 0.1, 0.15, 0.2, 0.25, 0.3\}$ and $p_t \in \{ 0.025 , 0.05, 0.075 \} \cdot n/m$, and we constrain \algname\ to use exactly $10\% m$ edges, considering head-tail allocation $(B_{M_h} + B_{A_h}, B_{M_t} + B_{A_t}) = \{ (0.01, 0.09), (0.05, 0.05), (0.09, 0.01) \} \cdot m$. 
Consequently, we set main-auxiliary allocation as explained in Sect.~\ref{sec:params_selection}: solving the quadratic equation based on $p_h$ for head budgets $B_{M_h}$ and $B_{A_h}$, and splitting equally tail budgets $B_{M_t} = B_{A_t}$.

Given that these experiments do not involve scalability of our algorithm, we do not consider our biggest datasets (\texttt{TW} and \texttt{FR}). 
Figure~\ref{fig:grid_search} depicts results; $x$-axis shows values of $p_h$, and $y$-axis shows RHAS for distribution $\{ NDCC(D_i) \}_i$ and degree intervals of powers of two (i.e., $D = [2^i, 2^{i+1})$). For each row (dataset), a subplot represents a different value of $p_t$. 
We see that the choice $B_{M_h} + B_{A_h} = 0.09m, B_{M_t} + B_{A_t} = 0.01m$ is consistently the best and the most stable, while the opposite assignment $B_{M_h} + B_{A_h} = 0.01m, B_{M_t} + B_{A_t} = 0.09m$ is performing badly and exhibits high variance. (It is cropped from the figure for the sake of visualization.)
The reason is that, in real-world datasets, degree distribution is heavy-tailed: there are many more edges incident to low-degree nodes (as they are a lot) conversely to high-degree nodes (as they are few). Hence, giving significantly more budget to the head region is the right move.

\PrintFigure{NDCC}{1.5}
\PrintFigure{NDCC}{1.1}
\PrintFigure{WDCC}{2}
\PrintFigure{WDCC}{1.2}
\PrintFigure{WDCC}{1.1}

The best and most robust choice of probabilities for sampling nodes would be $p_h = 0.3$ and $p_t = 0.075 n/m$, i.e., the highest values. 
However, this implies to store node-degrees counters for $45\%n$, on average.
We thus opted for the more conservative choice $p_h = 0.2$ and $p_t = 0.05 n/m$, which provides similar accuracy while tracking degrees only for a $30\%$ fraction of the nodes. 
In any case, we note that our algorithm \algname\ is very robust for different choices of $p_h$ and $p_t$, with the head-tail split of $0.09m, 0.01m$.

\subsection{Comparison with baselines}
\label{sec:comparison_baselines}
We compare \algname\ with \triestalg\ and \wrsalg\ baselines. In the following, we provide a brief description of baselines and how we implemented their degree estimator. 
\begin{myitemize}
	\item \triestalg~\cite{stefani2017triest} works by sampling edges via (a unique) reservoir sampling of size equal to a given memory budget. Consider the sample $\mathcal{R}$ of edges in the reservoir at the end of the stream. We output the local triangle estimates $\hat{t}_u$ for each node $u \in \mathcal{R}$, together with the unbiased degree estimator $\hat{d}_u = \frac{m}{|\mathcal{R}|} | \{ \Edge{u}{v} \in \mathcal{R} \}|$. 
	Estimates for NDCC are then computed by averaging LCCs for each degree interval.
	Note that reservoir sampling is an edge-sampling scheme, and thus keep nodes with probability proportional to their degree. As shown in Section~\ref{sec:exp}, \triestalg\ fails to provide estimate for low-degree nodes, as they are unrepresented in the sample.
	\item \wrsalg~\cite{shin2017wrs,lee2020temporal} couples reservoir sampling with waiting room, which stores the most recent edges in the stream. As suggested in their publication, we set the waiting room size to be a $0.1$ fraction of the total budget. 
	Consider the sample $\mathcal{R} \cup \mathcal{W}$ of edges in the reservoir and waiting room at the end of the stream. 
	We output the local triangle estimates $\hat{t}_u$ for each node $u \in \mathcal{R} \cup \mathcal{W}$, together with the unbiased degree estimator $\hat{d}_u = \frac{m}{|\mathcal{R}|} | \{ \Edge{u}{v} \in \mathcal{R} \}| + |\{ \Edge{u}{v} \in \mathcal{W} \}|$. 
	Because of the bias on low-degree nodes, which are counted only within the waiting room, we cap the average of local clustering coefficients to one, for each degree interval. 
	Estimates for NDCC are then computed by averaging LCCs for each degree interval.
	As shown in Section~\ref{sec:exp}, \wrsalg\ fails to provide estimate for low-degree nodes, as they are both unrepresented in the reservoir sampling and they are not captured precisely within the waiting room, which is useful for counting global triangles rather than clustering coefficients.
\end{myitemize}

\subsection{Additional Results}
\label{sec:additional_results}
We report all the results missing from Section~\ref{sec:exp}. 
We display estimates of \algname\ for $\{NDCC(D_i)\}_i$ and $\{WDCC(D_i)\}_i$ distributions, considering degree intervals $D_i = [b^i, b^{i+1})$, spanning $b = \{ 2, 1.5, 1.2, 1.1\}$. 
Figures~\ref{fig:compact_NDCC_base1.5} and~\ref{fig:compact_NDCC_base1.1} depict results for NDCC for $b = \{1.5, 1.1\}$, while Figures~\ref{fig:compact_WDCC_base2},~~\ref{fig:compact_WDCC_base1.2},  and~\ref{fig:compact_WDCC_base1.1} show results for WDCC for $b = \{2, 1.2, 1.1\}$. Similarly to what was observed in Section~\ref{sec:exp}, we see that \algname\ is qualitatively good, as the estimated distributions of NDCC and WDCC are very close to the exact distribution. Moreover, our algorithm is quantitavely accurate given that RHAS distance is consistently well below 15\%.

\subsection{Comparison with \kpalg}
\label{sec:comparison_lcc}
\FigureECC

In the following, we provide a more detailed comparison between algorithms \kpalg\ and \algname\ when varying the number $K$ of iterations used by \kpalg. 
Remember that \kpalg\ runs $K$ parallel copies of a sparsified graph with sparsification rate equal to $1 / C$, where $C$ is the number of colors.
As in the experiments in the main text, we fix memory budget to $0.1m$ for both algorithms, and we run algorithm \kpalg\ for number $K$ of iterations in $[50, 100, 150, 200, 250, 300, 350, 400]$, as in~\cite{kutzkov2013streaming}. The number $C$ of colors is then $C = K / 0.1$. 

Clearly, the lower the iterations, the higher the sparsification rate. This implies that \kpalg\ retrieves more nodes (of lower degrees), and evaluation of LCCs is performed with respect to a bigger set $H$ (implying different performance for \algname).  

Figure~\ref{fig:ecc_comparison} shows results, for our smallest datasets \texttt{SK}, \texttt{LJ} and \texttt{ORK}. 
We observe that \algname\ is significantly better than \kpalg\ for each number $K > 50$ of iterations and for each metric on \texttt{SK} and \texttt{LJ}, and comparable for $K = 50$, while on \texttt{ORK} our algorithm is better for $K > 100$. 
For small $K$'s, \kpalg\ uses higher sparsification rate and retrieves nodes of lower degree, which may be harmful for \algname\ as they are closer to the degree threshold $\tau$. Remember that \algname\ has theoretical guarantees for LCCs only in the tail estimator.

Most importantly, \algname\ employs \emph{only} the \emph{tail} estimator to estimate LCCs, thus consuming an edge budget of $0.01m$ (as $B_{M_t} = B_{A_t} = 0.005m$), an order of magnitude less compared to \kpalg\ (using $0.1m$ in expectation).
Therefore, if the final goal is to estimate LCCs, \algname\ can be configured so as to ignore the head component and to fully dedicate its allowed space to the tail estimator.
To empirically quantify this claim, we compare \algname\ with the same configuration used for Table~\ref{tab:lcc_comparison} (henceforth, we denote such configuration as \whead) against \algname\ without the head component (henceforth: \wohead), i.e., setting $p_h = B_{M_h} = B_{A_h} = 0$. 
Since we compute the degree threshold $\tau$ as a function of the number of nodes retained in the head sample, and since our goal here is to compare runtime and memory footprint, we provide \wohead\ with the degree thresholds used by \whead.

Table~\ref{tab:head_ablation} reports the results.
The two configurations attain comparable accuracy on every dataset and for every metric considered, with all differences lying well within one standard deviation.\footnote{Results are not exactly equal due to the uniqueness of the random seed, as the number of samples drawn differs across the two configurations.}
The reduction in edge budget, in contrast, translates directly into a reduction of the memory footprint: \wohead\ requires from $6.1\times$ (\texttt{SK}) up to $8.8\times$ (\texttt{ORK}, \texttt{TW}) less memory than \whead, thus closely approaching the order-of-magnitude gap.
The running time improves as well, although by a smaller factor ($2.2$--$3.1\times$), since a constant fraction of the cost is due to the single pass over the stream, which is independent of the allotted budget.

We conclude that, when only LCC estimates are required, the head component can be safely disabled, yielding a considerably lighter and faster estimator at no cost in accuracy.

\TableHeadAblation

\end{document}